\documentclass[11pt]{article}

\usepackage[utf8]{inputenc}

\usepackage[ruled,vlined,linesnumbered,noend]{algorithm2e}
\SetKwInOut{Parameter}{Parameters}
\newcommand {\ignore} [1] {}
\usepackage{amsthm}
\usepackage{amsmath}
\usepackage{amssymb}
\usepackage{enumerate}
\usepackage{graphicx}
\usepackage[margin=1in]{geometry}
\usepackage{dsfont}
\usepackage{booktabs} 
\usepackage{tabularx}
\usepackage{multirow}
\usepackage[shortlabels]{enumitem}
\usepackage{tablefootnote}
\usepackage{tocloft}

\usepackage[dvipsnames]{xcolor}
\definecolor{ForestGreen}{rgb}{0.1333,0.5451,0.1333}
\definecolor{DarkRed}{rgb}{0.65,0,0}

\usepackage{hyperref}
\hypersetup{
    colorlinks=true,
    linkcolor=DarkRed,
    filecolor=magenta,      
    urlcolor=cyan,
    citecolor=ForestGreen,
    linktocpage=true,
    pagebackref=true,
    bookmarks=true,
    bookmarksopen=true,
    bookmarksnumbered=true
}

\usepackage{cleveref}

\usepackage{framed}
\usepackage[framemethod=tikz]{mdframed}
\usepackage{tikz-cd}
\usetikzlibrary{shapes.misc}
\tikzset{cross/.style={cross out, draw=black, fill=none, minimum size=2*(#1-\pgflinewidth), inner sep=0pt, outer sep=0pt}, cross/.default={10pt}}
\usetikzlibrary{arrows.meta, calc}

\newcommand{\miss}{\mathsf{miss}}
\newcommand{\clr}{\mathsf{clr}}

\newcommand{\lo}{\mathrm{lo}}

\newcommand{\Vizing}{\textnormal{\textsf{Vizing}}}

\newcommand{\VizingF}{\textnormal{\texttt{Vizing-Fan}}}
\newcommand{\VizingP}{\textnormal{\texttt{Vizing-Path}}}

\newcommand{\ColorUFans}{\textnormal{\textsf{Color-U-Fans}}}
\newcommand{\PrimeUFans}{\textnormal{\textsf{Prime-U-Fans}}}
\newcommand{\ActUFans}{\textnormal{\textsf{Activate-U-Fans}}}

\newcommand{\ConUFans}{\textnormal{\textsf{Construct-U-Fans}}}
\newcommand{\PruneVFans}{\textnormal{\textsf{Prune-Vizing-Fans}}}
\newcommand{\ReduceUEdges}{\textnormal{\textsf{Reduce-U-Edges}}}

\newcommand{\UpdatePath}{\textnormal{\textsf{Update-Path}}}
\newcommand{\cost}{\textnormal{\texttt{cost}}}

\newcommand{\ceil}[1]{\left\lceil #1 \right\rceil}
\newcommand{\floor}[1]{\left\lfloor #1 \right\rfloor}

\DeclareMathAlphabet{\mathmybb}{U}{bbold}{m}{n}
\newcommand{\1}{\mathmybb{1}}

\newtheorem{theorem}{Theorem}[section]
\newtheorem{lemma}[theorem]{Lemma}

\newtheorem{corollary}[theorem]{Corollary}
\newtheorem{proposition}[theorem]{Proposition}
\newtheorem*{proposition*}{Proposition}

\newtheorem{definition}[theorem]{Definition}

\newtheorem{invariant}[theorem]{Invariant}

\newtheorem{claim}[theorem]{Claim}

\DeclareMathOperator*{\poly}{poly}
\DeclareMathOperator*{\polylog}{polylog}

\DeclareMathOperator*{\U}{\mathcal U}
\DeclareMathOperator*{\E}{\mathcal E}
\DeclareMathOperator*{\f}{\boldsymbol{f}}
\DeclareMathOperator*{\e}{\boldsymbol{e}}
\DeclareMathOperator*{\g}{\boldsymbol{g}}
\DeclareMathOperator*{\F}{\mathbf F}
\DeclareMathOperator*{\s}{\boldsymbol{S}}

\newcommand{\expect}[1]{\mathbb E \left[#1\right]}
\newcommand{\paren}[1]{\left( #1 \right)}
\newcommand{\sparen}[1]{\left[ #1 \right]}

\allowdisplaybreaks

\title{Vizing's Theorem in Near-Linear Time}

\author{
Sepehr Assadi\thanks{University of Waterloo, \texttt{sepehr@assadi.info}}
\and
Soheil Behnezhad\thanks{Northeastern University, \texttt{s.behnezhad@northeastern.edu}}
\and
Sayan Bhattacharya\thanks{University of Warwick, \texttt{s.bhattacharya@warwick.ac.uk}}
\and 
Mart\'in Costa\thanks{University of Warwick, \texttt{martin.costa@warwick.ac.uk}}
\and 
Shay Solomon\thanks{Tel Aviv University, \texttt{solo.shay@gmail.com}} \and 
Tianyi Zhang\thanks{ETH Z\"urich, \texttt{tianyi.zhang@inf.ethz.ch} \smallskip}
}

\date{}

\begin{document}

\maketitle

 \pagenumbering{gobble}

\begin{abstract}
    Vizing's theorem states that any $n$-vertex $m$-edge graph of maximum degree $\Delta$ can be \emph{edge colored} using at most $\Delta + 1$ different colors [Vizing, 1964]. Vizing's original proof is algorithmic and shows that such an edge coloring can be found in $O(mn)$ time. This was subsequently improved to $\tilde O(m\sqrt{n})$ time, independently by [Arjomandi, 1982] and by [Gabow et al., 1985].\footnote{Throughout, we use $\tilde O(f) := O(f \polylog{(n)})$ to suppress log-factors in the number of vertices of the graph.}
    
    \medskip
    
    Very recently, independently and concurrently, using randomization, this runtime bound was further improved to  $\tilde{O}(n^2)$ by [Assadi, 2024] and $\tilde O(mn^{1/3})$ by [Bhattacharya, Carmon, Costa, Solomon and Zhang, 2024] (and subsequently to $\tilde O(mn^{1/4})$ by [Bhattacharya, Costa, Solomon and Zhang, 2024]). 

     \medskip
    In this paper, we present a randomized algorithm that computes a $(\Delta+1)$-edge coloring in near-linear time---in fact, only $O(m\log{\Delta})$ time---with high probability, \emph{giving a near-optimal algorithm for this fundamental problem}.
 \end{abstract}

\clearpage

\setlength{\cftbeforesecskip}{7pt}

\setcounter{tocdepth}{3}
\tableofcontents
\clearpage
\pagenumbering{arabic}

\clearpage
\setcounter{page}{1}

\section{Introduction}

Given a simple undirected graph $G = (V, E)$ on $n$ vertices and $m$ edges, as well as an integer $\kappa \in \mathbb{N}^+$, a $\kappa$-edge coloring $\chi : E \rightarrow \{1, 2, \ldots, \kappa\}$ of $G$ assigns a color $\chi(e)$ to each edge $e \in E$ so that any two adjacent edges receive distinct colors.
The minimum possible value of $\kappa$ for which a $\kappa$-edge coloring exists in $G$ is known as the \emph{edge chromatic number} of $G$. If $G$ has maximum vertex degree $\Delta$, any proper edge coloring would require at least $\Delta$ different colors. A classical theorem by Vizing shows that $\Delta+1$ colors are always sufficient~\cite{Vizing}. Moreover, it was proven by~\cite{holyer1981np} that it is NP-complete to distinguish whether the edge chromatic number of a given graph is  $\Delta$ or $\Delta+1$, 
and therefore $\Delta+1$ is the best bound we can hope for with polynomial time algorithms. 

Vizing's original proof easily extends to an $O(mn)$ time algorithm, which was improved to $\tilde{O}(m\sqrt{n})$ in the 1980s by \cite{arjomandi1982efficient} and \cite{gabow1985algorithms} independently.
More recently, the algorithms of~\cite{arjomandi1982efficient,gabow1985algorithms} were simplified in \cite{sinnamon2019fast}, while shaving off extra logarithmic factors from their runtime complexities, achieving a clean $O(m\sqrt{n})$ runtime bound. Very recently, this longstanding $O(m\sqrt{n})$ time barrier was bypassed in two concurrent works \cite{Assadi24} and \cite{BhattacharyaCCSZ24} which improved the runtime bound to
two incomparable bounds of $\tilde{O}(n^2)$ and $\tilde{O}(mn^{1/3})$, respectively. In a follow-up work, the $\tilde{O}(mn^{1/3})$ runtime bound of~\cite{BhattacharyaCCSZ24} to was further improved to $\tilde{O}(mn^{1/4})$ in \cite{BhattacharyaCSZ24}. 

In this work, we resolve the time complexity of randomized $(\Delta+1)$-edge coloring up to at most a log factor by presenting a near-linear time algorithm for this problem.

\begin{theorem}\label{thm:main}
    There is a randomized algorithm that, given any simple undirected graph $G = (V, E)$ on $n$ vertices and $m$ edges with maximum degree $\Delta$, finds a $(\Delta + 1)$-edge coloring of $G$ in $O(m\log{n})$ time with high probability.
\end{theorem}

\noindent
{\bf Remarks:} Several remarks on our~\Cref{thm:main} are in order: 
\vspace{-5pt}
\begin{itemize}[leftmargin=10pt,itemsep=1pt]
\item Our algorithm in~\Cref{thm:main} does not rely on any of the recent developments in~\cite{Assadi24,BhattacharyaCCSZ24,BhattacharyaCSZ24} 
and takes an entirely different path. We present an overview of our approach, as well as a comparison to these recent developments in~\Cref{sec:overview}. 

\item Our main contribution in this work is to improve the time-complexity of $(\Delta+1)$-edge coloring by {polynomial} factors all the way to near-linear. 
For this reason, as well as for the sake of transparency of our techniques, we focus primarily on presenting an $O(m\log^3{n})$ time randomized algorithm (\Cref{thm:main-tech}), which showcases our most novel ideas. Later in \Cref{sec:log(n)}, we show that a more careful implementation of the same algorithm achieves a clean $O(m \log n)$ runtime.

\item We can additionally use a result of~\cite{BernshteynD23} to further replace the $\log{n}$ term in~\Cref{thm:main} with a $\log{\Delta}$ term, leading to 
an algorithm for $(\Delta+1)$-edge coloring in $O(m\log{\Delta})$ time with high probability (\Cref{cor:small-Delta}). This matches the longstanding time bound for $\Delta$-edge coloring bipartite graphs \cite{combinatorica/ColeOS01} (which, to the best of our knowledge, is also the best known runtime for $(\Delta+1)$-edge coloring bipartite graphs). This bound is also related to a recent line of work in~\cite{BernshteynD23,bernshteyn2024linear,dhawan2024simple} that focused on the $\Delta=n^{o(1)}$ case 
and gave a randomized $(\Delta+1)$-coloring algorithm that runs in $O(m \Delta^4 \log \Delta)$ time with high probability~\cite{bernshteyn2024linear}. 

\item Vizing's theorem generalizes for (loop-less) multigraphs, asserting that any multigraph with edge multiplicity at most $\mu$ can be $(\Delta+\mu)$-edge colored~\cite{Vizing,vizing1965chromatic}. A related 
result is Shannon's theorem~\cite{shannon1949theorem} that asserts that any multigraph can be $\floor{3\Delta/2}$ edge colored independent of $\mu$; both these bounds are tight: see the so-called {\em Shannon multigraphs}~\cite{vizing1965chromatic}.
We show that our techniques extend to these theorems as well, giving $O(m\log{\Delta})$ time algorithms for both problems (see~\Cref{thm:multi Viz,thm:shannon} and~\Cref{cor:small-Delta}). 
\end{itemize}

\subsection{Related Work}

In addition to algorithms for Vizing's theorem, there has also been a long line of work on fast algorithms {for edge coloring} that use more than $\Delta + 1$ colors. It was first shown in \cite{karloff1987efficient} that an edge coloring can be computed in $\tilde{O}(m)$ time when we have $\Delta+\tilde{O}(\sqrt{\Delta})$ different colors. In addition, there are algorithms which run in linear or near-linear time for $(1+\epsilon)\Delta$-edge coloring~\cite{duan2019dynamic,BhattacharyaCPS24,elkin2024deterministic,bernshteyn2024linear,dhawan2024simple} when $\epsilon\in (0, 1)$ is a constant. Most recently, it was shown in \cite{Assadi24} that even a $(\Delta+O(\log n))$-edge coloring can be computed in $O(m\log{\Delta})$ expected time.

There are other studies on restricted graph classes. In bipartite graphs, a $\Delta$-edge coloring can be computed in $\tilde{O}(m)$ time~\cite{cole1982edge,combinatorica/ColeOS01,alon2003simple,goel2010perfect}. In bounded degree graphs, one can  compute a $(\Delta+1)$-edge coloring in $\tilde{O}(m \Delta)$ time \cite{gabow1985algorithms}, and it was generalized recently for
bounded arboricity graphs~\cite{BhattacharyaCPS24b}; see also \cite{BhattacharyaCPS24c,ChristiansenRV24,Kowalik24} for further recent results on edge coloring in bounded arboricity graphs.
Subfamilies of bounded arboricity graphs, including planar graphs, bounded tree-width graphs and bounded genus graphs, were studied in \cite{chrobak1989fast,chrobak1990improved,cole2008new}. 

Beside the literature on classical algorithms, considerable effort has been devoted to the study of edge coloring in various computational models in the past few years, including dynamic~\cite{BarenboimM17,BhattacharyaCHN18,duan2019dynamic,Christiansen23,BhattacharyaCPS24,Christiansen24}, online~\cite{CohenPW19,BhattacharyaGW21,SaberiW21,KulkarniLSST22,BilkstadSVW24,BlikstadOnline2025,dudeja2024randomizedgreedyonlineedge}, distributed~\cite{panconesi2001some,elkin20142delta,fischer2017deterministic,ghaffari2018deterministic,grebik2020measurable,balliu2022distributed,ChangHLPU20,Bernshteyn22,Christiansen23,Davies23}, and streaming~\cite{BehnezhadDHKS19,behnezhad2023streaming,chechik2023streaming,ghosh2023low} models, among others.

\subsection{Technical Overview}
\label{sec:overview}

We now present an overview of prior approaches to $(\Delta+1)$-edge coloring and describe our techniques at a high level. 
For the following discussions, we will assume basic familiarity with Vizing's proof and its underlying algorithm; see \Cref{sec:prelim:fans} for more details on Vizing's algorithm.

\medskip
\noindent \textbf{Prior Approaches:}
A generic approach for $(\Delta+1)$-edge coloring, dating back to Vizing's proof itself, is to {extend} a partial $(\Delta+1)$-edge coloring of the graph one edge at a time, possibly by recoloring some edges, until the entire graph 
becomes colored. As expected, the main bottleneck in the runtime of this approach comes from extending the coloring to the last few uncolored edges. 
For instance, given a graph $G = (V, E)$ with maximum degree $\Delta$, we can apply Eulerian partitions to divide $G$ into two edge-disjoint subgraphs $G = G_1 \cup G_2$ with maximum degrees at most  $\ceil{\Delta/2}$. 
We then find $(\ceil{\Delta/2}+1)$-edge colorings of the subgraphs $G_1$ and $G_2$ recursively. Directly combining the colorings of $G_1$ and $G_2$ gives a coloring of $G$ with $\Delta+3$ colors, so we have to uncolor a $2/(\Delta + 3)$ fraction of the edges---amounting to $O(m/\Delta)$ edges---and try to extend the current partial $(\Delta+1)$-edge coloring to these edges. Thus, the ``only'' remaining part is to figure out a way to color
these final $O(m/\Delta)$ edges, and let the above approach take care of the rest. 

This task of extending the coloring to the last $\Theta (m/\Delta)$ edges is the common runtime bottleneck of all previous algorithms. Vizing's original algorithm \cite{Vizing} gives a procedure to extend any partial coloring to an arbitrary uncolored edge $(u, v)$ by rotating some colors around $u$ and flipping the colors of an alternating path starting at $u$. The runtime of this procedure would be proportional to the size of the rotation, usually called a \emph{Vizing fan}, and the length of the alternating path, usually called a \emph{Vizing chain}, which are bounded by $\Delta$ and $n$ respectively. As such, the total runtime for coloring the remaining $O(m/\Delta)$ edges using Vizing fans and Vizing chains will be $O(mn/\Delta)$ time. 

As one can see, flipping long alternating paths is the major challenge in Vizing's approach. 
To improve this part of the runtime, \cite{gabow1985algorithms} designed an algorithm that groups all uncolored edges into $O(\Delta^2)$ types depending on the two colors of the Vizing chain induced by this edge. Since all the Vizing chains of the same type are vertex-disjoint, they can be flipped simultaneously and their total length is only $O(n)$. This means that the runtime of coloring all edges of a single type 
can be bounded by $O(n)$ as well. This leads to an $O(n\Delta^2)$ time algorithm for handling all $O(\Delta^2)$ types; a more careful analysis can bound this even with $O(m\Delta)$ time. 
Finally, balancing the two different bounds of $O(mn/\Delta)$ and $O(m\Delta)$ yields a runtime bound of $O(m\sqrt n)$ for coloring $O(m/\Delta)$ edges, which leads to an $\tilde O(m\sqrt{n})$ time algorithm using the above framework. 

There has been some very recent progress that broke through this classical barrier of $O(m\sqrt n)$ time in \cite{gabow1985algorithms,arjomandi1982efficient}. In \cite{BhattacharyaCCSZ24}, 
the authors speed up the extension of the coloring to uncolored edges when these edges admit a small vertex cover.
They then show how to precondition the problem so that uncolored edges admit a small vertex cover, leading to a $\tilde O(mn^{1/3})$ time algorithm. 
In \cite{Assadi24}, the author avoided the need for Eulerian partition and recursion altogether 
by instead designing a new near-linear time algorithm for $(\Delta + O(\log n))$-edge coloring. This algorithm borrows insights from sublinear matching algorithms in regular bipartite graphs by \cite{goel2010perfect} and is thus completely different from other edge coloring algorithms mentioned above. By using this algorithm, finding a $(\Delta+1)$-edge coloring directly reduces to the 
 color extension problem with $O((m\log n)/\Delta)$ uncolored edges (by removing the colors of a $\Theta(\log{n})/\Delta$ fraction of the edges in the $(\Delta+O(\log{n}))$-edge coloring to obtain a partial $(\Delta+1)$-edge coloring first). 
Applying Vizing's procedure for these uncolored edges takes additional $\tilde O(mn/\Delta) = \tilde O(n^2)$ time, leading to
an $\tilde O(n^2)$ time algorithm for $(\Delta+1)$-edge coloring. 
Finally, in \cite{BhattacharyaCSZ24}, the authors showed that a $(\Delta+1)$-coloring can be computed in $\tilde{O}(m n^{1/4})$ time, by using the algorithm of~\cite{Assadi24} for initial coloring of the graph and then presenting an improved color extension subroutine with a runtime of $\tilde{O}(\Delta^2 + \sqrt{\Delta n})$ for coloring each remaining edge; the best previous color extension time bounds were either the trivial $O(n)$ bound or the bound $\tilde{O}(\Delta^4)$ by \cite{Bernshteyn22,bernshteyn2024linear}. 

\vspace{-10pt}
\subsubsection*{Our Approach: A Near-Linear Time Color Extension Algorithm}
\vspace{-5pt}

We will no longer attempt to design a faster color extension for a {\em single edge}, and instead color them in large \emph{batches} like in \cite{gabow1985algorithms}, which allows for a much better amortization of runtime in coloring multiple edges.  
This ultimately leads to our main technical contribution: a new randomized algorithm for solving the aforementioned color extension problem for the last $O(m/\Delta)$ edges in $\tilde{O}(m)$ time. 
With this algorithm at hand, we can  follow the aforementioned Eulerian partition approach and obtain a $(\Delta+1)$-edge coloring algorithm whose runtime $T(m)$ follows the recursion $T(m) \leq 2 T(m/2) + \tilde{O}(m)$ with high probability; 
this implies that $T(m) = \tilde{O}(m)$, hence, giving us a near-linear time randomized algorithm for $(\Delta+1)$-edge coloring. This way, we will not even need to rely on the $(\Delta+O(\log{n}))$-edge coloring algorithm of~\cite{Assadi24} to color the earlier parts of the graph (although one
can use that algorithm instead of Eulerian partition approach to the same effect). 

We now discuss the main ideas behind our color extension algorithm. In the following, it helps to think of the input graph as being near-regular (meaning that the degree of each vertex is $\Theta(\Delta)$), 
and thus the total number of edges will be $m=\Theta(n\Delta)$; this assumption is \emph{not} needed for our algorithm and is only made here to simplify the exposition.

\medskip
\noindent \textbf{Color Type Reduction:} Recall that the runtime of $O(n\Delta^2)$ for the color extension algorithm of~\cite{gabow1985algorithms} 
is due to the fact that there are generally $O(\Delta^2)$ types of alternating paths in the graph and that the total length of the paths of each color type
is bounded by $O(n)$ edges. However, \emph{if} it so happens that the existing partial coloring only involves $O(\Delta)$ color types instead, 
then the same algorithm will only take $O(n\Delta) = O(m)$ time (by the near-regularity assumption). The underlying idea behind our algorithm 
is to modify the current partial edge coloring (without decreasing the number of uncolored edges) so that the number of color types reduces from $O(\Delta^2)$ to $O(\Delta)$ only. 

To explore this direction, let us \textbf{assume for now the input graph is \underline{bipartite}}, which greatly simplifies the structure of Vizing fans and Vizing chains,
thus allowing us to convey the key ideas more clearly (see \Cref{sec:bipartite} for a more detailed exposition); later we highlight some of the key challenges that arise when dealing with general graphs.
We shall note that it has been known since the 80s that one can $\Delta$-edge color bipartite graphs in $\tilde{O}(m)$ time~\cite{cole1982edge,combinatorica/ColeOS01}. However, the algorithms for bipartite graphs use techniques that are entirely different from Vizing fans and Vizing chains, and which do not involve solving the color extension problem at all.
In particular, prior to this work, it was unclear whether one can efficiently solve the color extension problem in bipartite graphs.
Therefore, the assumption of a bipartite input graph does not trivialize our goal of using Vizing fans and Vizing chains for efficiently solving the color extension problem.
Additionally, we can assume that in the color extension problem, the last $O(m/\Delta)$ edges to be colored can be partitioned into $O(1)$ matchings (this guarantee follows immediately from the recursive framework we outlined earlier), and that we deal with each of these matchings separately. In other words, we can also \textbf{assume that the uncolored edges are \underline{vertex-disjoint}}.

Let $\chi$ be a partial $(\Delta+1)$-edge coloring, and for any vertex $w\in V$, let $\miss_\chi(w)\subseteq [\Delta+1]$ be the set of colors missing from the edges incident to $w$ under $\chi$. 
Given any uncolored edge $(u, v)$, the color type of this edge $(u, v)$ would be 
$\{c_u, c_v\}$ for some arbitrary choices of $c_u\in \miss_\chi(u)$ and $c_v\in \miss_\chi(v)$ (it is possible for an edge to be able to choose more than one color type, but we fix one arbitrary choice among them); in other words, if we flip  the $\{c_u, c_v\}$-alternating path starting at $u$, then we can assign $\chi(u, v)$ to be $c_v$. To reduce the total number of different color types to $O(\Delta)$, we would have to make some color types much more \emph{popular}: at the beginning, a type spans an $\Omega(1/\Delta^2)$ proportion of the uncolored edges but we would like to have a type spanning an $\Omega(1/\Delta)$ proportion. For this purpose, we fix an arbitrary color type $\{\alpha, \beta\}$, and want to modify $\chi$ to transform the type of an arbitrary uncolored edge $(u, v)$ from $\{c_u, c_v\}$ to $\{\alpha, \beta\}$ -- we call this \emph{popularizing} the edge $(u,v)$. 
To do this, we can simply flip the $\{\alpha, c_u\}$-alternating path $P_u$ starting at $u$ and the $\{\beta, c_v\}$-alternating path $P_v$ starting at $v$.

There are two technical issues regarding this path-flipping approach. Firstly, the alternating paths $P_u$ and $P_v$ could be very long, and require a long time for being flipped. More importantly, flipping $P_u$ and $P_v$ could possibly damage other $\{\alpha, \beta\}$-type (uncolored) edges that we popularized before. More specifically, say that we have popularized a set $\Phi$ of uncolored edges. When popularizing the next uncolored edge $(u, v)$, it could be the case that the $\{\alpha, c_u\}$-alternating path $P_u$ is ending at a vertex $u'$ for some edge $(u', v')\in \Phi$. If we flip the path $P_u$, then $(u', v')$ would no longer be of $\{\alpha, \beta\}$-type as $\alpha$ would not be missing at $u'$ anymore.  See \Cref{overview-flip} for an illustration. 

\begin{figure}[h]
	\centering
	\begin{tikzpicture}[thick,scale=0.8]
	\draw (-1, 0) node(1)[circle, draw, color=cyan, fill=black!50,
	inner sep=0pt, minimum width=10pt, label = $v$] {};
	
	\draw (1, 0) node(2)[circle, draw, color=teal, fill=black!50,
	inner sep=0pt, minimum width=10pt, label = $u$] {};
	
	\draw (-3, 0) node(3)[circle, draw, fill=black!50,
	inner sep=0pt, minimum width=6pt] {};
	
	\draw (-5, 0) node(4)[circle, draw, fill=black!50,
	inner sep=0pt, minimum width=6pt] {};
	
	\draw (-7, 0) node(5)[circle, draw, fill=black!50,
	inner sep=0pt, minimum width=6pt] {};
	
	\draw (3, 0) node(6)[circle, draw, fill=black!50,
	inner sep=0pt, minimum width=6pt] {};
	
	\draw (5, 0) node(7)[circle, draw, fill=black!50,
	inner sep=0pt, minimum width=6pt] {};
	
	\draw (7, 0) node(8)[circle, draw, fill=black!50,
	inner sep=0pt, minimum width=6pt] {};
	
	\draw (9, 0) node(9)[circle, draw, color=orange, fill=black!50,
	inner sep=0pt, minimum width=10pt, label = $u'$] {};
	\draw (11, 0) node(10)[circle, draw, color=red, fill=black!50,
	inner sep=0pt, minimum width=10pt, label = $v'$] {};

    \draw (9, -1) node[label={$c_{u'}=\alpha$}] {};
    \draw (11, -1) node[label={$c_{v'}=\beta$}] {};
    \draw (1, -1) node[label={$c_{u}=\gamma$}] {};
    \draw (-1, -1) node[label={$c_{v}=\lambda$}] {};
    \draw (10, 0) node[label={\tiny popular}] {};
	
	\draw [line width = 0.5mm, dashed] (1) to (2);
	\draw [line width = 0.5mm, color=red] (1) to node[above] {$\beta$} (3);
	\draw [line width = 0.5mm, color=cyan] (3) to node[above] {$\lambda$} (4);
	\draw [line width = 0.5mm, color=red] (4) to node[above] {$\beta$} (5);
	\draw [line width = 0.5mm, color=orange] (2) to node[above] {$\alpha$} (6);
	\draw [line width = 0.5mm, color=teal] (6) to node[above] {$\gamma$} (7);
	\draw [line width = 0.5mm, color=orange] (7) to node[above] {$\alpha$} (8);
	\draw [line width = 0.5mm, color=teal] (8) to node[above] {$\gamma$} (9);
	\draw [line width = 0.5mm, dashed] (9) to (10);
	
    \draw[->, >={Triangle}, thick, line width = 0.9mm] (2, -1) to (2, -2);
	
	\draw (-1, -3) node(11)[circle, draw, color=red, fill=black!50,
	inner sep=0pt, minimum width=10pt, label = $v$] {};
	
	\draw (1, -3) node(12)[circle, draw, color=orange, fill=black!50,
	inner sep=0pt, minimum width=10pt, label = $u$] {};
	
	\draw (-3, -3) node(13)[circle, draw, fill=black!50,
	inner sep=0pt, minimum width=6pt] {};
	
	\draw (-5, -3) node(14)[circle, draw, fill=black!50,
	inner sep=0pt, minimum width=6pt] {};
	
	\draw (-7, -3) node(15)[circle, draw, fill=black!50,
	inner sep=0pt, minimum width=6pt] {};
	
	\draw (3, -3) node(16)[circle, draw, fill=black!50,
	inner sep=0pt, minimum width=6pt] {};
	
	\draw (5, -3) node(17)[circle, draw, fill=black!50,
	inner sep=0pt, minimum width=6pt] {};
	
	\draw (7, -3) node(18)[circle, draw, fill=black!50,
	inner sep=0pt, minimum width=6pt] {};
	
	\draw (9, -3) node(19)[circle, draw, color=teal, fill=black!50,
	inner sep=0pt, minimum width=10pt, label = $u'$] {};
	\draw (11, -3) node(20)[circle, draw, color=red, fill=black!50,
	inner sep=0pt, minimum width=10pt, label = $v'$] {};

    \draw (9, -4) node[label={$c_{u'}=\gamma$}] {};
    \draw (11, -4) node[label={$c_{v'}=\beta$}] {};
    \draw (1, -4) node[label={$c_{u}=\alpha$}] {};
    \draw (-1, -4) node[label={$c_{v}=\beta$}] {};
    \draw (0, -3) node[label={\tiny popular}] {};
	
	\draw [line width = 0.5mm, dashed] (11) to (12);
	\draw [line width = 0.5mm, color=cyan] (11) to node[above] {$\lambda$} (13);
	\draw [line width = 0.5mm, color=red] (13) to node[above] {$\beta$} (14);
	\draw [line width = 0.5mm, color=cyan] (14) to node[above] {$\lambda$} (15);
	\draw [line width = 0.5mm, color=teal] (12) to node[above] {$\gamma$} (16);
	\draw [line width = 0.5mm, color=orange] (16) to node[above] {$\alpha$} (17);
	\draw [line width = 0.5mm, color=teal] (17) to node[above] {$\gamma$} (18);
	\draw [line width = 0.5mm, color=orange] (18) to node[above] {$\alpha$} (19);
	\draw [line width = 0.5mm, dashed] (19) to (20);
	
\end{tikzpicture}
	\caption{
 In this picture, we attempt to popularize edge $(u, v)$ by flipping the $\{\alpha, \gamma\}$-alternating path from $u$ and the $\{ \beta, \lambda\}$-alternating path from $v$. However, flipping the $\{\alpha, \gamma \}$-alternating path from $u$ makes a previously popular edge $(u', v')$ unpopular as $u'$ will not miss color $\alpha$ anymore.
 }\label{overview-flip}
\end{figure}
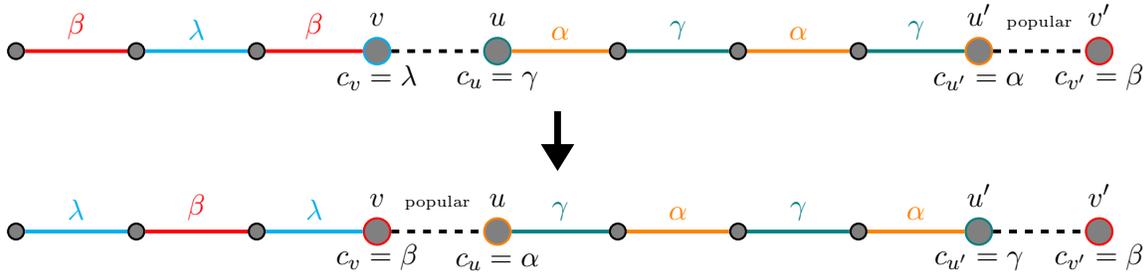

Our key observation is that when $|\Phi|$ is relatively small, most choices for an alternating path $P_u$ cannot be ending at edges in $\Phi$. Consider the above bad example where $P_u$ is a $\{c_u, \alpha\}$-alternating path 
ending at $u'$ for some $(u', v')\in \Phi$. Let us instead look at this from the perspective of the $\{\alpha,\beta\}$-type edge $(u',v') \in \Phi$. 
For any $(u', v')\in \Phi$ and any color $\gamma$, there can be at most one uncolored edge $(u, v)\notin \Phi$ whose corresponding path $P_u$ is the same $\{\gamma, \alpha\}$-alternating path starting at $u'$; 
this is because any vertex belongs to at most one alternating path of a fixed type (here, the type $\{\gamma,\alpha\}$) 
and each vertex belongs to at most one uncolored edge (recall that the uncolored edges form a matching).
This is 
also true for $\{\gamma,\beta\}$-type edges for the same exact reason. 
As such, ranging over all possible choices of $\gamma \in [\Delta+1]$, there are at most $O(|\Phi| \Delta)$ uncolored edges $(u, v)$ whose alternating paths could damage the set $\Phi$. Therefore, as long as the
size of $\Phi$ is a $o(1/\Delta)$ fraction (or more precisely, an $O(1/\Delta)$ fraction,
for a sufficiently small constant hiding in the $O$-notation) of all uncolored edges, the following property holds: a constant fraction of uncolored edges $e \notin \Phi$ can be popularized using the above method without damaging $\Phi$. See \Cref{overview-damage} for an illustration.

This property resolves both technical issues raised earlier simultaneously. Let $\lambda$ denote the number of uncolored edges. 
For the second issue, as long as $|{\Phi}| = o(\lambda/\Delta)$, we can take a random uncolored edge $(u, v)\notin \Phi$ and flip $P_u$ and $P_v$ if this does not 
damage any already popularized edge in $\Phi$; by the observation above, this happens with constant probability. For the first issue, we can
show that the expected length of alternating paths $P_u$ and $P_v$, for the random uncolored edge picked above, is $O(m/\lambda)$; 
indeed, this is because the number of edges colored $\alpha$ or $\beta$ is $O(n)$, hence the total length of all alternating paths with one color being fixed to $\alpha$ or $\beta$ is $O(m)$. 
All in all, the total time we spend to popularize a single type $\{\alpha,\beta\}$ to become a $\Theta(1/\Delta)$ fraction of all uncolored edges is $O(\lambda/\Delta \cdot m/\lambda) = O(m/\Delta)$. 
Since coloring edges of a single type takes $O(n) = O(m/\Delta)$ time by our earlier discussion,  we can color in this way a $\Theta(1/\Delta)$ fraction of all uncolored edges in $O(m/\Delta)$ expected time. As a direct corollary, we can color all uncolored edges in $O(m\log{n})$ time, hence solving the color extension problem, in (near-regular) bipartite graphs, in near-linear time. The above argument will be detailed in the proof of~\Cref{lem:main:bipartite}
in~\Cref{sec:bipartite} (see also \Cref{lem:color u-fans} of \Cref{sec:algo} for the analogous but not identical argument in general graphs).

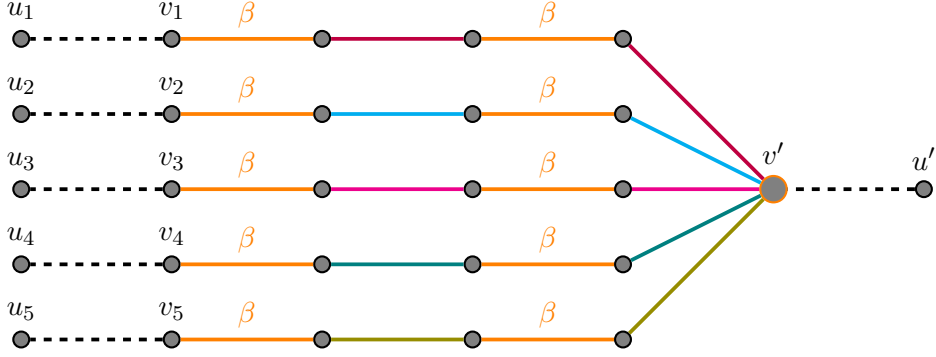
\begin{figure}
	\centering
	\begin{tikzpicture}[thick,scale=1]
	
	\draw (7, 0) node(1)[circle, draw, fill=black!50,
	inner sep=0pt, minimum width=6pt, label = $u'$] {};
	\draw (5, 0) node(2)[circle, draw, color=orange, fill=black!50,
	inner sep=0pt, minimum width=10pt, label = $v'$] {};	
	
	\draw (3, 2) node(3)[circle, draw, fill=black!50,
	inner sep=0pt, minimum width=6pt] {};
	\draw (1, 2) node(4)[circle, draw, fill=black!50,
	inner sep=0pt, minimum width=6pt] {};
	\draw (-1, 2) node(5)[circle, draw, fill=black!50,
	inner sep=0pt, minimum width=6pt] {};
	\draw (-3, 2) node(6)[circle, draw, fill=black!50, 
	inner sep=0pt, minimum width=6pt, label = $v_1$] {};
	\draw (-5, 2) node(7)[circle, draw, fill=black!50,
	inner sep=0pt, minimum width=6pt, label = $u_1$] {};

	\draw (3, 1) node(13)[circle, draw, fill=black!50,
	inner sep=0pt, minimum width=6pt] {};
	\draw (1, 1) node(14)[circle, draw, fill=black!50,
	inner sep=0pt, minimum width=6pt] {};
	\draw (-1, 1) node(15)[circle, draw, fill=black!50,
	inner sep=0pt, minimum width=6pt] {};
	\draw (-3, 1) node(16)[circle, draw, fill=black!50, 
	inner sep=0pt, minimum width=6pt, label = $v_2$] {};
	\draw (-5, 1) node(17)[circle, draw, fill=black!50,
	inner sep=0pt, minimum width=6pt, label = $u_2$] {};

	\draw (3, 0) node(23)[circle, draw, fill=black!50,
	inner sep=0pt, minimum width=6pt] {};
	\draw (1, 0) node(24)[circle, draw, fill=black!50,
	inner sep=0pt, minimum width=6pt] {};
	\draw (-1, 0) node(25)[circle, draw, fill=black!50,
	inner sep=0pt, minimum width=6pt] {};
	\draw (-3, 0) node(26)[circle, draw, fill=black!50, 
	inner sep=0pt, minimum width=6pt, label = $v_3$] {};
	\draw (-5, 0) node(27)[circle, draw, fill=black!50,
	inner sep=0pt, minimum width=6pt, label = $u_3$] {};

	\draw (3, -1) node(33)[circle, draw, fill=black!50,
	inner sep=0pt, minimum width=6pt] {};
	\draw (1, -1) node(34)[circle, draw, fill=black!50,
	inner sep=0pt, minimum width=6pt] {};
	\draw (-1, -1) node(35)[circle, draw, fill=black!50,
	inner sep=0pt, minimum width=6pt] {};
	\draw (-3, -1) node(36)[circle, draw, fill=black!50, 
	inner sep=0pt, minimum width=6pt, label = $v_4$] {};
	\draw (-5, -1) node(37)[circle, draw, fill=black!50,
	inner sep=0pt, minimum width=6pt, label = $u_4$] {};

	\draw (3, -2) node(43)[circle, draw, fill=black!50,
	inner sep=0pt, minimum width=6pt] {};
	\draw (1, -2) node(44)[circle, draw, fill=black!50,
	inner sep=0pt, minimum width=6pt] {};
	\draw (-1, -2) node(45)[circle, draw, fill=black!50,
	inner sep=0pt, minimum width=6pt] {};
	\draw (-3, -2) node(46)[circle, draw, fill=black!50, 
	inner sep=0pt, minimum width=6pt, label = $v_5$] {};
	\draw (-5, -2) node(47)[circle, draw, fill=black!50,
	inner sep=0pt, minimum width=6pt, label = $u_5$] {};

	\draw [line width = 0.5mm, dashed] (1) to (2);
	
	\draw [line width = 0.5mm, color=purple] (2) to (3);
	\draw [line width = 0.5mm, color=orange] (3) to node[above] {$\beta$} (4);
	\draw [line width = 0.5mm, color=purple] (4) to (5);
	\draw [line width = 0.5mm, color=orange] (5) to node[above] {$\beta$} (6);
	\draw [line width = 0.5mm, dashed] (6) to (7);
	
	\draw [line width = 0.5mm, color=cyan] (2) to (13);
	\draw [line width = 0.5mm, color=orange] (13) to node[above] {$\beta$} (14);
	\draw [line width = 0.5mm, color=cyan] (14) to (15);
	\draw [line width = 0.5mm, color=orange] (15) to node[above] {$\beta$} (16);
	\draw [line width = 0.5mm, dashed] (16) to (17);
	
	\draw [line width = 0.5mm, color=magenta] (2) to (23);
	\draw [line width = 0.5mm, color=orange] (23) to node[above] {$\beta$} (24);
	\draw [line width = 0.5mm, color=magenta] (24) to (25);
	\draw [line width = 0.5mm, color=orange] (25) to node[above] {$\beta$} (26);
	\draw [line width = 0.5mm, dashed] (26) to (27);
	
	\draw [line width = 0.5mm, color=teal] (2) to (33);
	\draw [line width = 0.5mm, color=orange] (33) to node[above] {$\beta$} (34);
	\draw [line width = 0.5mm, color=teal] (34) to (35);
	\draw [line width = 0.5mm, color=orange] (35) to node[above] {$\beta$} (36);
	\draw [line width = 0.5mm, dashed] (36) to (37);
	
	\draw [line width = 0.5mm, color=olive] (2) to (43);
	\draw [line width = 0.5mm, color=orange] (43) to node[above] {$\beta$} (44);
	\draw [line width = 0.5mm, color=olive] (44) to (45);
	\draw [line width = 0.5mm, color=orange] (45) to node[above] {$\beta$} (46);
	\draw [line width = 0.5mm, dashed] (46) to (47);
	
\end{tikzpicture}
	\caption{In this picture, $(u', v')\in \Phi$ with $c_{u'} = \alpha, c_{v'} = \beta$. For each uncolored edge $(u_i, v_i)$, flipping the $\{c_i, \beta\}$-alternating path from $v_i$ would damage the property that $c_{v'} = \beta$. Fortunately, there are at most $\Delta$ many different such $(u_i, v_i)$ because each of them is at the end of an $\{\beta, \cdot\}$-alternating path starting at $v'$.}\label{overview-damage}
\end{figure}

\medskip
\noindent \textbf{Collecting U-Fans in General Graphs:} 
We now discuss the generalization of our scheme above to non-bipartite graphs. Existence of odd cycles in non-bipartite graphs implies that 
we can no longer assign a color type $\{c_u, c_v\}$ to an uncolored edge $(u,v)$ for $c_u\in \miss_\chi(u), c_v\in \miss_\chi(v)$, and hope that flipping the $\{c_u,c_v\}$-alternating path
from $u$ allows us to color the edge $(u,v)$ with $c_v$ (because the path may end at $v$, and thus after flipping it $c_v$ will no longer be missing from $v$). 
This is where Vizing fans and Vizing chains come into play: in non-bipartite graphs, a color type of an uncolored edge $(u,v)$ is $\{c_u, \gamma_{u, v}\}$ where $c_u$ is an arbitrary color in $\miss_\chi(u)$ but $\gamma_{u, v}$ is determined by the Vizing fan around $u$ and the alternating path that we take for coloring $(u,v)$ (or the Vizing chain of $(u,v)$). Thus, while as before we can switch $c_u$ with some fixed color $\alpha$, it is unclear how to flip alternating paths to change $\gamma_{u, v}$ to $\beta$ also, in order to popularize the edge $(u,v)$ to be of some designated
type $\{\alpha,\beta\}$, without damaging another popularized edge as a result.

To address this challenge, we rely on the notion of a \emph{u-fan}, introduced by \cite{gabow1985algorithms}, which is the non-bipartite graph analogue of an uncolored edge in bipartite graphs. 
A u-fan of type $\{\alpha, \beta\}$ is a pair of uncolored edges $(u, v), (u, w)$ such that $\alpha\in \miss_\chi(u)$ and $\beta\in \miss_\chi(u)\cap \miss_\chi(v)$. 
Consider the $\{\alpha, \beta\}$-alternating path starting at $u$. As at least one of $v$ or $w$ is not the other endpoint of this alternating path (say $v$),  flipping this path allows us to assign the color $\beta$ to $(u,v)$. 
Consequently, as u-fans are similar to edges in bipartite graphs, we can still essentially (but not exactly) apply our color type reduction approach if all the uncolored edges are paired as u-fans. Therefore, it suffices to modify $\chi$ to pair all uncolored edges
together to form u-fans.

In order to pair different uncolored edges together and form $u$-fans, we should first be able to move uncolored edges around. Such operations already appeared in some previous work on dynamic edge coloring \cite{duan2019dynamic,Christiansen23,Christiansen24}. Basically, for any uncolored edge $(u, v)$, we can modify $\chi$ to shift this uncolored edge to any position on its Vizing chain. This naturally leads to a win-win situation: if the Vizing chain is  short, then $(u, v)$ can be colored efficiently using Vizing's procedure; otherwise if most Vizing chains are long, then there must be a pair of Vizing chains meeting together after a few steps, so we can shift two uncolored edges to form a u-fan efficiently.

\begin{figure}
	\centering
	\begin{tikzpicture}[thick,scale=0.48]
    \coordinate (shiftA) at (0, 0); 
    \coordinate (shiftB) at (20, 9); 

    
    \draw ($(shiftA) + (0, 3)$) node(1)[circle, draw, color=cyan, fill=black!50,
    inner sep=0pt, minimum width=10pt, label = $u_1$] {};
    
    \draw ($(shiftA) + (-2, 1)$) node(2)[circle, draw, fill=black!50,
    inner sep=0pt, minimum width=6pt, label = 180:{$v_1$}] {};
    
    \draw ($(shiftA) + (-1, 1)$) node(3)[circle, draw, fill=black!50,
    inner sep=0pt, minimum width=6pt] {};
    
    \draw ($(shiftA) + (0, 1)$) node(4)[circle, draw, fill=black!50,
    inner sep=0pt, minimum width=6pt] {};
    
    \draw ($(shiftA) + (1, 1)$) node(5)[circle, draw, fill=black!50,
    inner sep=0pt, minimum width=6pt] {};
    
    \draw ($(shiftA) + (3, 1)$) node(6)[circle, draw, fill=black!50,
    inner sep=0pt, minimum width=6pt] {};
    
    \draw ($(shiftA) + (5, 1)$) node(7)[circle, draw, fill=black!50,
    inner sep=0pt, minimum width=6pt] {};
    
    \draw ($(shiftA) + (7, 1)$) node(8)[circle, draw, fill=black!50,
    inner sep=0pt, minimum width=6pt, label=$w_1$] {};
    
    \draw ($(shiftA) + (9, 0)$) node(9)[circle, draw, fill=black!50,
    inner sep=0pt, minimum width=6pt, label=$x$] {};
    \draw ($(shiftA) + (11, 0)$) node(10)[circle, draw, fill=black!50,
    inner sep=0pt, minimum width=6pt, label=$y$] {};
    
    \draw ($(shiftA) + (0, -3)$) node(11)[circle, draw, color=cyan, fill=black!50,
    inner sep=0pt, minimum width=10pt, label = -90:{$u_2$}] {};
    
    \draw ($(shiftA) + (-2, -1)$) node(12)[circle, draw, fill=black!50,
    inner sep=0pt, minimum width=6pt, label = 180:{$v_2$}] {};
    
    \draw ($(shiftA) + (-1, -1)$) node(13)[circle, draw, fill=black!50,
    inner sep=0pt, minimum width=6pt] {};
    
    \draw ($(shiftA) + (0, -1)$) node(14)[circle, draw, fill=black!50,
    inner sep=0pt, minimum width=6pt] {};
    
    \draw ($(shiftA) + (1, -1)$) node(15)[circle, draw, fill=black!50,
    inner sep=0pt, minimum width=6pt] {};
    
    \draw ($(shiftA) + (3, -1)$) node(16)[circle, draw, fill=black!50,
    inner sep=0pt, minimum width=6pt] {};
    
    \draw ($(shiftA) + (5, -1)$) node(17)[circle, draw, fill=black!50,
    inner sep=0pt, minimum width=6pt] {};
    
    \draw ($(shiftA) + (7, -1)$) node(18)[circle, draw, fill=black!50,
    inner sep=0pt, minimum width=6pt, label=-90:{$w_2$}] {};
    
    \draw [line width = 0.5mm, dashed] (1) to (2);
    \draw [line width = 0.5mm, pink] (1) to (3);
    \draw [line width = 0.5mm, magenta] (1) to (4);
    \draw [line width = 0.5mm, purple] (1) to (5);
    \draw [line width = 0.5mm, cyan] (5) to node[above] {$\alpha$} (6);
    \draw [line width = 0.5mm, purple] (6) to (7);
    \draw [line width = 0.5mm, cyan] (7) to node[above] {$\alpha$} (8);
    \draw [line width = 0.5mm, purple] (8) to (9);
    \draw [line width = 0.5mm, cyan] (9) to node[above] {$\alpha$} (10);
    
    \draw [line width = 0.5mm, dashed] (11) to (12);
    \draw [line width = 0.5mm, lime] (11) to (13);
    \draw [line width = 0.5mm, teal] (11) to (14);
    \draw [line width = 0.5mm, olive] (11) to (15);
    \draw [line width = 0.5mm, cyan] (15) to node[above] {$\alpha$} (16);
    \draw [line width = 0.5mm, olive] (16) to (17);
    \draw [line width = 0.5mm, cyan] (17) to node[above] {$\alpha$} (18);
    \draw [line width = 0.5mm, olive] (18) to (9);
    
    \draw[->, >={Triangle}, thick, line width = 0.9mm] ($(13, 0)$) to ($(15, 0)$);
    
    
    \draw ($(shiftB) + (0, -6)$) node(21)[circle, draw, fill=black!50,
    inner sep=0pt, minimum width=6pt, label = $u_1$] {};
    
    \draw ($(shiftB) + (-2, -8)$) node(22)[circle, draw, fill=black!50,
    inner sep=0pt, minimum width=6pt, label = 180:{$v_1$}] {};
    
    \draw ($(shiftB) + (-1, -8)$) node(23)[circle, draw, fill=black!50,
    inner sep=0pt, minimum width=6pt] {};
    
    \draw ($(shiftB) + (0, -8)$) node(24)[circle, draw, fill=black!50,
    inner sep=0pt, minimum width=6pt] {};
    
    \draw ($(shiftB) + (1, -8)$) node(25)[circle, draw, fill=black!50,
    inner sep=0pt, minimum width=6pt] {};
    
    \draw ($(shiftB) + (3, -8)$) node(26)[circle, draw, fill=black!50,
    inner sep=0pt, minimum width=6pt] {};
    
    \draw ($(shiftB) + (5, -8)$) node(27)[circle, draw, fill=black!50,
    inner sep=0pt, minimum width=6pt] {};
    
    \draw ($(shiftB) + (7, -8)$) node(28)[circle, draw, color=cyan, fill=black!50,
    inner sep=0pt, minimum width=10pt, label=$w_1$] {};
    
    \draw ($(shiftB) + (9, -9)$) node(29)[circle, draw, fill=black!50,
    inner sep=0pt, minimum width=6pt, label = $x$] {};
    \draw ($(shiftB) + (11, -9)$) node(30)[circle, draw, fill=black!50,
    inner sep=0pt, minimum width=6pt, label = $y$] {};
    
    \draw ($(shiftB) + (0, -12)$) node(31)[circle, draw, fill=black!50,
    inner sep=0pt, minimum width=6pt, label = -90:{$u_2$}] {};
    
    \draw ($(shiftB) + (-2, -10)$) node(32)[circle, draw, fill=black!50,
    inner sep=0pt, minimum width=6pt, label = 180:{$v_2$}] {};
    
    \draw ($(shiftB) + (-1, -10)$) node(33)[circle, draw, fill=black!50,
    inner sep=0pt, minimum width=6pt] {};
    
    \draw ($(shiftB) + (0, -10)$) node(34)[circle, draw, fill=black!50,
    inner sep=0pt, minimum width=6pt] {};
    
    \draw ($(shiftB) + (1, -10)$) node(35)[circle, draw, fill=black!50,
    inner sep=0pt, minimum width=6pt] {};
    
    \draw ($(shiftB) + (3, -10)$) node(36)[circle, draw, fill=black!50,
    inner sep=0pt, minimum width=6pt] {};
    
    \draw ($(shiftB) + (5, -10)$) node(37)[circle, draw, fill=black!50,
    inner sep=0pt, minimum width=6pt] {};
    
    \draw ($(shiftB) + (7, -10)$) node(38)[circle, draw, color=cyan, fill=black!50,
    inner sep=0pt, minimum width=10pt, label=-90:{$w_2$}] {};
    
    \draw [line width = 0.5mm, pink] (21) to (22);
    \draw [line width = 0.5mm, magenta] (21) to (23);
    \draw [line width = 0.5mm, purple] (21) to (24);
    \draw [line width = 0.5mm, cyan] (21) to node[right] {$\alpha$} (25);
    \draw [line width = 0.5mm, purple] (25) to (26);
    \draw [line width = 0.5mm, cyan] (26) to node[above] {$\alpha$} (27);
    \draw [line width = 0.5mm, purple] (27) to (28);
    \draw [line width = 0.5mm, dashed] (28) to (29);
    \draw [line width = 0.5mm, cyan] (29) to node[above] {$\alpha$} (30);
    
    \draw [line width = 0.5mm, lime] (31) to (32);
    \draw [line width = 0.5mm, teal] (31) to (33);
    \draw [line width = 0.5mm, olive] (31) to (34);
    \draw [line width = 0.5mm, cyan] (31) to node[right] {$\alpha$} (35);
    \draw [line width = 0.5mm, olive] (35) to (36);
    \draw [line width = 0.5mm, cyan] (36) to node[above] {$\alpha$} (37);
    \draw [line width = 0.5mm, olive] (37) to (38);
    \draw [line width = 0.5mm, dashed] (38) to (29);
    
\end{tikzpicture}
	\caption{In this picture, we have two different uncolored edges $(u_1, v_1), (u_2, v_2)$ such that $\alpha\in \miss_{\chi}(u_1)\cap \miss_{\chi}(u_2)$, and their Vizing chains first intersect at edge $(x,y)$ which currently has color $\alpha$ under $\chi$. Then we can rotate their Vizing fans and flip part of their Vizing chains to shift $(u_1, v_1), (u_2, v_2)$ to $(w_1, x), (w_2, x)$ respectively to form a u-fan; note that $\alpha\in \miss_{\chi}(w_1)\cap \miss_{\chi}(w_2)$ after this shifting procedure.}\label{overview-pair}
\end{figure}
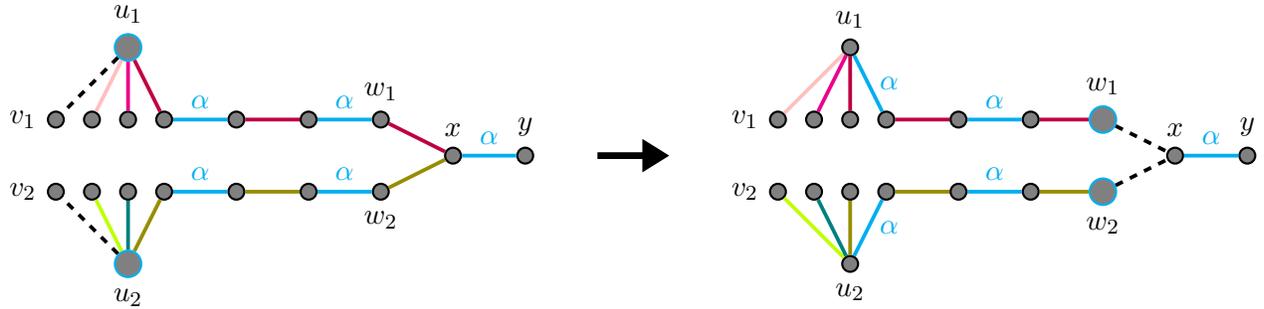

Let us make this a bit more precise. Fix an arbitrary color $\alpha$ and consider the set $U_\alpha$ of all the uncolored edges $(u, v)$ such that $\alpha\in \miss_\chi(u)$ and 
the respective Vizing chain is of type $\{\alpha,\cdot\}$.  Also, assume there are $m_\alpha$ edges colored $\alpha$ under $\chi$. If most $\{\alpha,\cdot\}$-Vizing chains have length larger than $\Omega(m_\alpha / |U_\alpha|)$, then on average, two Vizing chains should meet within the first $O(m_\alpha / |U_\alpha|)$ steps; in this case, we can repeatedly pick two intersecting Vizing chains and create a u-fan by shifting their initial uncolored edges 
to the intersection of these Vizing chains; see \Cref{overview-pair} for illustration. Given the length of the chains, this takes $O(m_\alpha / |U_\alpha|)$ time. Otherwise, the average cost of applying Vizing's color extension procedure is $O(m_\alpha / |U_\alpha|)$, and in this case we can directly color all those edges in $O(m_\alpha)$ time. Summing over all different $\alpha\in [\Delta+1]$ gives a near-linear runtime. The above argument will be detailed in the proof of \Cref{lem:build u-fans} in \Cref{sec:algo build u-fans}.

The above discussion leaves out various technical challenges. For instance, moving around uncolored edges as described above breaks the assumption that the uncolored edges form a matching. 
Handling this requires dedicating \emph{different} colors from $\miss_{\chi}(u)$ for every uncolored edge incident on a vertex $u$. This is formalized via the notion of {\em separability} in \Cref{sec:prelim:u}.
Additionally, we have ignored all algorithmic aspects of (efficiently) finding pairs of intersecting Vizing chains, as well as the corner cases of Vizing fan intersections and fan-chain intersections. 
We defer the discussions on these details to the actual proofs in subsequent sections.

\section{Basic Notation}
\label{sec:prelim:notation}

Let $G = (V, E)$ be graph on $n$ vertices with $m$ edges and maximum degree $\Delta$ and let $\chi : E \rightarrow [\Delta + 1] \cup \{ \bot\}$ be a (partial) $(\Delta + 1)$-edge coloring of $G$. We refer to edges $e \in E$ with $\chi(e) = \bot$ as \emph{uncolored}. Given a vertex $u \in V$, we denote the set of colors that are not assigned to any edge incident on $u$ by $\miss_\chi(u)$. We sometimes refer to $\miss_\chi(u)$ as the \emph{palette} of $u$. We say that the colors in $\miss_\chi(u)$ are \emph{missing} (or \emph{available}) at $u$. 

Given a path $P = e_1,\dots,e_k$ in $G$, we say that $P$ is an \emph{$\{\alpha, \beta\}$-alternating path} if $\chi(e_i)= \alpha$ whenever $i$ is odd and $\chi(e_i) = \beta$ whenever $i$ is even (or vice versa). We say that the alternating path $P$ is \emph{maximal} if one of the colors $\alpha$ or $\beta$ is missing at each of the endpoints of $P$. We refer to the process of changing the color of each edge $e_i \in P$ with color $\alpha$ (resp.~$\beta$) to $\beta$ (resp.~$\alpha$) as \emph{flipping} the path $P$. We denote by  $|P|$ the length (i.e., the number of edges) of the alternating path $P$. We define the \emph{length $i$ prefix} of the path $P$ to be the path $P_{\leq i} := e_1,\dots,e_i$.

Consider a set $U \subseteq E$ of edges that are uncolored under $\chi$, i.e., $\chi(e) = \bot$ for all $e \in U$. We use the phrase {\em ``extending $\chi$ to $U$''} to mean the following: Modify $\chi$ so as to ensure that $\chi(e)\neq \bot$ for all $e \in U$, without creating any new uncolored edges. When the set $U$ consists of a single edge $e$ (i.e., when $U = \{e\}$), we use the phrase {\em ``extending $\chi$ to the edge $e$''} instead of {\em ``extending $\chi$ to $U$''}.

Our algorithms will always work by modifying a partial coloring $\chi$; unless explicitly specified otherwise, every new concept we define (such as u-fans and separable collection in \Cref{sec:prelim:new}) will be defined with respect to this particular partial coloring $\chi$.

\section{Showcase: Our Algorithm Instantiated on Bipartite Graphs}
\label{sec:bipartite}

In this section, we instantiate  our algorithm on bipartite graphs to showcase some of our key insights, and outline a proof of \Cref{thm:main:bipartite} below. 

\begin{theorem}\label{thm:main:bipartite}
    There is a randomized algorithm that, given a bipartite graph $G = (V, E)$ with maximum degree $\Delta$, returns a $\Delta$-edge coloring of $G$ in $\tilde O(m)$ time with high probability.
\end{theorem}

As explained in \Cref{sec:overview}, focusing on bipartite graphs allows us to ignore the technical issues that arise while dealing with Vizing fans. At the same time, this does {\em not} trivialize the main conceptual ideas underpinning our algorithm. In particular, we prove \Cref{thm:main:bipartite} via \Cref{lem:main:bipartite} below, which gives a specific {\em color-extension} algorithm on bipartite graphs.  Although near-linear time $\Delta$-edge coloring algorithms on bipartite graphs existed since the 1980s~\cite{cole1982edge,combinatorica/ColeOS01}, to the best of our knowledge  there was no known algorithm for \Cref{lem:main:bipartite} prior to our work.

\begin{lemma}\label{lem:main:bipartite}
    Let $\chi : E \rightarrow [\Delta] \cup \{\bot\}$ be a partial $\Delta$-edge coloring in a bipartite graph $G = (V, E)$, and let $U \subseteq E$ be a \underline{matching} of size $\lambda$ such that every edge $e \in U$ is uncolored  in $\chi$. 
    Furthermore, suppose that we have access to an ``auxiliary data structure'', which allows us to detect in $\tilde{O}(1)$ time the two least common colors $\alpha, \beta \in [\Delta]$ in $\chi$.\footnote{Specifically,  for every $\gamma \in [\Delta] \setminus \{\alpha, \beta\}$ and $\gamma' \in \{\alpha, \beta\}$, we have $\left|\{ e \in E \mid \chi(e) = \gamma \right\}| \geq \left|\{ e \in E \mid \chi(e) = \gamma' \}\right|$.} 
    Then there is a randomized algorithm that can extend $\chi$ to $\Omega(\lambda/\Delta)$ many edges of $U$ in $\tilde O(m/\Delta)$ time with high probability.
\end{lemma}

Let us start by showing that~\Cref{lem:main:bipartite} implies~\Cref{thm:main:bipartite} easily. 

\begin{proof}[Proof of~\Cref{thm:main:bipartite}]
	We follow the strategy outlined in~\Cref{sec:overview}. Given a bipartite graph $G$, we first find an Eulerian partition of the
	graph to partition the edges of $G$ into two subgraphs of maximum degree $\ceil{\Delta/2}$ each, and color them recursively using different colors. 
	This leads to a $2 \cdot \ceil{\Delta/2} \leq (\Delta+2)$ edge coloring of $G$. We then form a partial $\Delta$ edge coloring $\chi$ by uncoloring the two color classes of this $(\Delta+2)$-edge coloring with the fewest edges assigned to them, 
	which leaves us with two edge-disjoint matchings $U_1$ and $U_2$ to color. This is our previously mentioned color extension problem. To solve this problem, 
	we apply~\Cref{lem:main:bipartite} to $U_1$ first to extend $\chi$ to $\Omega(1/\Delta)$ fraction of it, and keep applying this lemma to the remaining uncolored edges of $U_1$ until they are all colored. 
	We then move to $U_2$ in the same exact way and extend $\chi$ to its edges as well, obtaining a $\Delta$-coloring of the entire $G$ as a result. 
	
	The correctness of the algorithm follows immediately by induction. The runtime can also be analyzed as follows. When coloring $U_1$ (or $U_2$), each invocation of~\Cref{lem:main:bipartite} reduces the number of uncolored edges in $U_1$ (or $U_2$) 
	by a $(1-\Omega(1/\Delta))$ factor and thus we apply this lemma a total of $O(\Delta \cdot \log{n})$ time. Moreover, each application of~\Cref{lem:main:bipartite} takes $\tilde O(m/\Delta)$ time with high probability. Thus, 
	with high probability, it only takes $\tilde O(m)$ time to extend the coloring $\chi$ to $U_1$ and $U_2$ in the color extension problem. 
	Hence, the runtime of the algorithm, with high probability, follows the recurrence $T(m,\Delta) \leq 2T(m/2,\Delta/2) + \tilde O(m)$, 
	and thus itself is $\tilde O(m)$ time.

	Finally, note that with $O(m)$ preprocessing time, we can maintain access to the ``auxiliary data structure'' throughout the repeated invocations of \Cref{lem:main:bipartite} above: All we  need to do is to maintain a counter for each color $\gamma \in [\Delta]$, which keeps track of how many edges in $G$ are currently assigned the color $\gamma$ in $\chi$. We maintain these counters in a balanced search tree, and update the relevant counter whenever we change the color assignment of an edge in $G$.
\end{proof}

The rest of this section is dedicated to the proof of~\Cref{lem:main:bipartite}.

\subsection{Our Bipartite Color Extension Algorithm in~\Cref{lem:main:bipartite}}

At a high level, our algorithm for \Cref{lem:main:bipartite} consists of the following three steps.

\begin{enumerate}[leftmargin=*]
    \item Pick the two least common colors $\alpha, \beta \in [\Delta]$ in $\Delta$. This implies that there are at most $O(m/\Delta)$ edges in $G$ that are colored with $\alpha$ or $\beta$ in $\chi$.\label{step1}
    \item Modify the coloring $\chi$ so that $\Omega(\lambda/\Delta)$ of the edges $(u,v) \in U$ either receive a color under $\chi$, or have $\alpha \in \miss_\chi(u)$ and $\beta \in \miss_\chi(v)$. While implementing this step, we ensure that the total number of edges with colors $\alpha$ or $\beta$ remains at most $O(m/\Delta)$.\label{step2}
    \item Let $\Phi$ denote the set of edges $(u,v) \in U$ with $\alpha \in \miss_\chi(u)$ and $\beta \in \miss_\chi(v)$. We call these edges \emph{popular}. We extend  $\chi$ to a constant fraction of the edges in $\Phi$, by flipping a set of maximal $\{\alpha, \beta\}$-alternating paths.\label{step3}
\end{enumerate}

We now formalize the algorithm; the pseudocode is provided in \Cref{alg:bipartite}. As input, we are given a bipartite graph $G = (V, E)$, a partial $\Delta$-edge coloring $\chi$ of $G$, and a matching $U \subseteq E$ of uncolored edges of size $\lambda$.

The algorithm starts by fixing the two least common colors $\alpha, \beta \in [\Delta]$ in $\chi$. 
The main part is the {\bf while} loop in \Cref{line:while loop:bipartite}, which runs  in \emph{iterations}. In  each iteration of the {\bf while} loop, the algorithm samples an edge $e = (u,v)$ from $U$ u.a.r.~and \emph{attempts} to either (1) directly extend the coloring $\chi$ to $e$ (see \Cref{line:colorchange1}), which adds $e$ to a set $C \subseteq U$ of colored edges in $U$ or (2) modify $\chi$ so that $\alpha \in \miss_\chi(u)$ and $\beta \in \miss_\chi(v)$---we refer to this as making the edge $(u,v)$ \emph{popular}---, 
which adds $e$ to a set $\Phi \subseteq U$ of popular edges (see~\Cref{line:modify}). The attempt to modifying $\chi$ is done by essentially finding a maximal $\{c_u, \alpha\}$-alternating path $P_u$ starting at $u$ 
and a $\{c_v,\beta\}$-alternating path $P_v$ starting at $v$ for $c_u \in \miss_{\chi}(u)$ and $c_v \in \miss_{\chi}(v)$ (see~\Cref{line:swap} and~\Cref{line:path1,line:path2}). The modification itself 
is done only if $P_u$ and $P_v$ do not intersect any other popular edge already in $\Phi$. 

We say that the concerned iteration of the {\bf while} loop \texttt{FAILS} if it chooses an already colored edge in $C$ (\Cref{fail:line1}), 
or modifying the color leads to an already popular edge in $\Phi$ to no longer remain popular (\Cref{fail:line2}); otherwise, we say the iteration \emph{succeeds}. 
As stated earlier, the algorithm maintains a subset $\Phi \subseteq U$ of popular edges, and a subset of edges $C \subseteq U$ that got colored since the start of the {\bf while} loop. Initially, we have $C = \Phi = \varnothing$.  Thus, the quantity $|\Phi| + |C|$ denotes the number of successful iterations of the {\bf while} loop that have been performed so far. The algorithm performs iterations until $|\Phi| + |C| = \Omega(\lambda / \Delta)$, and then it proceeds to extend the coloring $\chi$ to at least a constant fraction of the edges in $\Phi$ by finding $\{\alpha,\beta\}$-alternating paths for edges in $\Phi$ that admits such paths (see the {\bf for} loop in \Cref{line:forloop}).

\subsection{Analysis of the Bipartite Color Extension Algorithm: Proof of \Cref{lem:main:bipartite}}

We start by summarizing a few key properties of \Cref{alg:bipartite}.

\begin{claim}\label{cl:small colors:bipartite}
    Throughout the {\bf while} loop in \Cref{alg:bipartite}, there are at most $O(m/\Delta)$ edges in $G$ that receive either the color $\alpha$ or the color $\beta$, under $\chi$.
\end{claim}

\begin{proof}
    We start with $O(m/\Delta)$ such edges in $G$ and each successful iteration of the {\bf while} loop increases the number of such edges by $O(1)$, and there are  $O(\lambda/\Delta) =  O(m/\Delta)$ such iterations.    
\end{proof}

\begin{algorithm}
    \SetAlgoLined
    \DontPrintSemicolon
    \SetKwRepeat{Do}{do}{while}
    \SetKwBlock{Loop}{repeat}{EndLoop}
    Let $\alpha, \beta \in [\Delta]$ be the two least common colors in $\chi$ \ \ // \texttt{We have $\alpha \neq \beta$} \label{line:1:bipartite}\;
    Initialize $\Phi \leftarrow \varnothing$, $C \leftarrow \varnothing$, and set $\lambda \leftarrow |U|$ \label{line:init}\;
    \While{$|\Phi| + |C| < \lambda / (10\Delta)$}{\label{line:while loop:bipartite}
        Sample an edge $e = (u, v) \sim U$ independently and u.a.r.\;
        \If{$(u, v) \in \Phi \cup C$}{
            The iteration \texttt{FAILS} \label{fail:line1} \;
            \textbf{go to \Cref{line:while loop:bipartite}} \label{line:type1}\;
        }
        Identify (arbitrarily) two colors $c_u \in \miss_\chi(u)$ and $c_v \in \miss_\chi(v)$ \label{line:pickcolors}\;
        \If{$c_u = c_v$}{\label{line:equal}
            Set $\chi(u,v) \leftarrow c_u$ and $C \leftarrow C \cup  \{ (u, v) \}$ \ \ \label{line:colorchange1}\;
            \textbf{go to \Cref{line:while loop:bipartite}} \label{line:type2}\;
        }
        \If{$c_u = \beta$ or $c_v = \alpha$}{\label{line:swap}
            Set $(u,v) \leftarrow (v,u)$ \ \ \ \ // \texttt{Now $c_u \neq c_v$, $c_u \neq \beta$, $c_v \neq \alpha$  (see \Cref{line:equal,line:swap})}\;
        }
        Let $P_u$ be the maximal $\{c_u, \alpha\}$-alternating path starting at $u$ ($P_u = \varnothing$ if $\alpha \in \miss_{\chi}(u)$) \label{line:path1}\;
        Let $P_v$ be the maximal $\{c_v, \beta\}$-alternating path starting at $v$ ($P_v = \varnothing$ if $\beta \in \miss_{\chi}(v)$) \label{line:path2}\;
        \If{either $P_u$ or $P_v$ ends at a node that is incident on some edge in $\Phi \setminus \{e\}$}
        {\label{line:fail:condition}
        The iteration \texttt{FAILS} \label{fail:line2} \;
        }
        \Else{
            Modify $\chi$ by flipping the alternating paths $P_u$ and $P_v$ \label{line:colorchange20}\;
            Set $\Phi \leftarrow \Phi \cup \{(u,v)\}$  \ \ \ // \texttt{Now $\alpha \in \miss_{\chi}(u)$ and $\beta \in \miss_{\chi}(v)$} \label{line:modify}\; 
        }
    }
    $\Phi' \leftarrow \Phi$ \label{line:forloop:start} \;
    \For{each edge $e = (u, v) \in \Phi'$}{\label{line:forloop}
    $\Phi' \leftarrow \Phi' \setminus \{e\}$\;
    W.l.o.g., suppose that $\alpha \in \miss_{\chi}(u)$ and $\beta \in \miss_{\chi}(v)$\;
    \If{there exists a color $c \in \{ \alpha, \beta \}$ such that $c \in \miss_{\chi}(u) \cap \miss_{\chi}(v)$}
    {
    $\chi(u, v) \leftarrow c$ \label{line:forloop:extend:1} \;
    }
    \Else{
    Let $P^{\star}_u$ be the maximal $\{\alpha, \beta\}$-alternating path starting at $u$ \label{line:block:path}\;
    (Since $G$ is bipartite, $\alpha \in \miss_{\chi}(u)$ and $\beta \in \miss_{\chi}(v)$,  $P^{\star}_u$ does {\em not} end at $v$) \label{line:forloop:good} \;
    Modify $\chi$ by flipping the alternating path $P^{\star}_u$, and set $\chi(u, v) \leftarrow \beta$ \label{line:forloop:extend:2}\;  
    \If{the path $P^{\star}_u$ ends at a node that is incident on some  edge in $e' \in \Phi' \setminus \{e \}$}
    {\label{line:block}
    $\Phi' \leftarrow \Phi' \setminus \{e'\}$  \label{line:block:1}\;
    }
  }  }\caption{\textsf{BipartiteExtension}$(G, \chi, U)$}
    \label{alg:bipartite}
\end{algorithm}

\begin{lemma}\label{lem:bipartite:0:main}
    Throughout the execution of the {\bf while} loop in \Cref{alg:bipartite}, the following conditions hold: (i) the set $C$ consists of all the edges in $U$ that are colored under $\chi$; (ii) for every edge $(u,v) \in \Phi$, we have $\alpha \in \miss_\chi(u)$ and $\beta \in \miss_\chi(v)$.
\end{lemma}

\begin{proof}
  Part (i) of the lemma follows from \Cref{line:init} and \Cref{line:colorchange1}. For part (ii), consider an edge $e = (u,v)$ that gets added to $\Phi$. This happens only after flipping the paths $P_u$ and $P_v$ in \Cref{line:colorchange20}. Just before we execute \Cref{line:colorchange20}, the following conditions hold:
    \begin{itemize}
     \item $\alpha \neq \beta$ (see \Cref{line:1:bipartite}).
    \item $c_u \in \miss_{\chi}(u)$ and $c_v \in \miss_{\chi}(v)$ (see \Cref{line:pickcolors}).
    \item $c_u \neq c_v$, $c_u \neq \beta$ and $c_v \neq \alpha$ (see \Cref{line:swap}). 
    \item If $\alpha \in \miss_{\chi}(u)$ then $P_u = \varnothing$ (see \Cref{line:path1}), and if $\beta \in \miss_{\chi}(v)$ then $P_v = \varnothing$ (see \Cref{line:path2}). 
    \item The path $P_u$ (resp.~$P_v$) does {\em not} end at that a vertex that is incident on some edge in $\Phi \setminus \{e \}$ (see \Cref{line:fail:condition,fail:line2}), although it might possibly end at $v$ (resp.~$u$).
    \end{itemize}
    From these conditions, it follows that the paths $P_u$ and $P_v$ are edge-disjoint, and after we flip them in \Cref{line:colorchange20}, we have 
   $\alpha \in \miss_\chi(u)$ and $\beta \in \miss_\chi(v)$ in \Cref{line:modify}. 
   
   In subsequent iterations of the {\bf while} loop, the only places where we change the coloring $\chi$ are \Cref{line:colorchange1,line:colorchange20}. Since the edges in $U$ form a matching, changing the coloring in \Cref{line:colorchange1} cannot affect whether or not the edge $(u,v) \in \Phi$ remains popular (i.e.,~has $\alpha \in \miss_\chi(u)$ and $\beta \in \miss_\chi(v)$). Finally, during a subsequent iteration of the {\bf while} loop where we sample an edge $(u', v') \sim U$, we  flip the paths $P_{u'}, P_{v'}$ in \Cref{line:colorchange20} only if their endpoints are {\em not} incident on any edges in $\Phi \setminus \{ (u',v') \}$ (see \Cref{line:fail:condition}), and in particular, on $u$ or $v$. Thus, this operation cannot change what colors are available at $u$ and $v$, and so cannot change whether or not the edge  $(u, v) \in \Phi$ remains popular.
\end{proof}

We use the following lemma to bound the number of iterations of the \textup{\textbf{while}} loop in \Cref{alg:bipartite}. 

\begin{lemma}\label{lem:bipartite:1:main}
    Each iteration of the \textup{\textbf{while}} loop in \Cref{alg:bipartite} increases the value of $|\Phi|+|C|$ by an additive one, with probability at least $1/2$ and otherwise keep it unchanged. 
\end{lemma}

\begin{proof}
    Fix any given iteration of the {\bf while} loop. At the start of this iteration, we sample an edge from $U$ u.a.r. We say that an edge $e \in U$ is {\em bad} if the iteration \texttt{FAILS} when we sample $e$ (see \Cref{fail:line1} and \Cref{fail:line2}), and {\em good} otherwise. Note that if we sample a good edge $e \in U$, then the  iteration either adds one edge to the set $\Phi$ (see \Cref{line:modify}), or adds one edge to the set $C$ (see \Cref{line:colorchange1}). In other words, if we sample a good (resp.~bad) edge $e \in U$ at the start of the iteration, then this increases in the value of $|\Phi| + |C|$ by one (resp.~keeps the value of $|\Phi| + |C|$ unchanged). We will show that at most $\lambda/2$ edges in $U$ are bad. Since $|U| = \lambda$, this will imply the lemma. 
    
    To see why this claimed upper bound on the number of bad edges holds, first note that there are  $(|\Phi| + |C|)$ many bad edges that cause the iteration to \texttt{FAIL} in \Cref{fail:line1}. It now remains to bound the number of bad edges which cause the iteration to \texttt{FAIL} in \Cref{fail:line2}.

    Towards this end, note that for each edge  $(u', v') \in \Phi$, there are at most $4\Delta$ many maximal $\{\alpha, \cdot\}$- or $\{\beta, \cdot\}$-alternating paths that end at either $u'$ or $v'$. Furthermore, each such alternating path has its other endpoint incident on at most one edge in $U$ since the edges in $U$ form a matching. Thus, for each edge $(u', v') \in \Phi$, there are at most $4\Delta$ many edges $f_{(u', v')} \in U$  that satisfy the following condition: Some alternating path constructed by the algorithm after sampling $f_{(u', v')}$  ends at either $u'$ or $v'$ (see \Cref{line:path1} and \Cref{line:path2}). Each such edge $f_{(u', v')}$ is a bad edge which causes the iteration to \texttt{FAIL} in \Cref{fail:line2}, and moreover, only such edges are causing the iteration to \texttt{FAIL} in \Cref{fail:line2}. Thus, the number of such bad edges is at most $|\Phi| \cdot 4 \Delta$.

    To summarize, the total number of bad edges is at most $(|\Phi| + |C|) + |\Phi| \cdot 4 \Delta < \lambda/2$, where the last inequality holds since $|\Phi| + |C| < \lambda/(10\Delta)$ (see \Cref{line:while loop:bipartite}). This concludes the proof.
\end{proof}

Similarly, we can bound the expected runtime of each iteration of the \textup{\textbf{while}} loop in \Cref{alg:bipartite}. 

\begin{lemma}\label{lem:bipartite:2:main}
    Each iteration of the \textup{\textbf{while}} loop in \Cref{alg:bipartite} takes $\tilde O(m / \lambda)$ time in expectation, regardless of the outcome of previous iterations. 
\end{lemma}

\begin{proof}
   Alternating path flips can be done in time proportional to the path lengths using standard data structures, so we only need to analyze the path lengths. Fix any given iteration of the {\bf while} loop. At the start of this iteration, we can classify the edges in $U$ into one of the following three categories: An edge $e \in U$ is of ``Type I'' if the iteration ends at \Cref{line:type1} when we sample $e$, is of ``Type II'' if the iteration ends at \Cref{line:type2} when we sample $e$, and is of ``Type III'' otherwise. Let $\lambda_1, \lambda_2$ and $\lambda_3$ respectively denote the total number of Type I, Type II and Type III edges, with $\lambda_1 + \lambda_2 + \lambda_3 = \lambda$. For every Type III edge $e = (u, v) \in U$, we refer to the alternating paths $P_u$ and $P_v$ (see \Cref{line:path1} and \Cref{line:path2}) as the ``characteristic alternating paths'' for $e$. Let $\mathcal{P}_3$ denote the collection of characteristic alternating paths of all Type III edges. Since the set of Type III edges is a subset of $U$, they form a matching, and hence different paths in $\mathcal{P}_3$ have different starting points. Furthermore, every path in $\mathcal{P}_3$ is either a maximal $\{\alpha, \cdot\}$-alternating path or a maximal $\{\beta, \cdot\}$-alternating path. Accordingly, \Cref{cl:small colors:bipartite} implies that the total length of all the paths in $\mathcal{P}_3$ is at most $O((m/\Delta) \cdot \Delta) = O(m)$.

  Now, if at the start of the iteration, we happen to sample either a Type I or a Type II edge $e \in U$, then  the concerned iteration takes $O(1)$ time. In the paragraph below, we condition on the event that the edge $e = (u, v) \in U$ sampled at the start of the iteration is of Type III. 

  Using appropriate data structures,  the time taken to implement the concerned iteration is proportional (up to $\tilde O(1)$ factors) to the lengths of the  alternating paths $P_u$ and $P_v$ (see \Cref{line:path1} and \Cref{line:path2}).  The key observation is that for each $x \in \{u, v\}$,  the path $P_x$ is sampled almost uniformly (i.e., with probability $\Theta(1/\lambda_3)$) from the collection $\mathcal{P}_3$. Since the total length of all the paths in $\mathcal{P}_3$ is $O(m)$, it follows that the expected length of each of the paths $P_u, P_v$ is $O(m/\lambda_3)$.

  To summarize, with probability $\lambda_3/\lambda$, we sample a Type III edge at the start of the concerned iteration of the {\bf while} loop, and conditioned on this event the expected time spent on that iteration is $\tilde{O}(m/\lambda_3)$. In contrast, if we sample a Type I or a Type II edge at the start of the concerned iteration, then the time spent on that iteration is $O(1)$. This implies that we spend at most $\tilde{O}(m/\lambda_3) \cdot (\lambda_3/\lambda) + O(1) = \tilde{O}(m/\lambda)$ expected time per iteration of the {\bf while} loop.
\end{proof}

Finally, we show that in the very last step of the algorithm, the {\bf for} loop in \Cref{line:forloop}, the algorithm succeeds in coloring a constant fraction of popular edges.  

\begin{lemma}
\label{lm:forloop:bipartite}
The {\bf for} loop in \Cref{line:forloop} extends the coloring $\chi$ to at least half of the edges in $\Phi$.
\end{lemma}

\begin{proof}
Consider any given iteration of the {\bf for} loop where we pick an edge $e = (u, v) \in \Phi'$ in \Cref{line:forloop},
where w.l.o.g.\ $\alpha \in \miss_{\chi}(u)$ and $\beta \in \miss_{\chi}(v)$. It is easy to verify that during this iteration, we successfully extend the coloring $\chi$ to $e$, either in \Cref{line:forloop:extend:1} or in \Cref{line:forloop:extend:2}. In the latter case, we crucially rely on the fact that the graph $G$ is bipartite (see \Cref{line:forloop:good}), and hence the 
maximal $\{\alpha, \beta\}$-alternating path $P^{\star}_u$ starting at $u$ 
cannot end at $v$; in fact, this is the only place where we rely on the biparteness of $G$. \Cref{line:block,line:block:1} ensure that the following invariant is satisfied: For every edge $e' = (u', v') \in \Phi'$, we have $\alpha \in \miss_{\chi}(u')$ and $\beta \in \miss_{\chi}(v')$, i.e., the edge $e'$ is popular; indeed, \Cref{lem:bipartite:0:main} implies that this invariant holds just before the {\bf for} loop starts (see  \Cref{line:forloop:start}), and any edge $e'  \in \Phi'$ that may violate this invariant at a later stage, which may only occur due to flipping an alternating path that ends at a node  incident on $e'$ (in \Cref{line:block}), is removed from $\Phi'$ (in \Cref{line:block:1}). 

Now, the lemma follows from the observation that each time we successfully extend the coloring to one edge $e$ in $\Phi'$, we remove at most one other edge $e' \neq e$ from $\Phi'$ (see \Cref{line:block:1}), due to the vertex-disjointness of the edges in $U \supseteq \Phi \supseteq \Phi'$.
\end{proof}

\noindent We are now ready to conclude the running time analysis of \Cref{alg:bipartite} and establish
the required lower bound on
the number of newly colored edges in $U$ under $\chi$ by this algorithm. 

\begin{lemma}\label{lem:bipartite-before-last}
\Cref{alg:bipartite} takes $\tilde O(m/\Delta)$ time in expectation and extends $\chi$ to $\Omega(\lambda/\Delta)$ new edges.
\end{lemma}
\begin{proof}
We start with the runtime analysis: 
\begin{itemize}[leftmargin=10pt,itemsep=1pt]
\item \Cref{line:1:bipartite} can be implemented in $\tilde{O}(1)$ time using the auxiliary data structure, and \Cref{line:init,line:forloop:start} take constant time. 

\item Next, we bound the running time of the {\bf while} loop 
(\Cref{line:while loop:bipartite}). For any integer $k \geq 0$, let $T(k)$ denote the expected runtime of {\bf while} loop if we start the loop under the condition that $|\Phi| + |C| = k$. 
We are interested in $T(0)$ and we know that $T(\lambda/10\Delta) = O(1)$ by the termination condition of the loop. By~\Cref{lem:bipartite:1:main} and~\Cref{lem:bipartite:2:main}, for any $0 < k < \lambda/10\Delta$, we have, 
\[
	T(k) \leq \tilde O(m/\lambda) + \frac{1}{2} \cdot T(k) + \frac{1}{2} \cdot T(k+1), 
\]
where we additionally used the monotonicity of $T(\cdot)$, as well as the fact that each while-loop of \Cref{alg:bipartite} has expected runtime $\tilde{O}(m/\lambda)$ regardless of previous iterations, according to \Cref{lem:bipartite:2:main}. Thus, $T(k) \leq T(k+1) + \tilde O(m/\lambda)$ and hence $T(0) \leq \lambda/10\Delta \cdot \tilde O(m/\lambda) = \tilde O(m/\Delta)$. 

\item Finally, since the total number of edges with colors $\alpha$ and $\beta$ just before \Cref{line:forloop} is $O(m/\Delta)$ (see \Cref{cl:small colors:bipartite}), the {\bf for} loop can be implemented in $\tilde{O}(m/\Delta)$ time deterministically in a straightforward manner. 
\end{itemize}
Thus, the total runtime is $\tilde O(m/\Delta)$ in expectation. 

We now establish the bound on the number of newly colored edges. When the {\bf while} loop terminates, we have $|\Phi| + |C| \geq \lambda/(10\Delta)$ (see \Cref{line:while loop:bipartite}), and all the edges in $C$ are colored under $\chi$ (see \Cref{lem:bipartite:0:main}). Next, by \Cref{lm:forloop:bipartite}, the {\bf for} loop in \Cref{line:forloop} further extends the coloring $\chi$ to a constant fraction of the edges in $\Phi$, by only flipping $\{\alpha, \beta\}$-alternating paths. 
Consequently, we get at least $\Omega(\lambda/\Delta)$ newly colored edges in $U$ under $\chi$. This concludes the proof. 
\end{proof}

We can now conclude the proof of~\Cref{lem:main:bipartite}. To achieve the algorithm in this lemma, we simply run~\Cref{alg:bipartite} in parallel $\Theta(\log{n})$ times and use the coloring of whichever one finishes first (and terminate the rest at that point). 
This ensures the high probability guarantee of~\Cref{lem:main:bipartite} still in $\tilde O(m/\Delta)$ runtime. This concludes the entire proof.

\subsection{Extension to General Graphs: Roadmap for the Rest of the Paper}

In our~\Cref{lem:main:bipartite}, we crucially need the graph $G$ to be bipartite while executing \Cref{line:block:path,line:forloop:good,line:forloop:extend:2} in \Cref{alg:bipartite}. Otherwise, if $G$ contains odd cycles, then the maximal $\{\alpha, \beta\}$-alternating path $P^{\star}_u$ starting from $u$ can end at $v$. In that case, the color $\beta$ will no longer be missing at $v$ once we flip the path $P^{\star}_v$, and so we will not be able to extend the coloring $\chi$ to the edge $(u, v)$ via $\chi(u, v) \leftarrow \beta$. 
We shall emphasize that this is not a minor technicality, but rather the key reason general graphs are not necessarily $\Delta$ edge colorable and rather require $(\Delta+1)$ colors. 

The standard machinery to deal with this issue is the Vizing fan (see \Cref{sec:prelim}). However, if we try to use Vizing fans inside the framework of \Cref{alg:bipartite} in a naive manner,  then we lose control over one of the colors in the alternating path being flipped while extending the coloring to an edge, leading to a weaker averaging argument and a running time of $\tilde O(\Delta m)$ instead of $\tilde O(m)$.

To address this bottleneck, one of our key conceptual contributions is to focus on Vizing fans with respect to an object called a {\em separable collection of u-components} (see \Cref{sec:prelim:new}). Using this concept, in \Cref{sec:main:algo} we present our algorithmic framework in general graphs. Our main result (see \Cref{thm:main-tech}) relies upon two fundamental subroutines. The second subroutine (see \Cref{lem:color u-fans})  generalizes \Cref{alg:bipartite} presented in this section. In contrast, the first subroutine (see \Cref{lem:build u-fans}) either efficiently extends the current coloring to a constant fraction of the uncolored edges,  or changes the colors of some edges in the input graph so as to create a situation whereby we can invoke \Cref{lem:color u-fans}. We devote \Cref{sec:algo build u-fans,sec:algo} towards proving \Cref{lem:build u-fans,lem:color u-fans},  respectively. For clarity of presentation, we defer the details on supporting data structures to \Cref{sec:data structs}, such as the use of hash tables or binary search tree data structures for basic operations like picking a missing color at a vertex. Finally,~\Cref{sec:log(n)} shows how to optimize our algorithmic framework further and achieve an $O(m\log{n})$ time with high probability.

\section{Preliminaries: Vizing Fans and Vizing Chains}
\label{sec:prelim}
\label{sec:prelim:fans}

We now define the notion of \emph{Vizing fans}, which has been used extensively in the edge coloring literature \cite{Vizing, gabow1985algorithms, sinnamon2019fast}.

\begin{definition}[Vizing fan]
A \emph{Vizing fan} is a sequence $\F = (u, \alpha),(v_1,c_1),\dots,(v_k,c_k)$ where $u,v_1,\dots,v_k$ are distinct vertices and $c_1,\dots,c_k$ are colors such that
\begin{enumerate}
    \item $\alpha \in \miss_\chi(u)$ and $c_i \in \miss_\chi(v_i)$ for all $i \in [k]$.
    \item $v_1,\dots,v_k$ are distinct neighbours of $u$.
    \item $\chi(u,v_1) = \bot$ and $\chi(u,v_i) = c_{i-1}$ for all $i > 1$.
    \item Either $c_k \in \miss_\chi(u)$ or $c_k \in \{c_1,\dots,c_{k-1}\}$.
\end{enumerate}    
\end{definition}

\noindent
We say that the Vizing fan $\F = (u, \alpha),(v_1,c_1),\dots,(v_k,c_k)$ is $\alpha$-\emph{primed}, has \emph{center} $u$ and \emph{leaves} $v_1,\dots,v_k$. We refer to $c_i$ as the color of $v_i$ within $\F$.
A crucial property is that we can \emph{rotate}  colors around the Vizing fan $\F$ by setting $\chi(u, v_1) \leftarrow c_1,\dots, \chi(u, v_{i-1}) \leftarrow c_{i-1}, \chi(u, v_i) \leftarrow \bot$ for any $i \in [k]$. 
We say that $\F$ is a \emph{trivial} Vizing fan if $c_k \in \miss_\chi(u)$.
Note that, if $\F$ is trivial, we can immediately extend the coloring $\chi$ to $(u, v_1)$ by rotating all the colors around $\F$ and setting $\chi(u, v_k) \leftarrow c_k$.

\Cref{alg:vizing fan} describes the standard procedure used to construct Vizing fans. 
As input, it takes a vertex $u$ and a color $\alpha \in \miss_\chi(u)$, and returns an $\alpha$-primed Vizing fan with center $u$.

\begin{algorithm}[H]
    \SetAlgoLined
    \DontPrintSemicolon
    \SetKwRepeat{Do}{do}{while}
    \SetKwBlock{Loop}{repeat}{EndLoop}
    For each $x \in V$, let $\clr(x) \in \miss_\chi(x)$\;
    $k \leftarrow 1$ and $v_1 \leftarrow v$\;
    $c_1 \leftarrow \clr(v_1)$\; 
    \While{$c_k \notin \{c_1 ,\dots, c_{k-1}\}$ \textnormal{and} $c_k \notin \miss_\chi(u)$}{
        Let $(u, v_{k+1})$ be the edge with color $\chi(u, v_{k+1}) = c_k$\;
        $c_{k+1} \leftarrow \clr(v_{k+1})$\;
        $k \leftarrow k + 1$\;
    }
    \Return{$(u, \alpha),(v_1,c_1),\dots,(v_k,c_k)$}
    \caption{$\VizingF(u, v, \alpha)$}
    \label{alg:vizing fan}
\end{algorithm}

\medskip
\noindent \textbf{The Algorithm \textnormal{\textsf{Vizing}}}:
We now describe the algorithm $\Vizing$ that, given a Vizing fan $\F = (u, \alpha), (v_1,c_1),\dots,(v_k,c_k)$ as input, extends the coloring $\chi$ to the edge $(u,v_1)$ by building a Vizing chain. \Cref{alg:vizing} gives a formal description of this procedure.

\begin{algorithm}[H]
    \SetAlgoLined
    \DontPrintSemicolon
    \SetKwRepeat{Do}{do}{while}
    \SetKwBlock{Loop}{repeat}{EndLoop}
    \If{$\F$ is trivial}{
        $\chi(u, v_i) \leftarrow c_i$ for all $i \in [k]$\;
        \Return\;
    }
    Let $P$ denote the maximal $\{\alpha, c_k\}$-alternating path starting at $u$\;
    Extend $\chi$ to $(u,v_1)$ by flipping the path $P$ and rotating colors in $\F$ (details in \Cref{lem:full vizing proof})\; 
    \caption{\textsf{Vizing}$(\F)$}
    \label{alg:vizing}
\end{algorithm}

\noindent Thus, running $\Vizing(\F)$ extends the coloring $\chi$ to the uncolored edge $(u,v_1)$ by rotating colors in the Vizing fan $\F$ and flipping the colors of the alternating path $P$. We sometimes refer to the process of running $\Vizing(\F)$ as \emph{activating} the Vizing fan $\F$.

\begin{lemma}\label{lem:full vizing proof}
    \Cref{alg:vizing} extends the coloring $\chi$ to the edge $(u,v_1)$ in time $O(\Delta + |P|)$.
\end{lemma}

\begin{proof}
    To see that the path $P$ is well-defined, note that $\alpha \in \miss_\chi(u)$ and $c_k \notin \miss_\chi(u)$, so there is an $\{\alpha,c_k\}$-alternating path starting at $u$.

    \medskip
    \noindent \textit{Extending the coloring $\chi$:}
    By the definition of a Vizing fan, $c_k$ is the first repetition of some color in $\{c_1, \ldots, c_{k-1}\}$, so there is a unique index $j \in [k-1]$ such that $c_j = c_k$, and $P$ has $(u, v_{j+1})$ as its first edge with color $c_j$. We consider the cases where the path $P$ does or does not have $v_{j}$ as an endpoint. If $P$ does not end at $v_j$, then we can rotate the colors of the first $j + 1$ edges in the fan by setting $\chi(u, v_1) \leftarrow c_1,\dots,\chi(u, v_{j}) \leftarrow c_{j}$ and flip the colors of the alternating path $P$.
    If $P$ does end at $v_j$, then we flip the colors of the alternating path $P$, rotate the colors of the fan by setting $\chi(u, v_1) \leftarrow \chi(u, v_2),\dots,\chi(u, v_{k-1}) \leftarrow \chi(u, v_{k})$, and set $\chi(u, v_k) \leftarrow c_k$. Note that, while rotating the fan in this last step, we have $\chi(u, v_{j+1}) = \alpha \neq c_{k}$.

    \medskip
    \noindent
    Using standard data structures (see \Cref{sec:data structs}), we can implement this algorithm to run in time proportional to the number of edges that change their colors, which is $O(\Delta + |P|)$.
\end{proof}

\noindent Given a Vizing fan $\F$, we denote the path $P$ considered by \Cref{alg:vizing} by $\texttt{Vizing-Path}(\F)$. If the Vizing fan $\F$ is trivial, then $\texttt{Vizing-Path}(\F)$ denotes an empty path $\varnothing$.

\section{Basic Building Blocks}\label{sec:prelim:new}

In this section, we introduce the notation of u-fans, u-edges and separable collections, which are the definitions that work as the basic building blocks for our algorithms.

\subsection{U-Fans and U-Edges}
\label{sec:prelim:u}

We begin by defining the notion of \emph{u-fans} that was used by \cite{gabow1985algorithms}.\footnote{The term `u-fan' was originally introduced by \cite{gabow1985algorithms} as an abbreviation for `uncolored fan'.}

\begin{definition}[u-fan, \cite{gabow1985algorithms}]
    A \emph{u-fan} is a tuple $\f = (u, v, w, \alpha, \beta)$ where $u$, $v$ and $w$ are distinct vertices and $\alpha$ and $\beta$ are distinct colors such that:
    \begin{enumerate}
        \item $(u,v)$ and $(u,w)$ are uncolored edges.
        \item $\alpha \in \miss_\chi(u)$ and $\beta \in \miss_\chi(v) \cap \miss_\chi(w)$.
    \end{enumerate}
\end{definition}
\noindent
We say that $u$ is the \emph{center} of $\f$ and that $v$ and $w$ are the \emph{leaves} of $\f$. We also say that the u-fan $\f$ is \emph{$\{\alpha, \beta\}$-primed} and define $c_{\f}(u) := \alpha$, $c_{\f}(v) := \beta$ and $c_{\f}(w) := \beta$ (i.e.~given a vertex $x \in \f$, $c_{\f}(x)$ is the available color that $\f$ `assigns' to $x$).

\medskip
\noindent \textbf{Activating U-Fans:}
Let $\f$ be an $\{\alpha,\beta\}$-primed u-fan with center $u$ and leaves $v$ and $w$. The key property of u-fans is that at most one of the $\{\alpha, \beta\}$-alternating paths starting at $v$ or $w$ ends at $u$. Say that the $\{\alpha, \beta\}$-alternating path starting at $v$ does not end at $u$. Then, after flipping this $\{\alpha, \beta\}$-alternating path, both $u, v$ are missing color $\alpha$. Thus, we can extend the coloring $\chi$ by assigning $\chi(u, v)\leftarrow \alpha$. 
We refer to this as \emph{activating} the u-fan $\f$.

\medskip
\noindent
We also define the notion of a \emph{u-edge} similarly to u-fans.

\begin{definition}[u-edge]\label{def:u-edge}
    A \emph{u-edge} is a tuple $\e = (u, v, \alpha)$ where $(u,v)$ is an uncolored edge and $\alpha$ is a color such that $\alpha \in \miss_\chi(u)$.
\end{definition}

\noindent We say that $u$ is the \emph{center} of $\e$ and that $\alpha$ is the \emph{active color} of $\e$. For notational convenience, we also say that the u-edge $\e$ is \emph{$\alpha$-primed} and define $c_{\e}(u) := \alpha$ and $c_{\e}(v) = \bot$.\footnote{Whenever we refer to an ``uncolored edge $e$'', we are referring to an edge $e \in E$ such that $\chi(e) = \bot$, whereas a `u-edge $\e$' always refers to the object from \Cref{{def:u-edge}} and is denoted in bold.}

\medskip
\noindent \textbf{Collections of U-Components:} While working with both u-fans and u-edges simultaneously, we sometimes refer to some $\g$ that is either a u-fan or a u-edge as a \emph{u-component}. Throughout this paper, we often consider collections of u-components $\U$. We only use the term `collection' in this context, so we abbreviate `collection of u-components' by just `collection'.
We will be particularly interested in collections satisfying the following useful property, which we refer to as \emph{separability}.

\begin{definition}[Separable Collection]
    Let $\chi$ be a partial $(\Delta + 1)$-edge coloring of $G$ and $\mathcal U$ be a collection of u-components (i.e.~u-fans and u-edges).
    We say that the collection $\mathcal U$ is \emph{separable} if the following holds:
    \begin{enumerate}
        \item All of the u-components in $\mathcal U$ are edge-disjoint.
        \item For each $x \in V$, the colors in the multi-set
        $C_{\mathcal U}(x) := \{ c_{\g}(x) \mid \g \in \mathcal U, \, x \in \g\}$ are distinct.
    \end{enumerate}
\end{definition}
We remark that the second property of this definition is rather important because we need to ensure that different u-components are not interfering with each other when they share common vertices. Check \Cref{fig:separable} for an illustration.

\begin{figure}
    \centering
    \begin{tikzpicture}[thick,scale=1.2]
	\draw (0, 0) node(x)[circle, draw, color=cyan, fill=black!50,
inner sep=0pt, minimum width=10pt, line width=2pt, label = $x$] {};
    \draw (0, 0) node(xx)[circle, draw, color=magenta, inner sep=0pt, minimum width=11pt] {};

	\draw (-2, 2) node(u1)[circle, draw, fill=black!50,
	inner sep=0pt, minimum width=10pt, label = $u_1$] {};
		
	\draw (-4, 0) node(x1)[circle, draw, color=magenta, fill=black!50,
	inner sep=0pt, minimum width=10pt, label = $x_1$] {};
		
	\draw (2, 2) node(u2)[circle, draw, fill=black!50,
	inner sep=0pt, minimum width=10pt, label=$u_2$] {};
	
	\draw (4, 0) node(x2)[circle, draw,  color=cyan, fill=black!50,
	inner sep=0pt, minimum width=10pt, label = $x_2$] {};
	
	\draw [line width = 0.5mm, dashed] (u1) to (x1);
	\draw [line width = 0.5mm, dashed] (u2) to (x2);
        \draw [line width = 0.5mm, dashed] (u1) to (x);
	\draw [line width = 0.5mm, dashed] (u2) to (x);
	
\end{tikzpicture}
    \caption{This picture shows two u-fans $(u_1, x_1, x, *, \beta_1)$ and $(u_2, x_2, x, *, \beta_2)$ sharing a common vertex $x$. The separable condition requires that $\beta_1\neq \beta_2$; for instance $\beta_1, \beta_2$ could be magenta and cyan as shown here.}
    \label{fig:separable}
\end{figure}

\medskip
\noindent \textbf{Damaged U-Components:} Suppose we have a partial $(\Delta + 1)$-edge coloring $\chi$ and a separable collection $\U$ w.r.t.~$\chi$. Now, suppose that we modify the coloring $\chi$. We say that a u-component $\g \in \U$ is \emph{damaged} if it is no longer a u-component w.r.t.~the new coloring $\chi$.

We note that this can only happen for one of the following two reasons: (1) one of the uncolored edges in $\g$ is now colored, or (2) there is a vertex $x \in \g$ such the color $c_{\g}(x)$ that $\g$ assigns to $x$ is no longer missing at $x$. 

The following lemma shows that flipping the colors of an alternating path cannot damage many u-components in a separable collection $\U$.

\begin{lemma}\label{lem:low damage flips}
    Let $\chi$ be a partial $(\Delta + 1)$-edge coloring of a graph $G$, $\U$ a separable collection and $P$ a maximal alternating path in $\chi$. Then flipping the colors of the alternating path $P$ damages at most $2$ u-components in $\U$ (corresponding to the two endpoints of the path).
\end{lemma}

\begin{proof}
    Let $x$ and $y$ be the endpoints of the path $P$. Since only the palettes of $x$ and $y$ change after flipping the colors of $P$, only u-components $\g \in \U$ that contain $x$ and $y$ may be damaged. Since the palette of $x$ changes by $1$  and the collection $\U$ is separable, only one $\g \in \U$ containing $x$ may be damaged by the color $c_{\g}(x)$ no longer being available at $x$. The same holds for the vertex $y$. Thus, at most 2 u-components in $\U$ are damaged by this operation.
\end{proof}

\subsection{Data Structures}\label{sec:data struc overview}

In \Cref{sec:data structs}, we describe the data structures that we use to implement our algorithm. On top of the standard data structures used to maintain the edge coloring $\chi$, which allows us to implement Algorithms~\ref{alg:vizing fan} and \ref{alg:vizing} efficiently, we also use data structures that allow us to efficiently maintain a separable collection $\U$. More specifically, the data structures that we use to implement a separable collection $\U$ support the following queries.
\begin{itemize}
    \item $\textsc{Insert}_{\U}(\g)$:  The input to this query is a u-component $\g$. In response, the data structure adds $\g$ to $\U$ if $\U \cup \{\g\}$ is separable and outputs $\texttt{fail}$ otherwise.
    \item $\textsc{Delete}_{\U}(\g)$: The input to this query is a u-component $\g$. In response, the data structure removes $\g$ from $\U$ if $\g \in \U$ and outputs $\texttt{fail}$ otherwise.
    \item $\textsc{Find-Component}_{\U}(x,c)$: The input to this query is a vertex $x \in V$ and a color $c \in [\Delta + 1]$. In response, the data structure returns the u-component $g \in \U$ with $c_{\g}(x) = c$ if such a u-component exists and outputs $\texttt{fail}$ otherwise.
    \item  $\textsc{Missing-Color}_{\U}(x)$: The input to this query is a vertex $x \in V$. In response, the data structure returns an arbitrary color from the set $\miss_\chi(x) \setminus C_{\U}(x)$.\footnote{Note that, since $|C_{\U}(x)| < |\miss_\chi(x)|$, such a color always exists.}
\end{itemize}
The following claim shows that it is always possible to answer a $\textsc{Missing-Color}$ query.

\begin{claim}\label{claim:missing color}
    For each $x \in V$, the set $\miss_\chi(x) \setminus C_{\U}(x)$ is non-empty.
\end{claim}

\begin{proof}
    Let $d$ denote the number of uncolored edges incident on $x$.
    Since the collection $\U$ is separable, we have that $|C_{\U}(x)| \leq d$. Since $|\miss_\chi(x)| \geq d + 1$, it follows that $|C_{\U}(x)| < |\miss_\chi(x)|$ and so $\miss_\chi(x) \setminus C_{\U}(x) \neq \varnothing$.
\end{proof}

\noindent
Furthermore, the data structure supports the following initialization operation.
\begin{itemize}
    \item $\textsc{Initialize}(G, \chi)$: Given a graph $G$ and an edge coloring $\chi$ of $G$, we can initialize the data structure with an empty separable collection $\U = \varnothing$.
\end{itemize}
In \Cref{sec:data structs}, we show how to implement the initialization operation in $O(m)$ time and each of these queries in $O(1)$ time with the appropriate data structures.
\textbf{These queries provide the `interface' via which our algorithm will interact with the u-components.}

\medskip
\noindent \textbf{A Remark on Randomization:}
In order to get good space complexity and running time simultaneously, we implement the standard data structures (used for Algorithms~\ref{alg:vizing fan} and \ref{alg:vizing}) and the data structure supporting the preceding queries using \emph{hashmaps} (see \Cref{sec:data structs}). The following proposition describes the hashmaps used by our data structures.\footnote{We could have even
used the construction of~\cite{Kuszmaul22} to obtain exponentially small error probability in this case, but since we do not need this stronger guarantee we stick
with the simpler work of~\cite{DietzfelbingerH90}.}
\begin{proposition}[\cite{DietzfelbingerH90}]\label{prop:hash-map}
	There exists a dynamic dictionary that, given a parameter $k$, can handle $k$ insertions, deletions, and map operations, 
	uses $O(k)$ space and takes $O(1)$ worst-case time per operation with probability at least $1-O(1/k^7)$. 
\end{proposition}
\noindent
This implementation gives us the guarantee that \emph{across the run of the entire algorithm, every query made to one of these data structures takes $O(1)$ time and each initialization operation takes $O(m)$ time, with high probability}.
Since the randomization used for the hashmaps is independent of the randomization used in the rest of the algorithm, \textbf{we implicitly condition on the high probability event that, every operation performed using a hashmap runs in $O(1)$ time throughout the rest of the paper}.\footnote{Alternatively, one can replace these hashmaps with balanced search trees to make these data structures deterministic, incurring $O(\log n)$ overhead in the running time of each operation.}

\subsection{Vizing Fans in Separable Collections}\label{sec:U-avoid}

Within our algorithm, we only construct Vizing fans and Vizing chains in a setting where there is some underlying separable collection $\U$.
To ensure that activating and rotating colors around Vizing fans does not damage too many u-components, we choose the missing colors involved in Vizing fan constructions so that they `avoid' the colors assigned to the u-components in $\U$.
\begin{definition}\label{def:U-avoid}
    Let $\U$ be a separable collection and $\F = (u,\alpha),(v_1,c_1),\dots, (v_k,c_k)$ be a Vizing fan. We say that the Vizing fan $\F$ is \emph{$\U$-avoiding} if $c_i \in \miss_\chi(v_i) \setminus C_{\U}(v_i)$ for each leaf $v_i \in \F$.
\end{definition}

\noindent  We say that a Vizing fan \emph{$\F$ is a Vizing fan of the u-edge $\e = (u, v, \alpha)$} if $\F$ is $\alpha$-primed, has center $u$ and its first leaf is $v$. The following lemma shows that we can always find a $\U$-avoiding Vizing fan for a u-edge.

\begin{lemma}\label{lem:fast U-avoid fan}
    Given a u-edge $\e \in \U$, there exists a $\U$-avoiding Vizing fan $\F$ of $\e$. Furthermore, we can compute such a Vizing fan in $O(\Delta)$ time.
\end{lemma}

\begin{proof}
    By \Cref{claim:missing color}, we can always find a collection of colors $\{\clr(x)\}_{x \in V}$ such that $\clr(x) \in \miss_\chi(x) \setminus C_{\U}(x)$ for each $x \in V$. If we construct a Vizing fan $\F$ by calling $\VizingF$ and using such a collection of colors within the algorithm, then it follows that the Vizing fan $\F$ is $\U$-avoiding.
    By combining the standard implementation of \Cref{alg:vizing fan} with the queries described in \Cref{sec:data struc overview} we can do this in $O(\Delta)$ time.
\end{proof}

\noindent
The following lemma describes some crucial properties of $\U$-avoiding Vizing fans.

\begin{lemma}\label{lem:Vfans are safe}
Let $\chi$ be a $(\Delta + 1)$-edge coloring of a graph $G$ and $\U$ be a separable collection. For any u-edge $\e = (u,v,\alpha) \in \U$ with a $\U$-avoiding Vizing fan $\F$, we have the following:
\begin{enumerate}
    \item Rotating colors around $\F$ does not damage any u-component in $\U \setminus \{\e\}$.
    \item Calling $\Vizing(\F)$ damages at most one u-component in $\U \setminus \{\e\}$. Furthermore, we can identify the damaged u-component in $O(1)$ time.
\end{enumerate}
\end{lemma}

\begin{proof}
    Let $\F = (u, \alpha),(v_1, c_1),\dots, (v_k, c_k)$.
    Rotating colors around $\F$ will only remove the color $c_i$ from the palette of a leaf $v_i$ appearing in $\F$. Since $c_i \notin C_{\U}(v_i)$, these changes to the palettes will not damage any u-component in $\U$. However, rotating colors around $\F$ will color the edge $(u,v)$, damaging the u-edge $\e$.

    Now, suppose we call $\Vizing(\F)$.
    If the Vizing fan $\F$ is trivial, then the algorithm rotates all the colors in $\F$ and sets $\chi(u, v_k) \leftarrow c_k$. Rotating the colors can only damage $\e \in \U$, and removing $c_k$ from the palette of $u$ can damage at most one other u-component in $\U$ since $\U$ is separable. Using the queries from \Cref{sec:data struc overview}, we can check if such a u-component exist and return it in $O(1)$ time.
    If $\F$ is not trivial, then the algorithm flips the $\{\alpha, c_k\}$-alternating path $P = \VizingP(\F)$ and rotates colors around $\F$.
    By similar arguments, rotating colors around $\F$ can only damage $\e \in \U$. Let $x$ denote the endpoint of $P$ that is not $u$. Flipping the path $P$ might remove either the color $\alpha$ or $c_k$ from the palette of $x$, which might damage at most one u-component in $\U\setminus \{e\}$ since $\U$ is separable. We can again identify this u-component in $O(1)$ time using the queries from \Cref{sec:data struc overview}. It also removes the color $\alpha$ from the palette of $u$, but this cannot damage any u-component other than $\e$ (again, since $\U$ is separable).
\end{proof}

\section{The Main Algorithm}
\label{sec:main:algo}

We are now ready to present our main technical result, which is a slightly weaker version of~\Cref{thm:main}, and focuses on achieving a near-linear time algorithm for $(\Delta+1)$ edge coloring (instead of 
the exact $O(m\log{n})$ time in~\Cref{thm:main}; see the remark after that theorem). We will then use this theorem in~\Cref{sec:log(n)} to conclude the proof of~\Cref{thm:main}. 

\begin{theorem}\label{thm:main-tech}
    There is a randomized algorithm that, given any simple undirected graph $G = (V, E)$ on $n$ vertices and $m$ edges with maximum degree $\Delta$, finds a $(\Delta + 1)$-edge coloring of $G$ in $\tilde O(m)$ time with high probability.
\end{theorem}

Our main algorithm consists of two main components.
The first component is an algorithm called $\ConUFans$ that takes a partial $(\Delta+1)$-edge coloring $\chi$ with $\lambda$ uncolored edges and either extends $\chi$ to $\Omega(\lambda)$ of these edges or modifies the coloring to construct a separable collection of $\Omega(\lambda)$ u-fans. \Cref{lem:build u-fans} summarizes the behavior of this algorithm.

\begin{lemma}\label{lem:build u-fans}
    Given a graph $G$, a partial $(\Delta + 1)$-edge coloring $\chi$ of $G$ and a set of $\lambda$ uncolored edges $U$, the algorithm $\ConUFans$ does one of the following in $O(m + \Delta \lambda)$ time:
    \begin{enumerate}
        \item  Extends the coloring to $\Omega(\lambda)$ uncolored edges.
        \item  Modifies $\chi$ to obtain a separable collection of $\Omega(\lambda)$ u-fans $\mathcal U$.
    \end{enumerate}
\end{lemma}

The second component is an algorithm called $\ColorUFans$ that takes a collection of $\lambda$ u-fans and extends the coloring to $\Omega(\lambda)$ of the edges in the u-fans. \Cref{lem:color u-fans} summarizes the behavior of this algorithm. The reader may find it helpful to keep in mind that the algorithm for \Cref{lem:color u-fans} is a generalization of algorithm for \Cref{lem:main:bipartite} in \Cref{sec:bipartite}.

\begin{lemma}\label{lem:color u-fans}
    Given a graph $G$, a partial $(\Delta + 1)$-edge coloring $\chi$ of $G$ and a separable collection of $\lambda$ u-fans $\mathcal U$, the algorithm $\ColorUFans$ extends $\chi$ to $\Omega(\lambda)$ edges in $O(m \log n)$ time w.h.p.
\end{lemma}

In Sections~\ref{sec:algo build u-fans} and \ref{sec:algo}, we  prove Lemmas~\ref{lem:build u-fans} and \ref{lem:color u-fans} respectively.
Using these lemmas, we now show how to efficiently extend an edge coloring $\chi$ to the remaining uncolored edges.

\begin{lemma}\label{lem:main extend}
    Given a graph $G$ and a partial $(\Delta + 1)$-edge coloring $\chi$ of $G$ with $\lambda$ uncolored edges $U$, we can extend $\chi$ to the remaining uncolored edges in time $O((m + \Delta \lambda)\log^2 n)$ w.h.p.
\end{lemma}

\begin{proof}
    Let $U$ denote  the set of edges that are uncolored by $\chi$.
    We can then apply \Cref{lem:build u-fans} to either extend $\chi$ to a constant fraction of the edges in $U$ or construct a separable collection of $\Omega(\lambda)$ u-fans $\mathcal U$ in $O(m + \Delta \lambda)$ time. In the second case, we can then apply \Cref{lem:color u-fans} to color $\Omega(\lambda)$ of the edges contained in these u-fans in $O(m \log n)$ time w.h.p. Thus, we can extend $\chi$ to a constant fraction of the uncolored edges in $O(m \log n + \Delta \lambda)$ time w.h.p. After repeating this process  $O(\log \lambda) \leq O(\log n)$ many times, no edges remain uncolored. Thus, we can extend the coloring $\chi$ to the entire graph in $O((m + \Delta \lambda)\log^2 n)$ time w.h.p. 
\end{proof}

\noindent
We now use this lemma to prove \Cref{thm:main-tech}.

\begin{proof}[Proof of \Cref{thm:main-tech}.]
We prove this by applying \Cref{lem:main extend} to the standard Euler partition framework \cite{gabow1985algorithms, sinnamon2019fast}. Given a graph $G$, we partition it into two edge-disjoint subgraphs $G_1$ and $G_2$ on the same vertex set such that $\Delta(G_i) \leq \lceil \Delta/2 \rceil$ for each $G_i$, where $\Delta(G_i)$ denotes the maximum degree of $G_i$. We then recursively compute a $(\Delta(G_i) + 1)$-edge coloring $\chi_i$ for each $G_i$. Combining $\chi_1$ and $\chi_2$, we obtain a $(\Delta + 3)$-edge coloring $\chi$ of $G$. We then uncolor the two smallest color classes in $\chi$, which contain $O(m / \Delta)$ edges, and apply \Cref{lem:main extend} to recolor all of the uncolored edges in $\chi$ using only $\Delta + 1$ colors in $O(m \log^2 n)$ time w.h.p.

To show that the total running time of the algorithm is $O(m \log^3n)$,
first note that the depth of the recursion tree is $O(\log \Delta)$.
Next, consider the $i^{th}$ level of the recursion tree, for an arbitrary $i = O(\log \Delta)$: we have $2^i$ edge-disjoint subgraphs $G_1,\dots,G_{2^i}$ such that $\Delta(G_j) \leq O(\Delta/2^i)$ and $\sum_{j = 1}^{2^i} |E(G_j)| = m$.
Since the total running time at recursion level $i$ is $O(m \log^2 n)$ and the depth of the recursion tree is $O(\log \Delta)$, it follows that the total running time is $O(m \log^3 n)$ w.h.p.
\end{proof}

\section{The Algorithm $\ConUFans$: Proof of \Cref{lem:build u-fans}}\label{sec:algo build u-fans}

As input, the algorithm $\ConUFans$ is given a graph $G$ and a partial $(\Delta + 1)$-edge coloring $\chi$ of $G$ with $\lambda$ uncolored edges. It begins by taking the $\lambda$ uncolored edges in $\chi$ and using them to construct a separable collection $\U$ of $\lambda$ u-edges in an obvious way, which we describe in \Cref{lem:create u-edges}.
It then uses two subroutines, $\PruneVFans$ and $\ReduceUEdges$, to reduce the number of u-edges in $\U$, either by coloring them or modifying $\chi$ to turn them into u-fans. More specifically, for each color $\alpha \in [\Delta + 1]$, the algorithm considers the collection $\E_\alpha(\U)$ of u-edges that are $\alpha$-primed and (1) calls the subroutine $\PruneVFans$ to ensure that the u-edges in $\E_\alpha(\U)$ have vertex-disjoint Vizing fans by either coloring the u-edges with overlapping Vizing fans or using them to create u-fans, and (2) calls the subroutine $\ReduceUEdges$ which either extends the coloring to the u-edges in $\E_\alpha(\U)$ or uses them to create u-fans. \Cref{alg:constructUfans} gives the pseudocode for this algorithm.

\begin{algorithm}[H]
    \SetAlgoLined
    \DontPrintSemicolon
    \SetKwRepeat{Do}{do}{while}
    \SetKwBlock{Loop}{repeat}{EndLoop}
    Construct a separable collection of $\lambda$ u-edges $\U$\;
    \For{\textup{\textbf{each}} $\alpha \in [\Delta + 1]$}{
        $\mathcal F \leftarrow \PruneVFans(\U, \alpha)$\;
        $\ReduceUEdges(\U, \alpha, \mathcal F)$\;
    }
    \Return{$\U$}
    \caption{$\ConUFans$}
    \label{alg:constructUfans}
\end{algorithm}

\medskip
\noindent \textbf{Organization of \Cref{sec:algo build u-fans}:} We first describe and analyze the subroutines $\PruneVFans$ and $\ReduceUEdges$ used by the algorithm $\ConUFans$ before proving \Cref{lem:build u-fans} in \Cref{sec:proof of L2}.

\subsection{The Subroutine $\PruneVFans$}
\label{sec:prune}

As input, the subroutine $\PruneVFans$ is given a graph $G$, a partial $(\Delta + 1)$-edge coloring $\chi$ of $G$, a color $\alpha$ and a separable collection $\U$ with $\lambda_\alpha$ $\alpha$-primed u-edges.
We let $\E(\U) \subseteq \U$ denote the set of u-edges in $\U$ and $\E_\alpha(\U) \subseteq \E(\U)$ denote the set of u-edges in $\U$ that are $\alpha$-primed.
The subroutine then modifies the coloring $\chi$ and collection $\U$ to ensure that the Vizing fans of the u-edges in $\E_\alpha(\U)$ are vertex-disjoint.

\medskip
\noindent \textbf{Algorithm Description:}
The subroutine $\PruneVFans$ scans through the u-edges in $\E_\alpha(\U)$ and constructs $\U$-avoiding Vizing fans at these u-edges (see \Cref{lem:fast U-avoid fan}), maintaining a vertex-disjoint subset $\mathcal F$ of these Vizing fans throughout the process. After constructing a Vizing fan $\F$ for a u-edge $\e \in \E_\alpha(\U)$, the subroutine checks if $\F$ intersects some other fan in $\mathcal F$. If not, it adds $\F$ to $\mathcal F$. Otherwise, it uses these intersecting Vizing fans to either extend the coloring $\chi$ or construct a u-fan, removing the corresponding Vizing fans and u-edges from $\mathcal F$ and $\U$ respectively. See \Cref{prune-fan} for an illustration.

\begin{figure}
	\centering
	\begin{tikzpicture}[thick,scale=1.2]
	\draw (0, 0) node(0)[circle, draw, fill=black!50,
inner sep=0pt, minimum width=6pt, label = $w$] {};

	\draw (-2, 2) node(1)[circle, draw, color=cyan, fill=black!50,
	inner sep=0pt, minimum width=10pt, label = $u_1$] {};
		
	\draw (-4, 0) node(2)[circle, draw, fill=black!50,
	inner sep=0pt, minimum width=6pt, label = $v_1$] {};
	
	\draw (-3, 0) node(3)[circle, draw, fill=black!50,
	inner sep=0pt, minimum width=6pt] {};
	
	\draw (-2, 0) node(4)[circle, draw, fill=black!50,
	inner sep=0pt, minimum width=6pt] {};
	
	\draw (-1, 0) node(5)[circle, draw, fill=black!50,
	inner sep=0pt, minimum width=6pt] {};
	
	\draw (2, 2) node(6)[circle, draw, color=cyan, fill=black!50,
	inner sep=0pt, minimum width=10pt, label=$u_2$] {};
	
	\draw (1, 0) node(7)[circle, draw, fill=black!50,
	inner sep=0pt, minimum width=6pt] {};
	
	\draw (2, 0) node(8)[circle, draw, fill=black!50,
	inner sep=0pt, minimum width=6pt] {};
	
	\draw (3, 0) node(9)[circle, draw, fill=black!50,
	inner sep=0pt, minimum width=6pt] {};
	\draw (4, 0) node(10)[circle, draw, fill=black!50,
	inner sep=0pt, minimum width=6pt, label = $v_2$] {};
	
	\draw [line width = 0.5mm, dashed] (1) to (2);
	\draw [line width = 0.5mm, color=pink] (1) to (3);
	\draw [line width = 0.5mm, color=magenta] (1) to (4);
	\draw [line width = 0.5mm, color=purple] (1) to (5);
	\draw [line width = 0.5mm, color=Plum] (1) to (0);
	
	\draw [line width = 0.5mm, dashed] (6) to (10);
	\draw [line width = 0.5mm, color=lime] (6) to (9);
	\draw [line width = 0.5mm, color=ForestGreen] (6) to (8);
	\draw [line width = 0.5mm, color=teal] (6) to (7);
	\draw [line width = 0.5mm, color=olive] (6) to (0);
	
        \draw[->, >={Triangle}, thick, line width = 0.9mm] (0, -1.5) to (0, -2.5);
	
	\draw (0, -6) node(30)[circle, draw, fill=black!50,
	inner sep=0pt, minimum width=6pt, label=$w$] {};
	
	\draw (-2, -4) node(11)[circle, draw, color=cyan, fill=black!50,
	inner sep=0pt, minimum width=10pt, label = $u_1$] {};
	
	\draw (-4, -6) node(12)[circle, draw, fill=black!50,
	inner sep=0pt, minimum width=6pt, label = $v_1$] {};
	
	\draw (-3, -6) node(13)[circle, draw, fill=black!50,
	inner sep=0pt, minimum width=6pt] {};
	
	\draw (-2, -6) node(14)[circle, draw, fill=black!50,
	inner sep=0pt, minimum width=6pt] {};
	
	\draw (-1, -6) node(15)[circle, draw, fill=black!50,
	inner sep=0pt, minimum width=6pt] {};
	
	\draw (2, -4) node(16)[circle, draw, color=cyan, fill=black!50,
	inner sep=0pt, minimum width=10pt, label=$u_2$] {};
	
	\draw (1, -6) node(17)[circle, draw, fill=black!50,
	inner sep=0pt, minimum width=6pt] {};
	
	\draw (2, -6) node(18)[circle, draw, fill=black!50,
	inner sep=0pt, minimum width=6pt] {};
	
	\draw (3, -6) node(19)[circle, draw, fill=black!50,
	inner sep=0pt, minimum width=6pt] {};
	\draw (4, -6) node(20)[circle, draw, fill=black!50,
	inner sep=0pt, minimum width=6pt, label = $v_2$] {};
	
	\draw [line width = 0.5mm, pink] (11) to (12);
	\draw [line width = 0.5mm, color=magenta] (11) to (13);
	\draw [line width = 0.5mm, color=purple] (11) to (14);
	\draw [line width = 0.5mm, color=Plum] (11) to (15);
	\draw [line width = 0.5mm, dashed] (11) to (30);
	
	\draw [line width = 0.5mm, lime] (16) to (20);
	\draw [line width = 0.5mm, color=ForestGreen] (16) to (19);
	\draw [line width = 0.5mm, color=teal] (16) to (18);
	\draw [line width = 0.5mm, color=olive] (16) to (17);
	\draw [line width = 0.5mm, dashed] (16) to (30);
	
\end{tikzpicture}
	\caption{In this picture, we look at the Vizing fan around an uncolored edge $(u_1, v_1)$, and it intersects with the Vizing fan of another edge $(u_2, v_2)$ currently residing in $\mathcal{F}$ at vertex $w$. Then, we can rotate both Vizing fans and pair these two uncolored edges as a u-fan.}\label{prune-fan}
\end{figure}
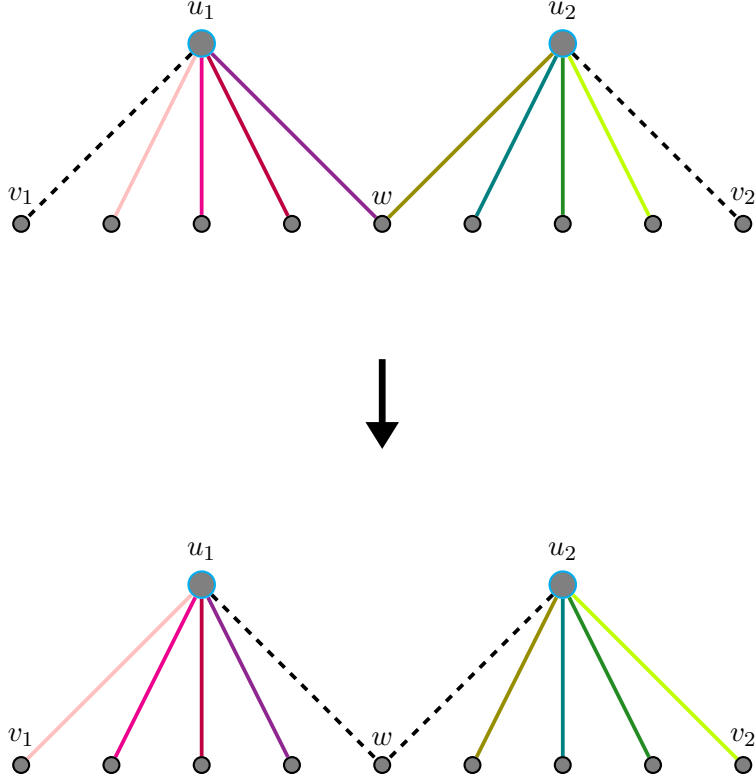

More formally, let $\e_1,\dots,\e_\ell$ denote the u-edges in $\E_\alpha(\U)$ at the time when the subroutine is called and let $\e_i = (u_i, v_i, \alpha)$. Initialize an empty set of Vizing fans $\mathcal F \leftarrow \varnothing$.
Then, for each $i\in [\ell]$, the subroutine constructs a $\U$-avoiding Vizing fan $\F_i$ of $\e_i$ and checks if any of the vertices in $\F_i$ (i.e.~any of the leaves or the center of $\F_i$) are contained in any of the Vizing fans in $\mathcal F$.
Note that, since $\U$ is separable and $u_1,\dots,u_\ell$ are the centers of $\alpha$-primed u-edges, these vertices are all distinct.
If $\F_i$ is vertex-disjoint from the Vizing fans in $\mathcal F$, we add $\F_i$ to $\mathcal F$. Otherwise, consider the following 3 cases. 
\begin{itemize}
    \item If $u_i$ is a leaf in $\F_j$ for some $\F_j \in \mathcal F$: Then we rotate the fan $\F_j$ so that $(u_i, u_j)$ is uncolored, set $\chi(u_i, u_j) \leftarrow \alpha$, remove $\F_j$ from $\mathcal F$ and remove $\e_i$ and $\e_j$ from $\E_\alpha(\U)$.
\end{itemize}
If this is not the case, then some leaf of $\F_i$ appears in a Vizing fan in $\mathcal F$.
Let $w$ be the first such leaf in $\F_i$.
\begin{itemize}
    \item If $w = u_j$ for some $\F_j \in \mathcal F$: Then we rotate the fan $\F_i$ so that $(u_i,w)$ is uncolored, set $\chi(u_i, w) \leftarrow \alpha$, remove $\F_j$ from $\mathcal F$ and remove $\e_i$ and $\e_j$ from $\E_\alpha(\U)$.
    \item If $w$ is a leaf in $\F_j \in \mathcal F$: Then we rotate the fans $\F_i$ and $\F_j$ so that $(u_i,w)$ and $(u_j,w)$ are uncolored, create a u-fan $\f = (w, u_i, u_j, \beta, \alpha)$ for an arbitrary $\beta \in \miss_\chi(w) \setminus C_{\U}(w)$, remove $\F_j$ from $\mathcal F$, add $\f$ to $\U$ and remove $\e_i$ and $\e_j$ from $\E_\alpha(\U)$.\footnote{Note that, since the vertices $u_1, \dots, u_\ell$ are all distinct, these cases are exhaustive.}
\end{itemize}
The subroutine then returns the set $\mathcal F$ of vertex-disjoint Vizing fans.

\subsubsection*{Analysis of $\PruneVFans$}

The following lemmas summarize the main properties of the subroutine $\PruneVFans$.

\begin{lemma}\label{lem:good Vfans}
    $\PruneVFans$ returns a set $\mathcal F$ of vertex-disjoint $\U$-avoiding Vizing fans for all of the u-edges in $\E_\alpha(\U)$.
\end{lemma}

\begin{proof}
    We can show by induction that, after the subroutine finishes scanning the u-edge $\e_i$, the set $\mathcal F$ consists of vertex-disjoint $\U$-avoiding Vizing fans for the u-edges in $\{\e_1,\dots,\e_i\} \cap \E_\alpha(\U)$.\footnote{Since we remove edges from $\E_\alpha(\U)$ throughout this process, $\{\e_1,\dots,\e_{i-1}\}$ might not be contained in $\E_\alpha(\U)$.}
    This is true trivially for $i = 1$. Now, suppose that this is true for $i-1$. If the $\U$-avoiding Vizing fan $\F_i$ of $\e_i$ does not intersect any of the Vizing fans in $\mathcal F$, then we add $\F_i$ to $\mathcal F$ and are done.
    
    Thus, we assume this is not the case. Let $x_0$ denote $u_i$ and $x_1,\dots,x_k$ denote the leaves of the Vizing fan $\F_i$ (in order). Let $x_p$ be the first vertex in the sequence $x_0,x_1,\dots,x_k$ which appears in some Vizing fan $\F_j \in \mathcal F$.
    It follows from our algorithm that we remove $\e_i$ and $\e_j$ from $\E_\alpha(\U)$ while only affecting the palettes of vertices in $\F_j$ and $x_0,\dots,x_p$ by rotating the fans $\F_i$ and $\F_j$. Since all of the Vizing fans in $\mathcal F$ were vertex-disjoint, changing the palettes of vertices in $\F_j$ cannot affect the palettes of vertices in any of the Vizing fans in $\mathcal F \setminus \{\F_j\}$, and hence they are still valid. Similarly, none of the vertices $x_0,\dots,x_p$ are contained in any of the Vizing fans in $\mathcal F \setminus \{\F_j\}$, so changing their palettes does not affect the validity of any other Vizing fans in $\mathcal F$. Thus, after removing $\F_j$ from $\mathcal F$, removing $\e_i$ and $\e_j$ from $\E_\alpha(\U)$ and making the changes to $\chi$, the set $\mathcal F$ consists of vertex-disjoint Vizing fans of the u-edges in $\{\e_1,\dots,\e_i\} \cap \E_\alpha(\U)$. To see why all of these Vizing fans also remain $\U$-avoiding, note that any u-fan $\f$ added to $\U$ is vertex-disjoint from these Vizing fans by an analogous argument.
\end{proof}

\begin{lemma}
    The subroutine $\PruneVFans$ maintains the invariant that $\U$ is separable.
\end{lemma}

\begin{proof}
    We can show by induction that, after the subroutine finishes scanning the u-edge $\e_i$, the collection $\U$ is separable.
    This is true trivially for $i = 1$, since the subroutine does not change $\chi$ or $\U$ while scanning this u-edge. Now, suppose that this is true for $i-1$. It follows from the inductive argument in \Cref{lem:good Vfans} that the Vizing fans currently in $\mathcal F$ are all $\U$-avoiding.
    If we do not change $\chi$ or $\U$ while scanning $\e_i$, then we are done. Thus, we assume that this is not the case. The subroutine can modify $\chi$ and $\U$ in the following ways:
    \begin{itemize}
    \item Rotating colors around some $\F_j$ so that $(u_j, u_i)$ is uncolored, setting $\chi(u_j, u_i) \leftarrow \alpha$ and removing $\e_i$ and $\e_j$ from $\U$.
    \item Rotating colors around $\F_i$ so that $(u_i, u_j)$ is uncolored for some $\F_j \in \mathcal F$, setting $\chi(u_j, u_i) \leftarrow \alpha$ and removing $\e_i$ and $\e_j$ from $\U$.
    \item Rotating colors around $\F_i$ and some $\F_j \in \mathcal F$ so that $(u_i, u_j)$ is uncolored, removing $\e_i$ and $\e_j$ from $\U$ and adding a u-fan $\f$ to $\U$.
    \end{itemize}
    In the first case, it follows from \Cref{lem:Vfans are safe} (since the Vizing fans are all $\U$-avoiding) that rotating colors around $\F_j$ does not damage any u-component in $\U \setminus \{\e_j\}$. Furthermore, since $\U$ is separable, setting $\chi(u_j, u_i) \leftarrow \alpha$ does not damage any u-component in $\U \setminus \{\e_i, \e_j\}$. Since we remove $\e_i$ and $\e_j$ from $\U$, we ensure that $\U$ remains separable. The same argument extends to the second case.
    
    Finally, for the third case, it follows from \Cref{lem:Vfans are safe} that rotating colors around $\F_i$ and $\F_j$ does not damage an u-component in $\U \setminus \{\e_i, \e_j\}$, so after removing $\e_i$ and $\e_j$ from $\U$ we have that $\U$ is separable. Since we have that $\alpha \in C_{\U}(u_i) \cap C_{\U}(u_j)$ after removing $\e_i$ and $\e_j$ from $\U$, we can see that $\f$ is indeed a u-fan and that $\U \cup \{\f\}$ is separable.\footnote{Note that, before removing $\e_i$ and $\e_j$ from $\U$, $c_{\e_i}(u_i) = c_{\e_j}(u_j) = \alpha$.}
\end{proof}

\begin{lemma}\label{lem:potential prune}
    Each time $\PruneVFans$ modifies the coloring $\chi$, it removes at most $2$ u-edges from $\E_\alpha(\U)$ and either (1) extends the coloring $\chi$ to one more edge, or (2) adds a u-fan to $\U$.
\end{lemma}

\begin{proof}
    This follows immediately from the description of the subroutine.
\end{proof}

\noindent The following lemma shows that this subroutine can be implemented efficiently.

\begin{lemma} \label{lem:prune}
    We can implement $\PruneVFans$ to run in $O(\Delta \lambda_\alpha)$ time.
\end{lemma} 
\begin{proof}
    By \Cref{lem:fast U-avoid fan}, we can construct each of the $\U$-avoiding Vizing fans $\F_i$ in $O(\Delta)$ time.
    As we scan through the u-edges $\e_1,\dots,\e_\ell$, we can maintain the set of vertices $S$ that are contained in the Vizing fans in $\mathcal F$, along with pointers to the corresponding Vizing fan. This allows us to either verify that $\F_i$ is vertex-disjoint from the Vizing fans in $\mathcal F$ or to find the first vertex in $\F_i$ that is contained in one of these fans in $O(\Delta)$ time. We can then appropriately update the coloring $\chi$ and the collection $\U$ in $O(\Delta)$ time using the data structures outlined in \Cref{sec:data struc overview}. It follows that the entire subroutine takes $O(\Delta \lambda_\alpha)$ time.
\end{proof}

\subsection{The Subroutine $\ReduceUEdges$}\label{sec:reduce}

As input, the subroutine $\ReduceUEdges$ is given a graph $G$, a partial $(\Delta + 1)$-edge coloring $\chi$ of $G$, a color $\alpha$, a separable collection $\U$ with $\lambda_\alpha$ $\alpha$-primed u-edges and a set $\mathcal F$ of vertex-disjoint $\U$-avoiding Vizing fans for all of the u-edges in $\E_\alpha(\U)$. Let $m_{\alpha}$ denote the number of edges with color $\alpha$ in the input coloring $\chi$. For each u-edge $\e \in \E_\alpha(\U)$, let $\F_{\e}$ denote its Vizing fan in $\mathcal F$. Similarly, for each $\F \in \mathcal F$, let $\e_{\F}$ denote its corresponding u-edge in $\E_{\alpha}(\U)$.
The subroutine then either extends the coloring to the u-edges in $\E_\alpha(\U)$ or uses them to construct u-fans.

\medskip
\noindent \textbf{Algorithm Description:}
The subroutine $\ReduceUEdges$ begins by activating all of the trivial Vizing fans in $\mathcal F$ and then removing them from $\mathcal F$. It then constructs Vizing chains at each of the remaining Vizing fans in $\mathcal F$ and explores the alternating paths $\VizingP(\F)$ `in parallel' (meaning, one edge at a time across all these paths). The subroutine continues this process until it either (1) reaches the end of one of these paths, or (2) two of these paths intersect.
In the first case, it identifies the Vizing fan $\F \in \mathcal F$ whose Vizing chain is fully explored and calls $\Vizing(\F)$ before removing $\F$ from $\mathcal F$. In the second case, it identifies the Vizing fans $\F, \F' \in \mathcal F$ corresponding to the intersecting Vizing chains and either extends the coloring $\chi$ to one more uncolored edge or creates a u-fan by shifting uncolored edges down these Vizing chains, before removing $\F$ and $\F'$ from $\mathcal F$. It repeats this process until at least half of the Vizing fans have been removed from $\mathcal F$.

More formally, the subroutine begins by scanning through all of the trivial Vizing fans $\F \in \mathcal F$ and calling $\Vizing(\F)$, removing the $\F$ from $\mathcal F$ and $\e_{\F}$ from $\U$, and removing any other u-component damaged by calling $\Vizing(\F)$ from $\U$ (see \Cref{lem:Vfans are safe}).
Let $L \leftarrow 0$, $S \leftarrow \varnothing$ and $P_{\F} \leftarrow \varnothing$ for all $\F \in \mathcal F$.
The subroutine then proceeds in \emph{rounds} which consist of updating each of the paths $P_{\F}$ while maintaining the following invariant.
\begin{invariant}\label{invaraint:paths}
After updating a path $P_{\F}$, the following hold:
    \begin{enumerate}
    \item For each $\F \in \mathcal F$, $P_{\F}$ is the length $|P_{\F}|$ prefix of the alternating path $\VizingP(\F)$, which we denote by $\VizingP(\F)_{\leq |P_{\F}|}$.\label{inv:1}
    \item The prefix paths $\{P_{\F}\}_{\F \in \mathcal F}$ are all edge-disjoint.\label{inv:2}
    \item The set $S$ is the union of all the edges in the prefix paths $\{P_{\F}\}_{\F \in \mathcal F}$.\label{inv:3}
    \item The collection $\U$ is separable and $|\mathcal F| = |\mathcal E_\alpha(\U)|$.\label{inv:4}
    \item The Vizing fans in $\mathcal F$ are $\U$-avoiding and vertex-disjoint.\label{inv:5}
\end{enumerate}
\end{invariant}
\noindent The prefix paths $\{P_{\F}\}_{\F \in \mathcal F}$ maintained by the algorithm also satisfy the following invariant.
\begin{invariant}\label{invaraint:iterations}
    At the start of each round, we have that $|P_{\F}| = L$ for all $\F \in \mathcal F$.
\end{invariant}

\noindent
Immediately after starting a round, the subroutine increases the value of $L$ by $1$, invalidating \Cref{invaraint:iterations}. To restore this invariant,
we need to update each of the prefix paths $P_{\F}$, which in turn may require us to update the other objects maintained by the subroutine to ensure that the other conditions of \Cref{invaraint:paths} are satisfied. The subroutine does this by scanning through each $\F \in \mathcal F$ and calling $\UpdatePath(\F)$, which we describe formally in \Cref{alg:update path}. The subroutine performs these rounds until $|\mathcal F| \leq \lambda_\alpha/2$, at which point it terminates.

\begin{algorithm}
    \SetAlgoLined
    \DontPrintSemicolon
    \SetKwRepeat{Do}{do}{while}
    \SetKwBlock{Loop}{repeat}{EndLoop}
    $P_{\F} \leftarrow \VizingP(\F)_{\leq |P_{\F}| + 1}$\;
    Let $(x,y)$ be the $L^{th}$ edge in $P_{\F}$ and $y$ be an endpoint of $P_{\F}$\;
    \If{$(x,y) \in S$}{
        Let $\F' \in \mathcal F$ be the Vizing fan such that $(x,y) \in P_{\F'}$\;
        \If{$(x,y)$ appears in the same orientation in $P_{\F}$ and $P_{\F'}$\label{line:case 1}}{
            Let $(z, x)$ and $(z', x)$ be the edges appearing before $(x,y)$ in $P_{\F}$ and $P_{\F'}$ respectively\;
            $\beta \leftarrow \chi(z, x)$\;
            $\chi(z, x) \leftarrow \bot$ and $\chi(z', x) \leftarrow \bot$\;
            $\Vizing(\F)$\;
            $\Vizing(\F')$\;
            Remove $\F$ and $\F'$ from $\mathcal F$ and $\e_{\F}$ and $\e_{\F'}$ from $\U$\;
            Add the u-fan $\f = (x, z, z', \beta, \alpha)$ to $\U$\label{lin:12}\;
        }
        \Else{ 
            \label{line:case 2}
            $\chi(x,y) \leftarrow \bot$\label{lin:14}\;
            $\Vizing(\F)$\;
            $\Vizing(\F')$\label{lin:16}\;
            Remove $\F$ and $\F'$ from $\mathcal F$ and $\e_{\F}$ and $\e_{\F'}$ from $\U$\;
        }
        $S \leftarrow S \setminus (P_{\F} \cup P_{\F'})$\;
        \Return
    }
    $S \leftarrow S \cup \{(x,y)\}$\;
    \If{$P_{\F}$ is maximal \label{line:case 3}}{
        $\Vizing(\F)$\label{lin:22}\;
        Remove $\F$ from $\mathcal F$ and $\e_{\F}$ from $\U$\;
        $S \leftarrow S \setminus P_{\F}$\;
        \If{there exists $\g \in \U$ with $c_{\g}(y) = \chi(x,y)$}{
            Remove $\g$ from $\U$\label{line:g} (see Lemma~\ref{lem:Vfans are safe})\;
        }
        \If{there exists $\F' \in \mathcal F$ such that $y \in \F$}{
            Remove $\F'$ from $\mathcal F$ and $\e_{\F'}$ from $\U$ \textit{\texttt{// $\e_{F'}$ might also get removed in \Cref{line:g}}}\label{line:g to F}\;
            $S \leftarrow S \setminus P_{\F'}$\label{lin:29}\;
        }
    }
    \caption{$\UpdatePath(\F)$}
    \label{alg:update path}
\end{algorithm}

\subsubsection*{Analysis of $\ReduceUEdges$}

The following lemmas summarize the main properties of the subroutine $\ReduceUEdges$.

\begin{lemma}\label{lem:inv 1 proof}
    The subroutine $\ReduceUEdges$ satisfies \Cref{invaraint:paths}.
\end{lemma}

\begin{proof}
    We show that, as long as \Cref{invaraint:paths} is satisfied, calling $\UpdatePath(\F)$ for some $\F \in \mathcal F$ maintains the invariant. We first note that \Cref{invaraint:paths} is trivially satisfied immediately after initializing $L$, $S$ and $\{P_{\F}\}_{\F \in \mathcal F}$.
    
    After activating and removing all of the trivial Vizing fans in $\mathcal F$, Conditions~\ref{inv:1}-\ref{inv:3} are clearly still satisfied (since the paths are all empty). For Condition~\ref{inv:4}, note that the subroutine removes any damaged u-components from $\U$, removes $\e$ from $\E_{\alpha}(\U)$ if it removes $\F_{\e}$ from $\mathcal F$, and that activating some trivial $\F_{\e} \in \mathcal F$ cannot damage any $\e' \in \E_{\alpha}(\U) \setminus \{\e\}$ since their Vizing fans are vertex-disjoint. For Condition~\ref{inv:5}, note that the subroutine does not add any u-component to $\U$, so the vertex-disjoint Vizing fans remain $\U$-avoiding. This establishes the base case of the induction.

    Now, assume that \Cref{invaraint:paths} is satisfied and suppose that we call $\UpdatePath(\F)$ for some $\F \in \mathcal F$. We now argue that each condition of \Cref{invaraint:paths} is satisfied after handling this call.

    \medskip
    \noindent \textit{Conditions~\ref{inv:2} and \ref{inv:3}:} After extending $P_{\F}$ by one more edge, we check if there is some other path $P_{\F'}$ that intersects the updated path $P_{\F}$ at this new edge. If so, we remove both $\F$ and $\F'$ from $\mathcal F$, ensuring that the remaining paths are edge-disjoint. Otherwise, the paths are all edge-disjoint. It's straightforward to verify that $S$ is updated correctly in each case.
    
    \medskip
    \noindent \textit{Condition~\ref{inv:4}:} We first show that the collection $\U$ remains separable. If this call to $\UpdatePath$ does not change $\chi$ or $\U$, then clearly $\U$ remains separable. 
    Note that the Vizing fans in $\mathcal F$ are all $\U$-avoiding and recall \Cref{lem:Vfans are safe}.
    We now consider the following three cases.
	\begin{enumerate}[(1),leftmargin=*]
		\item If we enter the \textbf{if} statement on \Cref{line:case 1}, then calling $\Vizing(\F)$ and $\Vizing(\F')$ after uncoloring edges on their Vizing chains does not damage any u-component in $\U$ apart from $\e_{\F}$ and $\e_{\F'}$, which are both removed from $\U$. This is because shifting colors around the Vizing fans $\F$ and $\F'$ does not damage any u-components apart from $\e_{\F}$ and $\e_{\F'}$ (see \Cref{lem:Vfans are safe}) and truncating an alternating path before flipping it ensures that we do not remove any colors from the palettes of the vertices on that path. Finally, note that the u-fan $\f$ that we add to $\U$ only uses colors that were previously unavailable. See the example \Cref{pair-same} for an illustration.
		\begin{figure}
			\centering
			\begin{tikzpicture}[thick,scale=0.9]
	\draw (0, 3) node(1)[circle, draw, color=cyan, fill=black!50,
	inner sep=0pt, minimum width=10pt, label = $u_1$] {};
	
	\draw (-2, 1) node(2)[circle, draw, fill=black!50,
	inner sep=0pt, minimum width=6pt, label = 180:{$v_1$}] {};
	
	\draw (-1, 1) node(3)[circle, draw, fill=black!50,
	inner sep=0pt, minimum width=6pt] {};
	
	\draw (0, 1) node(4)[circle, draw, fill=black!50,
	inner sep=0pt, minimum width=6pt] {};
	
	\draw (1, 1) node(5)[circle, draw, fill=black!50,
	inner sep=0pt, minimum width=6pt] {};
	
	\draw (3, 1) node(6)[circle, draw, fill=black!50,
	inner sep=0pt, minimum width=6pt] {};
	
	\draw (5, 1) node(7)[circle, draw, fill=black!50,
	inner sep=0pt, minimum width=6pt] {};
	
	\draw (7, 1) node(8)[circle, draw, fill=black!50,
	inner sep=0pt, minimum width=6pt, label=$z$] {};
	
	\draw (9, 0) node(9)[circle, draw, fill=black!50,
	inner sep=0pt, minimum width=6pt, label=$x$] {};
	\draw (11, 0) node(10)[circle, draw, fill=black!50,
	inner sep=0pt, minimum width=6pt, label=$y$] {};
	
	\draw (0, -3) node(11)[circle, draw, color=cyan, fill=black!50,
	inner sep=0pt, minimum width=10pt, label = -90:{$u_2$}] {};
	
	\draw (-2, -1) node(12)[circle, draw, fill=black!50,
	inner sep=0pt, minimum width=6pt, label = 180:{$v_2$}] {};
	
	\draw (-1, -1) node(13)[circle, draw, fill=black!50,
	inner sep=0pt, minimum width=6pt] {};
	
	\draw (0, -1) node(14)[circle, draw, fill=black!50,
	inner sep=0pt, minimum width=6pt] {};
	
	\draw (1, -1) node(15)[circle, draw, fill=black!50,
	inner sep=0pt, minimum width=6pt] {};
	
	\draw (3, -1) node(16)[circle, draw, fill=black!50,
	inner sep=0pt, minimum width=6pt] {};
	
	\draw (5, -1) node(17)[circle, draw, fill=black!50,
	inner sep=0pt, minimum width=6pt] {};
	
	\draw (7, -1) node(18)[circle, draw, fill=black!50,
	inner sep=0pt, minimum width=6pt, label=-90:{$z'$}] {};
	
	\draw [line width = 0.5mm, dashed] (1) to (2);
	\draw [line width = 0.5mm, pink] (1) to (3);
	\draw [line width = 0.5mm, magenta] (1) to (4);
	\draw [line width = 0.5mm, purple] (1) to (5);
	\draw [line width = 0.5mm, cyan] (5) to node[below] {$\alpha$} (6);
	\draw [line width = 0.5mm, purple] (6) to (7);
	\draw [line width = 0.5mm, cyan] (7) to node[below] {$\alpha$} (8);
	\draw [line width = 0.5mm, purple] (8) to (9);
	\draw [line width = 0.5mm, cyan] (9) to node[below] {$\alpha$} (10);
	
	\draw [line width = 0.5mm, dashed] (11) to (12);
	\draw [line width = 0.5mm, lime] (11) to (13);
	\draw [line width = 0.5mm, teal] (11) to (14);
	\draw [line width = 0.5mm, olive] (11) to (15);
	\draw [line width = 0.5mm, cyan] (15) to node[below] {$\alpha$} (16);
	\draw [line width = 0.5mm, olive] (16) to (17);
	\draw [line width = 0.5mm, cyan] (17) to node[below] {$\alpha$} (18);
	\draw [line width = 0.5mm, olive] (18) to (9);
	
	\draw [gray!50] plot [smooth cycle] coordinates {(0, 4) (-3, 1)  (-0.5, 0) (2, 1)};
	\node at (0, 4.5) {Vizing fan $\F$ around $u_1$};
	
	\draw [gray!50] plot [smooth cycle] coordinates {(0, -4) (-0.5, -3) (0.5, -1) (1.5, -0.5) (7, -0.5) (9, 0.8) (11, 0.8) (12, 0) (11, -0.5) (9, -0.5) (7, -1.8) (1.8, -1.5)};
	\node at (7, -2.5) {A prefix of the Vizing chain $P_{\F'}$ starting at $u_2$};
	
        \draw[->, >={Triangle}, thick, line width = 0.9mm] (4, -4.5) to (4, -5.5);
	
	\draw (0, -7) node(21)[circle, draw, fill=black!50,
	inner sep=0pt, minimum width=6pt, label = $u_1$] {};
	
	\draw (-2, -9) node(22)[circle, draw, fill=black!50,
	inner sep=0pt, minimum width=6pt, label = 180:{$v_1$}] {};
	
	\draw (-1, -9) node(23)[circle, draw, fill=black!50,
	inner sep=0pt, minimum width=6pt] {};
	
	\draw (0, -9) node(24)[circle, draw, fill=black!50,
	inner sep=0pt, minimum width=6pt] {};
	
	\draw (1, -9) node(25)[circle, draw, fill=black!50,
	inner sep=0pt, minimum width=6pt] {};
	
	\draw (3, -9) node(26)[circle, draw, fill=black!50,
	inner sep=0pt, minimum width=6pt] {};
	
	\draw (5, -9) node(27)[circle, draw, fill=black!50,
	inner sep=0pt, minimum width=6pt] {};
	
	\draw (7, -9) node(28)[circle, draw, color=cyan, fill=black!50,
	inner sep=0pt, minimum width=10pt, label=$z$] {};
	
	\draw (9, -10) node(29)[circle, draw, fill=black!50,
	inner sep=0pt, minimum width=6pt, label = $x$] {};
	\draw (11, -10) node(30)[circle, draw, fill=black!50,
	inner sep=0pt, minimum width=6pt, label = $y$] {};
	
	\draw (0, -13) node(31)[circle, draw, fill=black!50,
	inner sep=0pt, minimum width=6pt, label = -90:{$u_2$}] {};
	
	\draw (-2, -11) node(32)[circle, draw, fill=black!50,
	inner sep=0pt, minimum width=6pt, label = 180:{$v_2$}] {};
	
	\draw (-1, -11) node(33)[circle, draw, fill=black!50,
	inner sep=0pt, minimum width=6pt] {};
	
	\draw (0, -11) node(34)[circle, draw, fill=black!50,
	inner sep=0pt, minimum width=6pt] {};
	
	\draw (1, -11) node(35)[circle, draw, fill=black!50,
	inner sep=0pt, minimum width=6pt] {};
	
	\draw (3, -11) node(36)[circle, draw, fill=black!50,
	inner sep=0pt, minimum width=6pt] {};
	
	\draw (5, -11) node(37)[circle, draw, fill=black!50,
	inner sep=0pt, minimum width=6pt] {};
	
	\draw (7, -11) node(38)[circle, draw, color=cyan, fill=black!50,
	inner sep=0pt, minimum width=10pt, label=-90:{$z'$}] {};
	
	\draw [line width = 0.5mm, pink] (21) to (22);
	\draw [line width = 0.5mm, magenta] (21) to (23);
	\draw [line width = 0.5mm, purple] (21) to (24);
	\draw [line width = 0.5mm, cyan] (21) to node[right] {$\alpha$} (25);
	\draw [line width = 0.5mm, purple] (25) to (26);
	\draw [line width = 0.5mm, cyan] (26) to node[below] {$\alpha$} (27);
	\draw [line width = 0.5mm, purple] (27) to (28);
	\draw [line width = 0.5mm, dashed] (28) to (29);
	\draw [line width = 0.5mm, cyan] (29) to node[below] {$\alpha$} (30);
	
	\draw [line width = 0.5mm, lime] (31) to (32);
	\draw [line width = 0.5mm, teal] (31) to (33);
	\draw [line width = 0.5mm, olive] (31) to (34);
	\draw [line width = 0.5mm, cyan] (31) to node[right] {$\alpha$} (35);
	\draw [line width = 0.5mm, olive] (35) to (36);
	\draw [line width = 0.5mm, cyan] (36) to node[below] {$\alpha$} (37);
	\draw [line width = 0.5mm, olive] (37) to (38);
	\draw [line width = 0.5mm, dashed] (38) to (29);
	
\end{tikzpicture}
			\caption{In this picture, $\alpha$ is blue, and two uncolored edges $\e_{\F} = (u_1, v_1), \e_{\F'} = (u_2, v_2)$ generate Vizing fans $\F$ and $\F'$, and the two Vizing chains are joining at $(x, y)$ for the first time and in the same direction. Then, we can rotate both Vizing fans and partially flip $P_{\F}$ and $P_{\F'}$ to make a u-fan}\label{pair-same}
		\end{figure}
		
		\item If we enter the \textbf{else} statement on \Cref{line:case 2}, then the argument is completely analogous to the previous case, except that we do not add any u-components to $\U$ (see \Cref{lem:Vfans are safe}). See \Cref{pair-opposite} for an illustration.
		\begin{figure}
			\centering
			\begin{tikzpicture}[thick,scale=0.9]
	\draw (0, 3) node(1)[circle, draw, color=cyan, fill=black!50,
	inner sep=0pt, minimum width=10pt, label = $u_1$] {};
	
	\draw (-2, 1) node(2)[circle, draw, fill=black!50,
	inner sep=0pt, minimum width=6pt, label = 180:{$v_1$}] {};
	
	\draw (-1, 1) node(3)[circle, draw, fill=black!50,
	inner sep=0pt, minimum width=6pt] {};
	
	\draw (0, 1) node(4)[circle, draw, fill=black!50,
	inner sep=0pt, minimum width=6pt] {};
	
	\draw (1, 1) node(5)[circle, draw, fill=black!50,
	inner sep=0pt, minimum width=6pt] {};
	
	\draw (3, 1) node(6)[circle, draw, fill=black!50,
	inner sep=0pt, minimum width=6pt] {};
	
	\draw (5, 1) node(7)[circle, draw, fill=black!50,
	inner sep=0pt, minimum width=6pt] {};
	
	\draw (7, 1) node(8)[circle, draw, fill=black!50,
	inner sep=0pt, minimum width=6pt] {};
	
	\draw (9, 1) node(9)[circle, draw, fill=black!50,
	inner sep=0pt, minimum width=6pt, label=$x$] {};
	\draw (9, -1) node(10)[circle, draw, fill=black!50,
	inner sep=0pt, minimum width=6pt, label=-90:{$y$}] {};
	
	\draw (0, -3) node(11)[circle, draw, color=cyan, fill=black!50,
	inner sep=0pt, minimum width=10pt, label = -90:{$u_2$}] {};
	
	\draw (-2, -1) node(12)[circle, draw, fill=black!50,
	inner sep=0pt, minimum width=6pt, label = 180:{$v_2$}] {};
	
	\draw (-1, -1) node(13)[circle, draw, fill=black!50,
	inner sep=0pt, minimum width=6pt] {};
	
	\draw (0, -1) node(14)[circle, draw, fill=black!50,
	inner sep=0pt, minimum width=6pt] {};
	
	\draw (1, -1) node(15)[circle, draw, fill=black!50,
	inner sep=0pt, minimum width=6pt] {};
	
	\draw (3, -1) node(16)[circle, draw, fill=black!50,
	inner sep=0pt, minimum width=6pt] {};
	
	\draw (5, -1) node(17)[circle, draw, fill=black!50,
	inner sep=0pt, minimum width=6pt] {};
	
	\draw (7, -1) node(18)[circle, draw, fill=black!50,
	inner sep=0pt, minimum width=6pt] {};
	
	\draw [line width = 0.5mm, dashed] (1) to (2);
	\draw [line width = 0.5mm, pink] (1) to (3);
	\draw [line width = 0.5mm, magenta] (1) to (4);
	\draw [line width = 0.5mm, purple] (1) to (5);
	\draw [line width = 0.5mm, cyan] (5) to node[below] {$\alpha$} (6);
	\draw [line width = 0.5mm, purple] (6) to (7);
	\draw [line width = 0.5mm, cyan] (7) to node[below] {$\alpha$} (8);
	\draw [line width = 0.5mm, purple] (8) to (9);
	\draw [line width = 0.5mm, cyan] (9) to node[right] {$\alpha$} (10);
	
	\draw [line width = 0.5mm, dashed] (11) to (12);
	\draw [line width = 0.5mm, lime] (11) to (13);
	\draw [line width = 0.5mm, teal] (11) to (14);
	\draw [line width = 0.5mm, olive] (11) to (15);
	\draw [line width = 0.5mm, cyan] (15) to node[below] {$\alpha$} (16);
	\draw [line width = 0.5mm, olive] (16) to (17);
	\draw [line width = 0.5mm, cyan] (17) to node[below] {$\alpha$} (18);
	\draw [line width = 0.5mm, olive] (18) to (10);
	
        \draw[->, >={Triangle}, thick, line width = 0.9mm] (4, -4.5) to (4, -5.5);
	
	\draw (0, -7) node(21)[circle, draw, fill=black!50,
	inner sep=0pt, minimum width=6pt, label = $u_1$] {};
	
	\draw (-2, -9) node(22)[circle, draw, fill=black!50,
	inner sep=0pt, minimum width=6pt, label = 180:{$v_1$}] {};
	
	\draw (-1, -9) node(23)[circle, draw, fill=black!50,
	inner sep=0pt, minimum width=6pt] {};
	
	\draw (0, -9) node(24)[circle, draw, fill=black!50,
	inner sep=0pt, minimum width=6pt] {};
	
	\draw (1, -9) node(25)[circle, draw, fill=black!50,
	inner sep=0pt, minimum width=6pt] {};
	
	\draw (3, -9) node(26)[circle, draw, fill=black!50,
	inner sep=0pt, minimum width=6pt] {};
	
	\draw (5, -9) node(27)[circle, draw, fill=black!50,
	inner sep=0pt, minimum width=6pt] {};
	
	\draw (7, -9) node(28)[circle, draw, fill=black!50,
	inner sep=0pt, minimum width=6pt] {};
	
	\draw (9, -9) node(29)[circle, draw, fill=black!50,
	inner sep=0pt, minimum width=6pt, label = $x$] {};
	\draw (9, -11) node(30)[circle, draw, fill=black!50,
	inner sep=0pt, minimum width=6pt, label = -90:{$y$}] {};
	
	\draw (0, -13) node(31)[circle, draw, fill=black!50,
	inner sep=0pt, minimum width=6pt, label = -90:{$u_2$}] {};
	
	\draw (-2, -11) node(32)[circle, draw, fill=black!50,
	inner sep=0pt, minimum width=6pt, label = 180:{$v_2$}] {};
	
	\draw (-1, -11) node(33)[circle, draw, fill=black!50,
	inner sep=0pt, minimum width=6pt] {};
	
	\draw (0, -11) node(34)[circle, draw, fill=black!50,
	inner sep=0pt, minimum width=6pt] {};
	
	\draw (1, -11) node(35)[circle, draw, fill=black!50,
	inner sep=0pt, minimum width=6pt] {};
	
	\draw (3, -11) node(36)[circle, draw, fill=black!50,
	inner sep=0pt, minimum width=6pt] {};
	
	\draw (5, -11) node(37)[circle, draw, fill=black!50,
	inner sep=0pt, minimum width=6pt] {};
	
	\draw (7, -11) node(38)[circle, draw, fill=black!50,
	inner sep=0pt, minimum width=6pt] {};
	
	\draw [line width = 0.5mm, pink] (21) to (22);
	\draw [line width = 0.5mm, magenta] (21) to (23);
	\draw [line width = 0.5mm, purple] (21) to (24);
	\draw [line width = 0.5mm, cyan] (21) to node[right] {$\alpha$} (25);
	\draw [line width = 0.5mm, purple] (25) to (26);
	\draw [line width = 0.5mm, cyan] (26) to node[below] {$\alpha$} (27);
	\draw [line width = 0.5mm, purple] (27) to (28);
	\draw [line width = 0.5mm, cyan] (28) to node[below] {$\alpha$} (29);
	\draw [line width = 0.5mm, dashed] (29) to (30);
	
	\draw [line width = 0.5mm, lime] (31) to (32);
	\draw [line width = 0.5mm, teal] (31) to (33);
	\draw [line width = 0.5mm, olive] (31) to (34);
	\draw [line width = 0.5mm, cyan] (31) to node[right] {$\alpha$} (35);
	\draw [line width = 0.5mm, olive] (35) to (36);
	\draw [line width = 0.5mm, cyan] (36) to node[below] {$\alpha$} (37);
	\draw [line width = 0.5mm, olive] (37) to (38);
	\draw [line width = 0.5mm, cyan] (38) to node[below] {$\alpha$} (30);
	
\end{tikzpicture}
			\caption{In this picture, $\alpha$ is blue, and two uncolored edges $\e_{\F} = (u_1, v_1), \e_{\F'} = (u_2, v_2)$ generate Vizing fans $\F$ and $\F'$, and the two Vizing chains are joining at $(x, y)$ for the first time and in the opposite direction. Then, we can uncolor $(x, y)$ and color both $\e_{\F}, \e_{\F'}$.}\label{pair-opposite}
		\end{figure}
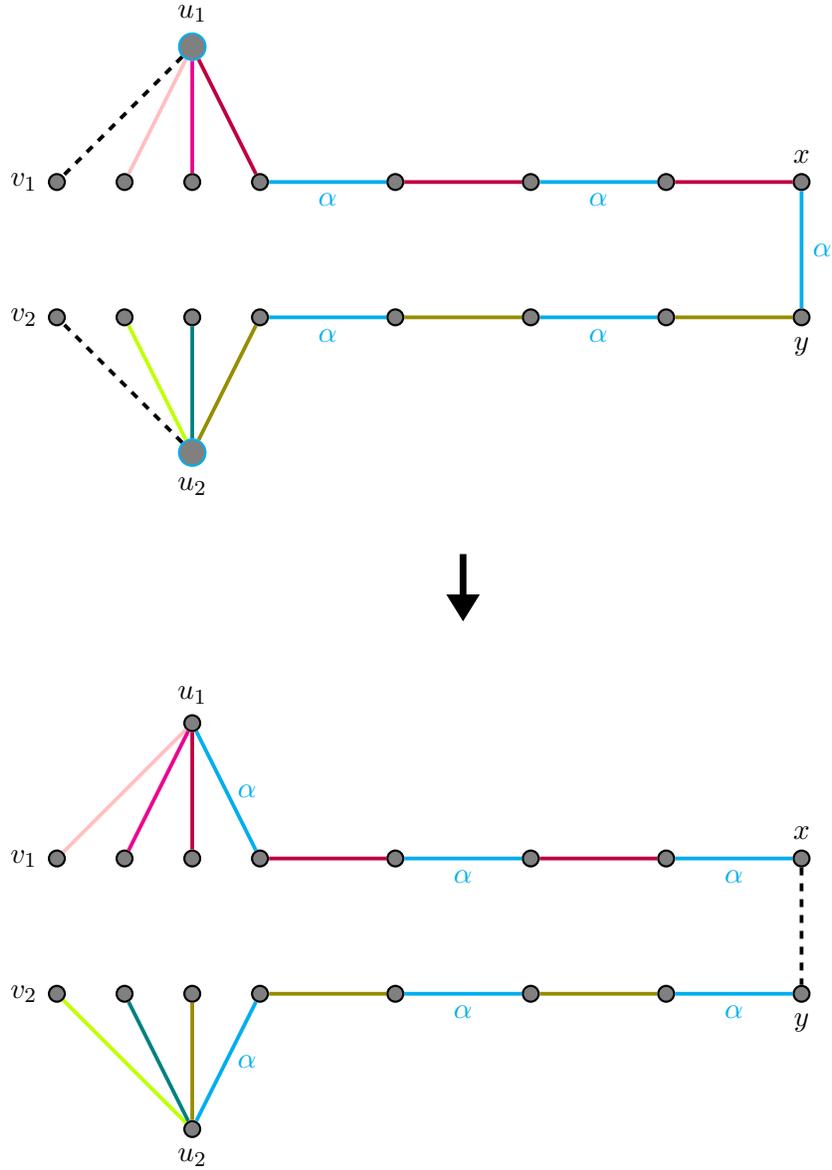
		
		\item If we enter the \textbf{if} statement on \Cref{line:case 3}, then calling $\Vizing(\F)$ damages at most one u-component in $\U$ apart from $\e_{\F}$ (see Lemma~\ref{lem:Vfans are safe}). In particular, if there is such a u-component $\g$, it will contain the vertex $y$ and have $c_{\g}(x) = \chi(x,y)$. Thus, we check if such a u-component exists and remove it from $\U$. 
	\end{enumerate}
    Finally, to see that $|\mathcal F| = |\mathcal E_{\alpha}(\U)|$, we can verify that we remove $\e$ from $\E_{\alpha}(\U)$ if and only if we remove $\F_{\e}$ from $\mathcal F$. It's clear that whenever we remove some $\F'$ from $\mathcal F$, we also remove $\e_{\F'}$ from $\E_{\alpha}(\U)$. Similarly, if we remove a u-edge $\e$ from $\mathcal E_{\alpha}(\U)$ in \Cref{line:g}, then we can see that we then remove $\F_{\e}$ from $\mathcal F$ in \Cref{line:g to F} immediately afterwards.

    \medskip
    \noindent \textit{Condition~\ref{inv:5}:} Let $\F'$ be some Vizing fan that remains in $\mathcal F$ after handling this call. Then, for any Vizing fan $\F''$ that was activated during the call, we know that its corresponding Vizing chain did \emph{not} end at $\F'$, otherwise $\F'$ would have been removed from $\mathcal F$ (see \Cref{lin:29}). Thus, activating $\F''$ (which is vertex-disjoint from $\F'$) does not change the palettes of any vertices in $\F'$. Note that, since the center of $\F'$ is missing $\alpha$, it is not possible for the Vizing chain of $\F''$ to contain any edges in $\F'$ without ending at $\F'$. Hence, the Vizing fans remaining in $\mathcal F$ are still valid Vizing fans and remain vertex-disjoint. To see that they also remain $\U$-avoiding, note that if we add a u-fan $\f$ to $\U$ (see \Cref{lin:12}), the color that $\f$ assigns to a vertex $x \in \f$ was previously not available at $x$ before the start of the call to $\UpdatePath$. Thus, if $x$ is contained in the Vizing fan $\F'$, its color within $\F'$ is not $c_{\f}(x)$.

    \medskip
    \noindent \textit{Condition~\ref{inv:1}:} Let $\F'$ be some Vizing fan that remains in $\mathcal F$ after handling this call. We now show that the first $|P_{\F'}|$ edges of $\VizingP(\F')$ do not change throughout this call. It follows from the argument for Condition~\ref{inv:5} that neither the colors of the edges nor the palettes of the vertices in $\F'$ change during this call. It remains to show that none of the edges in $P_{\F'}$ change color during this call. We can see that $P_{\F'}$ and $P_{\F}$ are edge-disjoint since otherwise $\F'$ would have been removed from $\mathcal F$. Consequently, $P_{\F'}$ is edge-disjoint from any Vizing chain that is activated during the call, and hence none of the edges in $P_{\F'}$ change colors.
\end{proof}

\begin{lemma}
    The subroutine $\ReduceUEdges$ satisfies \Cref{invaraint:iterations}.
\end{lemma}

\begin{proof}
    This follows from the fact that initially $|P_{\F}| = 0$ for all $\F \in \mathcal F$ and that the subroutine calls $\UpdatePath(\F)$ for each $\F \in \mathcal F$ during each round, which increases $|P_{\F}|$ by $1$.
\end{proof}

\begin{lemma}\label{lem:potential for reduce}
    Each time $\ReduceUEdges$ modifies the coloring $\chi$, it removes at most $2$ u-edges from $\E_\alpha(\U)$ and either (1) extends the coloring $\chi$ to one more edge and removes at most one other u-component from $\U$, or (2) adds a u-fan to $\U$.
\end{lemma}

\begin{proof}
    While activating some trivial Vizing fan $\F_{\e} \in \mathcal F$ at the start, we remove $\e$ from $\U$ (and thus also from $\E_{\alpha}(\U)$) along with at most one other u-component in $\U$ that is damaged by this operation (see \Cref{lem:Vfans are safe}). Now, suppose that the subroutine calls $\UpdatePath(\F)$ for some $\F \in \mathcal F$. By considering each case, we can verify that the subroutine either makes no changes to $\chi$ and $\U$ or it removes at most $2$ u-edges from $\E_\alpha(\U)$ and either (1) extends the coloring $\chi$ to one more edge (see Lines~\ref{lin:12}-\ref{lin:14} and \Cref{lin:22}) and removes at most one other u-component from $\U$ (see \Cref{line:g}), or (2) adds a u-fan to $\U$ (see \Cref{lin:12}).
\end{proof}

\medskip
\noindent \textbf{Running Time:} We now show how to implement this subroutine and analyse its running time. We begin with the following claim which shows that $\UpdatePath$ can be implemented efficiently.

\begin{claim}\label{claim:UpdatePathTime}
    Each call to $\UpdatePath(\F)$ for some $\F \in \mathcal F$ that does not remove $\F$ from $\mathcal F$ takes $O(1)$ time. Otherwise, it takes $O(\Delta + L)$ time.
\end{claim}

\begin{proof}
    By Invariants~\ref{invaraint:paths} and \ref{invaraint:iterations}, we know that $P_{\F} = \VizingP(\F)_{\leq L}$ when the subroutine calls $\UpdatePath(\F)$. Thus, updating $P_{\F}$ to $\VizingP(\F)_{\leq L+1}$ only requires computing the next edge in the path, which can be done in $O(1)$ time using our data structures (see \Cref{sec:data structs}).

    If this call does not remove $\F$ from $\mathcal F$, then we know that $P_{\F}$ is not maximal and is also edge-disjoint from the other prefix paths maintained by the subroutine. In this case, we update $S$ in $O(1)$ time and do not modify $\chi$ or $\U$.

    On the other hand, if this call does remove $\F$ from $\mathcal F$, then we need to activate $O(1)$ many Vizing chains of length $O(L)$, which can be done in $O(\Delta + L)$ time. Removing the edges of the corresponding paths from $S$ can also be done in $O(L)$ time. Finally, using the data structures outlined in \Cref{sec:data struc overview}, we can update $\mathcal F$ and $\U$ in $O(1)$ time.
\end{proof}

\noindent
Let $\mathcal F^\star$ denote the subset of Vizing fans that get removed from $\mathcal F$ by the subroutine $\ReduceUEdges$ before it terminates.
For each $\F \in \mathcal F^\star$, let $L_{\F}$ denote the value of $L$ at the time that $\F$ is removed from $\mathcal F$.

\begin{claim}\label{claim:UPtime bound}
    The total time spent handling calls to $\UpdatePath$ is at most
    \begin{equation}\label{eq:updatepath time}
    O(\Delta \lambda_{\alpha}) + O(1) \cdot \sum_{\F \in \mathcal F^\star} L_{\F}.
    \end{equation}
\end{claim}

\begin{proof}
    Let $\F \in \mathcal F^\star$. We can observe that the subroutine calls $\UpdatePath(\F)$ at most $L_{\F}$ times. It follows from \Cref{claim:UpdatePathTime} that the last call takes $O(\Delta + L_{\F})$ time while the rest take $O(1)$ time. Thus, the total time spent handling calls to $\UpdatePath(\F)$ is at most $O(\Delta + L_{\F})$. Summing over each $\F \in \mathcal F^\star$, we get that the total time spent handling calls to $\UpdatePath$ is at most
    $$ \sum_{\F \in \mathcal F^\star} O(\Delta + L_{\F}) \leq O(\Delta \lambda_{\alpha}) + O(1) \cdot \sum_{\F \in \mathcal F^\star} L_{\F}.\qedhere$$
\end{proof}

\noindent
Recall that $m_\alpha$ denotes the number of edges with color $\alpha$ when we first call the subroutine.

\begin{claim}\label{claim:path packing}
    For each $\F \in \mathcal F^\star$, we have that $L_{\F} \leq O(m_{\alpha}/ \lambda_{\alpha})$.
\end{claim}

\begin{proof}
    Let $L_{\max}$ denote the value of $L$ at the start of the final round performed by the subroutine. We now show that $L_{\max} \leq 8(m_{\alpha} + 4) / \lambda_{\alpha}$, which implies the claim.
    At the start of the final round, we know that $|\mathcal F| > \lambda_\alpha/2$, otherwise the subroutine would terminate. Furthermore, it follows from Invariants~\ref{invaraint:paths} and \ref{invaraint:iterations} that, at the start of this round, the alternating paths $\{P_{\F}\}_{\F \in \mathcal F}$ form a collection of at least $\lambda_\alpha/2$ edge-disjoint $\{\alpha, \cdot\}$-alternating paths of length $L_{\max}$ in $G$. Let $T$ denote the total length of these paths. We can observe that $T$ is at most $3 m_\alpha'$, where $m_\alpha'$ is the number of edges that currently have color $\alpha$, since at least a third of the edges in each of these paths has color $\alpha$ (note that any $\{\alpha, \cdot\}$-alternating path of length $k$ has at least $\lfloor k / 2 \rfloor$ edges with color $\alpha$). We can also observe that $m_\alpha' \leq m_{\alpha} + 2\lambda_{\alpha}$ since the number of edges with color $\alpha$ increases by at most $2$ each time we activate a Vizing fan in $\mathcal F$. Thus, it follows that
    $$ \frac{\lambda_\alpha}{2} \cdot L_{\max} \leq T \leq 3m_\alpha' \leq 3(m_{\alpha} + 2\lambda_{\alpha}), $$
    and so $L_{\max} \leq 6 m_\alpha/\lambda_\alpha + 12$.
\end{proof}

\begin{lemma}\label{lem:new running time reduce}
    We can implement $\ReduceUEdges$ to run in $O(m_\alpha + \Delta \lambda_{\alpha})$ time.
\end{lemma}

\begin{proof}
    Activating the trivial Vizing fans in $\mathcal F$ when the subroutine is first called can be done in $O(\Delta \lambda_{\alpha})$ time.
    The running time of the rest of the subroutine is dominated by the time taken to handle the call to $\UpdatePath$.
    Combining Claims~\ref{claim:UPtime bound} and \ref{claim:path packing}, it follows that the total time spent handling calls to $\UpdatePath$ is at most
    $$ O(\Delta \lambda_{\alpha}) + O(1) \cdot \sum_{\F \in \mathcal F^\star} L_{\F} \leq O(\Delta \lambda_{\alpha}) + O(\lambda_\alpha) \cdot O\!\left( \frac{m_{\alpha}}{\lambda_{\alpha}}\right) \leq O(m_\alpha + \Delta \lambda_{\alpha}).\qedhere$$
\end{proof}

\subsection{Analysis of $\ConUFans$: Proof of \Cref{lem:build u-fans}}\label{sec:proof of L2}

The algorithm begins by constructing a separable collection of $\lambda$ u-components $\U$. The following lemma shows that this can be done efficiently.

\begin{lemma}\label{lem:create u-edges}
    Given a graph $G$ and a partial $(\Delta + 1)$-edge coloring $\chi$ of $G$ with $\lambda$ uncolored edges, we can construct a separable collection of $\lambda$ u-edges $\U$ in $O(m)$ time.
\end{lemma}

\begin{proof}
    The algorithm first initializes an empty separable collection $\U = \varnothing$. It then scans through the edges of the graph and retrieves the $\lambda$ edges $e_1, \dots, e_\lambda$ that are uncolored by $\chi$. The algorithm then scans through each of these uncolored edges and, for each $e_i = (u_i, v_i)$, picks a missing color $\alpha_i \in \miss_{\chi}(u_i) \setminus C_{\U}(u_i)$ (see \Cref{claim:missing color}) and adds the u-edge $\e_i := (u_i, v_i, \alpha_i)$ to $\U$.
    It's easy to verify that the resulting collection $\U$ is separable and contains $\lambda$ u-edges.
    Using the data structures outlined in \Cref{sec:data struc overview}, finding such a color and adding a u-edge to $\U$ can be done in $O(1)$ time. Thus, this entire process can be implemented in $O(m)$ time.
\end{proof}

\noindent 
Let $m_\alpha$ denote the number of edges that have color $\alpha$ w.r.t.~the initial coloring $\chi$ when we first call $\ConUFans$.
The following claim shows that we cannot create too many more edges with the color $\alpha$ throughout this sequence of calls to the subroutines $\PruneVFans$ and $\ReduceUEdges$ for each color $\alpha \in [\Delta + 1]$.

\begin{claim}\label{cl:bounded colors}
    For each color $\alpha \in [\Delta + 1]$, we have that at most $m_\alpha + O(\lambda)$ edges have color $\alpha$ throughout the entire run of the algorithm $\ConUFans$.
\end{claim}

\begin{proof}
    This follows from the fact that the number of edges with color $\alpha$ can only increase when we extend the coloring $\chi$ to some uncolored edge. Furthermore, it can only increase by at most $2$ every time this happens. Thus, for any $\alpha \in [\Delta + 1]$, the number of edges with color $\alpha$ can increase by at most $\lambda$ throughout the entire run of the algorithm $\ConUFans$.
\end{proof}

\noindent Let $\lambda_\alpha$ denote the number of $\alpha$-primed u-edges in the initial collection $\U$. Note that the number of such u-edges can only decrease throughout the run of this algorithm.
It follows from Lemmas~\ref{lem:prune} and \ref{lem:new running time reduce} that the total running time of the algorithm $\ConUFans$ across all of these calls to the subroutines $\PruneVFans$ and
$\ReduceUEdges$ is at most
$$\sum_{\alpha \in [\Delta + 1]} \left( O(\lambda_\alpha \Delta) + O(m_\alpha + \lambda + \lambda_\alpha \Delta) \right) \leq O(m + \Delta \lambda).$$

\begin{lemma}
    The algorithm $\ConUFans$ either extends the coloring $\chi$ to at least $\lambda/18$ more edges or returns a separable collection of at least $\lambda/18$ u-fans $\U$.
\end{lemma}

\begin{proof}
    To see why this lemma holds, consider the following three quantities and how they evolve over time throughout the execution of the algorithm $\ConUFans$: The number of edges that the algorithm has extended the coloring to so far, $\Psi_{\texttt{c}}$, the number of u-edges in $\U$, $\Psi_{\texttt{e}} := |\mathcal E|$, and the number of u-fans in $\U$, $\Psi_{\texttt{f}} := |\mathcal U \setminus \mathcal E|$. 
    Immediately after constructing the separable collection of $\lambda$ u-edges $\U$, it holds that $\Psi_{\texttt{c}} = 0$, $\Psi_{\texttt{e}} = \lambda$ and $\Psi_{\texttt{f}} = 0$.
    We next argue that by the time the execution of the algorithm has finished,  either $\Psi_{\texttt{c}} \geq \lambda / 18$ or $\Psi_{\texttt{f}} \geq \lambda / 18$ must hold.
    To this end, we employ the following potential function argument.

    \begin{claim}\label{cl:potential}
        The quantity $3(\Psi_{\textnormal{\texttt{f}}} + 2 \Psi_{\textnormal{\texttt{c}}}) + \Psi_{\textnormal{\texttt{e}}}$ is non-decreasing throughout the algorithm's execution.
    \end{claim}

    \begin{proof}
        Consider first a call made by algorithm $\ConUFans$ to the $\PruneVFans$ subroutine. By \Cref{lem:potential prune}, the subroutine repeatedly removes at most $2$ u-edges from $\U$ and either (1) adds a u-fan to $\U$, or (2) extends the coloring to another edge. This may decrease $\Psi_{\texttt{e}}$ by at most $2$ while increasing one of $\Psi_{\texttt{f}}$ or $\Psi_{\texttt{c}}$ by at least $1$, so $3(\Psi_{\texttt{f}} + 2 \Psi_{\texttt{c}}) + \Psi_{\texttt{e}}$ increases in this case.

        Next, consider a call to the $\ReduceUEdges$ subroutine. By \Cref{lem:potential for reduce},  during any iteration, the subroutine removes at most $3$ u-edges from $\U$ and either (1) adds a u-fan to $\U$, or (2) extends the coloring to another edge while removing at most one u-fan from $\U$. This may decrease $\Psi_{\texttt{e}}$ by at most $3$ while increasing $\Psi_{\texttt{f}} + 2\Psi_{\texttt{c}}$  by at least $1$, so $3(\Psi_{\texttt{f}} + 2 \Psi_{\texttt{c}}) + \Psi_{\texttt{e}}$ cannot decrease in value.
    \end{proof}
\noindent
Immediately after constructing the separable collection of $\lambda$ u-edges in $\U$, we have
$3(\Psi_{\texttt{f}} + 2 \Psi_{\texttt{c}}) + \Psi_{\texttt{e}} = \lambda$.
    \Cref{cl:potential} implies that $3(\Psi_{\texttt{f}} + 2 \Psi_{\texttt{c}}) + \Psi_{\texttt{e}} \geq \lambda$
     holds at all times afterwards. Since we have that $|\mathcal E| \leq \lambda/2$ when the algorithm terminates,
     we know that $\Psi_{\texttt{e}} \leq \lambda/2$ at this time, so at that moment we have $3(\Psi_{\texttt{f}} + 2 \Psi_{\texttt{c}}) \geq \lambda/2$. Hence, either $\Psi_{\texttt{f}} \geq \lambda/18$ or $\Psi_{\texttt{c}} \geq \lambda/18$ must hold.
\end{proof}

\section{The Algorithm $\ColorUFans$: Proof of \Cref{lem:color u-fans}}\label{sec:algo}

As input, the algorithm $\ColorUFans$ is given a graph $G$, a partial $(\Delta + 1)$-edge coloring $\chi$ of $G$, and a collection of separable u-fans $\mathcal U$ of size $\lambda$. It then uses two further subroutines, $\PrimeUFans$ and $\ActUFans$, to prime u-fans with the same colors and then activate them.
More specifically, the algorithm first identifies the two least common colors $\alpha, \beta \in [\Delta + 1]$.
It calls the subroutine $\PrimeUFans$ which proceeds to prime $\Omega(\lambda / \Delta)$ of the u-fans in $\U$ with the colors $\alpha$ and $\beta$ (i.e.~to modify the coloring $\chi$ and these u-fans so that they are $\{\alpha, \beta\}$-primed). It then calls $\ActUFans$ which extends the coloring $\chi$ to $\Omega(\lambda /\Delta)$ of the edges in these u-fans. The algorithm repeats this process until it has colored $\Omega(\lambda)$ edges. \Cref{alg:colorufansfinal} gives the pseudocode for $\ColorUFans$.

\begin{algorithm}[H]
    \SetAlgoLined
    \DontPrintSemicolon
    \SetKwRepeat{Do}{do}{while}
    \SetKwBlock{Loop}{repeat}{EndLoop}
    \For{$\Delta/2$ \textbf{\textup{iterations}}}{
        Let $\alpha, \beta \in [\Delta + 1]$ be the two least common colors in $\chi$\;
        $\PrimeUFans(\U, \alpha, \beta)$\;
        $\ActUFans(\U, \alpha, \beta)$\;
    }
    \caption{$\ColorUFans(\mathcal U)$}
    \label{alg:colorufansfinal}
\end{algorithm}

\medskip
\noindent \textbf{Organization of \Cref{sec:algo}:}
We begin by describing and analyzing the subroutines $\PrimeUFans$ and $\ActUFans$ used by $\ColorUFans$ before proving \Cref{lem:color u-fans} in \Cref{sec:proof of L1}.

\subsection{The Subroutine $\PrimeUFans$}\label{sec:primeUfans}

As input, this subroutine is given a graph $G$, a partial $(\Delta + 1)$-edge coloring $\chi$ of $G$, a collection of separable u-fans $\mathcal U$ of size $\lambda$, and the two least common colors $\alpha, \beta \in [\Delta + 1]$ (see \Cref{cl:small colors}).
It then proceeds in \emph{iterations}, where in each iteration it samples a u-fan $\f$ from $\mathcal U$ uniformly at random and attempts to prime $\f$ with the colors $\alpha$ and $\beta$.
The subroutine maintains a subset $\Phi \subseteq \U$ of $\{\alpha, \beta\}$-primed u-fans.
The subroutine performs iterations until $|\Phi| = \Omega(\lambda / \Delta)$, after which we proceed to call the subroutine $\ActUFans$ to extend the coloring $\chi$ to edges in the u-fans in $\Phi$.
The pseudocode in \Cref{alg:colorufans} gives a formal description of the subroutine $\PrimeUFans$.

\begin{algorithm}[H]
    \SetAlgoLined
    \DontPrintSemicolon
    \SetKwRepeat{Do}{do}{while}
    \SetKwBlock{Loop}{repeat}{EndLoop}
    $\Phi \leftarrow \varnothing$ and $\lambda \leftarrow |\mathcal U|$\;
    \While{$|\Phi| < \lambda / (48\Delta)$}{\label{line:while loop}
        Sample $\f = (u, v, w, \gamma, \delta) \sim U$ independently and u.a.r.\;
        $(\alpha', \beta') \leftarrow (\alpha, \beta)$ \textit{\texttt{ ~~~~// $\alpha'$ and $\beta'$ ensure we never flip an $\{\alpha, \beta\}$-alt.~path}}\;
        \If{$\gamma = \beta$ or $\delta = \alpha$}{
            $(\alpha', \beta') \leftarrow (\beta, \alpha)$\;
        }
        Let $P_u$ be the $\{\alpha', \gamma\}$-alternating path starting at $u$ \;
        Let $P_v$ and $P_w$ be the $\{\beta', \delta\}$-alternating paths starting at $v$ and $w$ respectively\;
        $\cost(\f) := |P_u| + |P_v| + |P_w|$\;
        \If{$\cost(\f) \leq 128m/\lambda$\label{line: if 1}}{
            Let $S$ denote the set of endpoints of $P_u$, $P_v$ and $P_w$\;
            \If{the vertices in $S\cup \{u,v,w\}$ are not in any u-fan in $\Phi$\label{line: if 2}}{
                Flip the alternating paths $P_u$, $P_v$ and $P_w$\label{line:colorchange2}\;
                Remove $\f$ and any damaged u-fans from $\mathcal U$\;
                Add the u-fan $(u,v,w,\alpha',\beta')$ to $\U$ and $\Phi$\;
            }
        }
    }
    \Return $(\alpha, \beta)$\;
    \caption{$\PrimeUFans(\mathcal U)$}
    \label{alg:colorufans}
\end{algorithm}

\subsubsection*{Analysis of $\PrimeUFans$}\label{sec:analprimeUfans}

We say that an iteration of (the \textup{\textbf{while}} loop in) $\PrimeUFans$ is \emph{successful} if it adds a u-fan to the set $\Phi$. Otherwise, we say that the iteration \emph{fails} (see \Cref{socialize-flip} for an illustration). Note that the algorithm repeatedly performs iterations until it has performed at least $\lambda/(48 \Delta)$ successful iterations.
The following lemmas summarize the main properties of the subroutine $\PrimeUFans$.

\begin{figure}
	\centering
	\begin{tikzpicture}[thick,scale=0.8]
	\draw (-1, 1) node(1)[circle, draw, color=cyan, fill=black!50,
	inner sep=0pt, minimum width=10pt, label = $u$] {};
	
	\draw (1, 0) node(2)[circle, draw, color=teal, fill=black!50,
	inner sep=0pt, minimum width=10pt, label = $v$] {};
	
	\draw (1, 2) node(3)[circle, draw, color=teal, fill=black!50,
	inner sep=0pt, minimum width=10pt, label = $w$] {};
	
	\draw (-3, 1) node(4)[circle, draw, fill=black!50,
	inner sep=0pt, minimum width=6pt] {};
	
	\draw (-5, 1) node(5)[circle, draw, fill=black!50,
	inner sep=0pt, minimum width=6pt] {};
	
	\draw (-7, 1) node(6)[circle, draw, fill=black!50,
	inner sep=0pt, minimum width=6pt] {};
	
	\draw (3, 0) node(7)[circle, draw, fill=black!50,
	inner sep=0pt, minimum width=6pt] {};
	
	\draw (5, 0) node(8)[circle, draw, fill=black!50,
	inner sep=0pt, minimum width=6pt] {};
	
	\draw (7, 0) node(9)[circle, draw, fill=black!50,
	inner sep=0pt, minimum width=6pt] {};
	
	\draw (3, 2) node(10)[circle, draw, fill=black!50,
	inner sep=0pt, minimum width=6pt] {};
	
	\draw (5, 2) node(11)[circle, draw, fill=black!50,
	inner sep=0pt, minimum width=6pt] {};
	
	\draw (7, 2) node(12)[circle, draw, fill=black!50,
	inner sep=0pt, minimum width=6pt] {};
	
	\draw (9, 0) node(13)[circle, draw, color=orange, fill=black!50, inner sep=0pt, minimum width=10pt, label = $v'$] {};
	\draw (11, -1) node(14)[circle, draw, color=red, fill=black!50,
	inner sep=0pt, minimum width=10pt, label = $u'$] {};
	\draw (9, -2) node(15)[circle, draw, color=orange, fill=black!50,
	inner sep=0pt, minimum width=10pt, label = $w'$] {};
	
	\draw [gray!50] plot [smooth cycle] coordinates {(-2, 1) (1, 3) (2, 1) (1, -1)};
	\node at (-1, 3) {u-fan $\f$};
	\draw [gray!50] plot [smooth cycle] coordinates {(12, -1) (9, 1) (8, -1) (9, -3)};
	\node at (11, -3) {u-fan $\f'\in \Phi$};
	
	\draw [line width = 0.5mm, dashed] (1) to (2);
	\draw [line width = 0.5mm, dashed] (1) to (3);
	\draw [line width = 0.5mm, color=red] (1) to node[above] {$\alpha$} (4);
	\draw [line width = 0.5mm, color=cyan] (4) to node[above] {$\gamma$} (5);
	\draw [line width = 0.5mm, color=red] (5) to node[above] {$\alpha$} (6);
	\draw [line width = 0.5mm, color=orange] (2) to node[above] {$\beta$} (7);
	\draw [line width = 0.5mm, color=teal] (7) to node[above] {$\delta$} (8);
	\draw [line width = 0.5mm, color=orange] (8) to node[above] {$\beta$} (9);
	\draw [line width = 0.5mm, color=teal] (9) to node[above] {$\delta$} (13);
	\draw [line width = 0.5mm, color=orange] (3) to node[above] {$\beta$} (10);
	\draw [line width = 0.5mm, color=teal] (10) to node[above] {$\delta$} (11);
	\draw [line width = 0.5mm, color=orange] (11) to node[above] {$\beta$} (12);
	
	\draw [line width = 0.5mm, dashed] (13) to (14);
	\draw [line width = 0.5mm, dashed] (14) to (15);
	
        \draw[->, >={Triangle}, thick, line width = 0.9mm] (1, -2) to (1, -4);
		
	\draw (-1, -7) node(21)[circle, draw, color=red, fill=black!50,
	inner sep=0pt, minimum width=10pt, label = $u$] {};
	
	\draw (1, -8) node(22)[circle, draw, color=orange, fill=black!50,
	inner sep=0pt, minimum width=10pt, label = $v$] {};
	
	\draw (1, -6) node(23)[circle, draw, color=orange, fill=black!50,
	inner sep=0pt, minimum width=10pt, label = $w$] {};
	
	\draw (-3, -7) node(24)[circle, draw, fill=black!50,
	inner sep=0pt, minimum width=6pt] {};
	
	\draw (-5, -7) node(25)[circle, draw, fill=black!50,
	inner sep=0pt, minimum width=6pt] {};
	
	\draw (-7, -7) node(26)[circle, draw, fill=black!50,
	inner sep=0pt, minimum width=6pt] {};
	
	\draw (3, -8) node(27)[circle, draw, fill=black!50,
	inner sep=0pt, minimum width=6pt] {};
	
	\draw (5, -8) node(28)[circle, draw, fill=black!50,
	inner sep=0pt, minimum width=6pt] {};
	
	\draw (7, -8) node(29)[circle, draw, fill=black!50,
	inner sep=0pt, minimum width=6pt] {};
	
	\draw (3, -6) node(30)[circle, draw, fill=black!50,
	inner sep=0pt, minimum width=6pt] {};
	
	\draw (5, -6) node(31)[circle, draw, fill=black!50,
	inner sep=0pt, minimum width=6pt] {};
	
	\draw (7, -6) node(32)[circle, draw, fill=black!50,
	inner sep=0pt, minimum width=6pt] {};
	
	\draw (9, -8) node(33)[circle, draw, color=orange, fill=black!50, inner sep=0pt, minimum width=10pt, label = $v'$] {};
	\draw (11, -9) node(34)[circle, draw, color=red, fill=black!50,
	inner sep=0pt, minimum width=10pt, label = $u'$] {};
	\draw (9, -10) node(35)[circle, draw, color=orange, fill=black!50,
	inner sep=0pt, minimum width=10pt, label = $w'$] {};
	
	\draw [gray!50] plot [smooth cycle] coordinates {(-2, -7) (1, -5) (2, -7) (1, -9)};
	\node at (-1, -5) {u-fan $\f$};
	\draw [gray!50] plot [smooth cycle] coordinates {(12, -9) (9, -7) (8, -9) (9, -11)};
	\node at (11, -11) {u-fan $\f'\in \Phi$};
	
	\draw [line width = 0.5mm, dashed] (21) to (22);
	\draw [line width = 0.5mm, dashed] (21) to (23);
	\draw [line width = 0.5mm, color=cyan] (21) to node[above] {$\gamma$} (24);
	\draw [line width = 0.5mm, color=red] (24) to node[above] {$\alpha$} (25);
	\draw [line width = 0.5mm, color=cyan] (25) to node[above] {$\gamma$} (26);
	\draw [line width = 0.5mm, color=teal] (22) to node[above] {$\delta$} (27);
	\draw [line width = 0.5mm, color=orange] (27) to node[above] {$\beta$} (28);
	\draw [line width = 0.5mm, color=teal] (28) to node[above] {$\delta$} (29);
	\draw [line width = 0.5mm, color=orange] (29) to node[above] {$\beta$} (33);
	\draw [line width = 0.5mm, color=teal] (23) to node[above] {$\delta$} (30);
	\draw [line width = 0.5mm, color=orange] (30) to node[above] {$\beta$} (31);
	\draw [line width = 0.5mm, color=teal] (31) to node[above] {$\delta$} (32);
	
	\draw [line width = 0.5mm, dashed] (33) to (34);
	\draw [line width = 0.5mm, dashed] (34) to (35);
	
	\draw (8.5, -8.5) node[cross=6, red, label=225:{\text{violating $c_{v'}=\beta$}}] {};
	
\end{tikzpicture}
	\caption{In this picture, $\alpha$ is red, $\beta$ is orange, $\gamma$ is blue, and $\delta$ is green. $(u', v', w')\in \Phi$ is an existing $\{\alpha, \beta\}$-primed u-fan in $\Phi$. If we flip the $\{\alpha, \gamma\}$-alternating path at $u$ and the $\{\beta, \delta\}$-alternating paths at $v$ and $w$, then we would destroy the property that $\beta\in \miss_{\chi}(v')$.}\label{socialize-flip}
\end{figure}
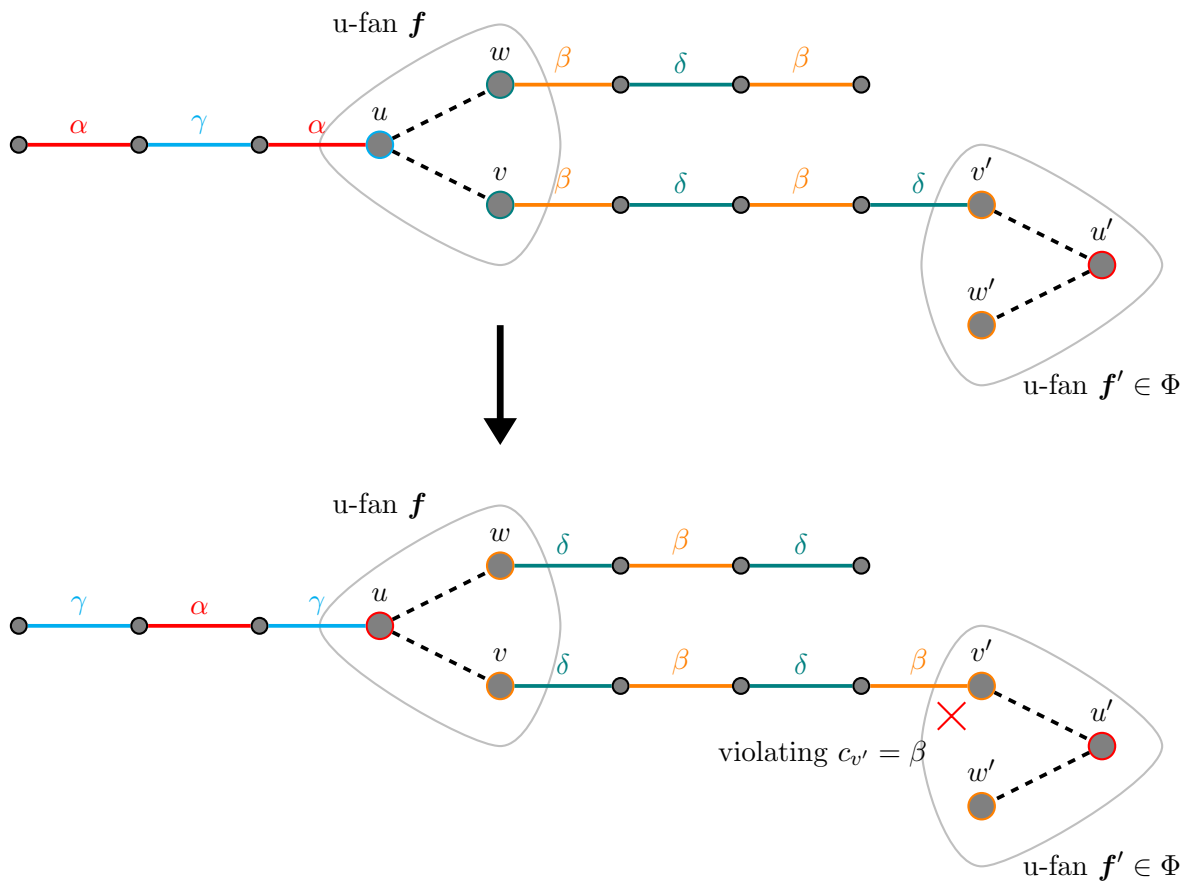

\begin{claim}\label{cl:small colors}
    The total number of edges with color $\alpha$ or $\beta$ is $O(m/\Delta)$ at any point throughout the run of $\PrimeUFans$.
\end{claim}

\begin{proof}
    Since $\alpha$ and $\beta$ are initially the two least common colors, there are initially $O(m/\Delta)$ such edges.
    Since each successful iteration increases the number of such edges by $O(1)$ and there are at most $O(\lambda/\Delta) \leq  O(m/\Delta)$ such iterations, the claim follows.
\end{proof}

\begin{claim}\label{cl:big U}
    $|\mathcal U| \geq (1 - 1/(2\Delta)) \lambda$ at any point throughout the run of $\PrimeUFans$.
\end{claim}

\begin{proof}
    Consider a successful iteration of the algorithm where we sample a u-fan $\f \in \mathcal U$. The algorithm removes $\f$ from $\mathcal U$, along with any other damaged u-fans in $\mathcal U$. 
    It follows from \Cref{lem:low damage flips} that flipping the colors of the alternating paths $P_u$, $P_v$ and $P_w$ damages at most $6$ u-fans in $\U$.
    Since our algorithm runs for at most $\lambda/(48 \Delta)$ iterations, it follows that $|\mathcal U| \geq \lambda - 7 \cdot \lambda/(48 \Delta)$.
\end{proof}

\begin{lemma}\label{lem:bipartite:0}
    The u-fans in $\Phi$ are all vertex-disjoint and $\{\alpha, \beta\}$-primed at any point throughout the run of $\PrimeUFans$.
\end{lemma}

\begin{proof}
    If the algorithm samples a u-fan $\f$ that shares a vertex $x$ with a u-fan $\f' \in \Phi$, then the iteration fails. Thus, the u-fans in $\Phi$ are all vertex-disjoint.
    
    Now, note that whenever we add a u-fan $(u,v,w,\alpha',\beta')$ to $\Phi$, we do this immediately after flipping the paths $P_u$, $P_v$, and $P_w$. Thus, we have that $\alpha' \in \miss_\chi(u)$ and $\beta' \in \miss_\chi(v) \cap \miss_\chi(w)$, so the u-fan is $\{\alpha, \beta\}$-primed when we add it to $\Phi$.
    Since we only flip the paths in \Cref{line:colorchange2} if their endpoints do not touch any u-fans in $\Phi$, this operation cannot change what colors are available at vertices in u-fans in $\Phi$, and hence cannot change whether or not any u-fan in $\Phi$ is $\{\alpha, \beta\}$-primed.
\end{proof}

\noindent The following standard claim bounds the total length of all maximal $\{c, \cdot\}$-alternating paths.

\begin{claim}\label{claim:total ap bound}
    For any color $c \in [\Delta + 1]$, the total length of all maximal $\{c, \cdot\}$-alternating paths in $\chi$ is at most $4m$.
\end{claim}

\begin{proof}
    Let $\mathcal P_c$ denote the set of all such alternating paths. First note that the total length of all alternating paths in $\mathcal P_c$ with length $1$ is at most $m$.
    Now, let $c' \in [\Delta + 1] \setminus \{c\}$ and let $P$ be a maximal $\{c,c'\}$-alternating path with $|P| \geq 2$. We can observe that at least a third of the edges in $P$ have color $c'$, and that each edge with color $c'$ only appears in one path in $\mathcal P_c$. Thus, we have that
    $$\sum_{P \in \mathcal P_c} |P| \leq m + 3 \cdot |\{e \in E \mid \chi(e) \neq c\}| \leq 4m.\qedhere$$
\end{proof}

\begin{lemma}\label{lem:bipartite:1}
    Each iteration of $\PrimeUFans$ succeeds with probability at least $1/4$.
\end{lemma}

\begin{proof}
    Consider the state of the algorithm at the start of some iteration. For each u-fan $\f' \in \U$, let $P_x(\f')$ denote the alternating path starting at $x \in \f'$ that is considered by the algorithm if it samples $\f' \in \U$. We define the \emph{cost} of the u-fan $\f'$ to be 
    $$ \cost(\f{\!'}) := \sum_{x \in \f'} |P_x(\f{\!'})|. $$
    Let $\mathcal U^\star \subseteq \mathcal U$ denote the subset of u-fans $\f' \in \U$ such that none of the alternating paths $\{P_{x}(\f')\}_{x \in \f'}$ have endpoints at some u-fan in $\Phi$. We can see that the iteration is successful if and only if the u-fan $\f \in \U$ sampled during the iteration satisfies $\cost(\f) \leq 128m/\lambda$ (see \Cref{line: if 1}) and $\f \in \U^\star$ (see \Cref{line: if 2}). Thus, we now show that this happens with probability at least $1/4$.
        
    We first begin with the following claim.

    \begin{claim}\label{cl:1P2F}
        Let $P$ be an $\{\alpha, \cdot\}$- or $\{\beta, \cdot\}$-alternating path in $\chi$. 
        Then there are at most $2$ u-fans $\f' \in \mathcal U$ that, if sampled during the iteration, might cause the algorithm to flip $P$.
    \end{claim}

    \begin{proof}
        Suppose that the path $P$ is a $\{c, c'\}$-alternating path for some colors $c \in \{\alpha, \beta\}$ and $c' \notin \{\alpha, \beta\}$ with endpoints $x$ and $y$.\footnote{Note that the algorithm never flips $\{\alpha, \beta\}$-alternating paths.} Since the collection of u-fans $\mathcal U$ is separable, we know that at most one u-fan in $\mathcal U$ containing $x$ (resp.~$y$) has $c'$ assigned as the color missing at $x$ (resp.~$y$). Thus, we must sample one of these $2$ u-fans for the algorithm to flip the path $P$.
    \end{proof}

    \noindent
    Let $\mathcal P$ denote the set of all maximal $\{\alpha, \cdot\}$- and $\{\beta, \cdot\}$-alternating paths in $\chi$.
    We can see that the expected cost of the u-fan $\f$ sampled during the iteration is
    $$ \mathbb E[\cost(\f)] = \frac{1}{|\mathcal U|} \cdot \sum_{\f \in \U} \sum_{x \in \f} |P_x(\f)| \leq \frac{2}{|\mathcal U|} \cdot \sum_{P \in \mathcal P} |P| \leq \frac{32m}{\lambda}, $$
    where we are using the facts that (1) each path in $P \in \mathcal P$ appears at most twice while summing over the paths $|P_x(\f)|$ by \Cref{cl:1P2F}, (2) that $|\mathcal U| \geq \lambda/2$ by \Cref{cl:big U}, and (3) that that the total length of all paths in $\mathcal P$ is $8m$ by \Cref{claim:total ap bound}.
    Applying Markov's inequality, it follows that
    $$ \Pr \! \left[\cost(\f) \geq \frac{128 m}{\lambda} \right] \leq \frac{1}{4}. $$

    Since each u-fan in $\Phi$ has 3 vertices and at most $2\Delta$ $\{\alpha, \cdot\}$- or $\{\beta, \cdot\}$-alternating paths end at each of these vertices, we know that there are at most $6\Delta |\Phi|$ alternating paths that could be considered during some iteration that end at a u-fan in $\Phi$. By \Cref{cl:1P2F}, there are at most $12\Delta |\Phi|$ u-fans
    in $\mathcal U$ that could cause the algorithm to consider one of these paths (see \Cref{socialize-damage} for an illustration).
    Thus, at least $|\mathcal U| - 12\Delta |\Phi|$ of the u-fans in $\mathcal U$ are contained in $\mathcal U^\star$. It follows that
    $$ |\mathcal U^\star| \geq |\mathcal U| - 12\Delta |\Phi| \geq |\mathcal U| - \frac{\lambda}{4} \geq \frac{|\mathcal U|}{2}, $$
    where we are using the facts that $|\Phi| \leq \lambda /(48\Delta)$ and $|\mathcal U| \geq \lambda / 2$ from \Cref{cl:big U}. Since $\f$ is sampled uniformly at random from $\U$, it follows that
    $\Pr [\f \in \mathcal U^\star] \geq 1/2$. The lemma follows by applying a union bound.

    \begin{figure}
    	\centering
    	\begin{tikzpicture}[thick,scale=0.7]
	\draw (7, -1.5) node(0)[circle, draw, color = red, fill=black!50,
	inner sep=0pt, minimum width=10pt, label = $u'$] {};
	\draw (5, 0) node(1)[circle, draw, color=orange, fill=black!50,
	inner sep=0pt, minimum width=10pt, label = $v'$] {};
	\draw (5, -3) node(2)[circle, draw, color=orange, fill=black!50, inner sep=0pt, minimum width=10pt, label = -90:{$w'$}] {};	
	
	\draw (3, -1) node(3)[circle, draw, fill=black!50,
	inner sep=0pt, minimum width=6pt] {};
	\draw (1, -1) node(4)[circle, draw, fill=black!50,
	inner sep=0pt, minimum width=6pt] {};
	\draw (-1, -1) node(5)[circle, draw, fill=black!50,
	inner sep=0pt, minimum width=6pt] {};
	\draw (-3, -1) node(6)[circle, draw, fill=black!50, 
	inner sep=0pt, minimum width=6pt] {};
	
	\draw (3, 1) node(13)[circle, draw, fill=black!50,
	inner sep=0pt, minimum width=6pt] {};
	\draw (1, 1) node(14)[circle, draw, fill=black!50,
	inner sep=0pt, minimum width=6pt] {};
	\draw (-1, 1) node(15)[circle, draw, fill=black!50,
	inner sep=0pt, minimum width=6pt] {};
	\draw (-3, 1) node(16)[circle, draw, fill=black!50, 
	inner sep=0pt, minimum width=6pt] {};
	
	\draw (3, 0) node(23)[circle, draw, fill=black!50,
	inner sep=0pt, minimum width=6pt] {};
	\draw (1, 0) node(24)[circle, draw, fill=black!50,
	inner sep=0pt, minimum width=6pt] {};
	\draw (-1, 0) node(25)[circle, draw, fill=black!50,
	inner sep=0pt, minimum width=6pt] {};
	\draw (-3, 0) node(26)[circle, draw, fill=black!50, 
	inner sep=0pt, minimum width=6pt] {};	
	
	\draw [line width = 0.5mm, color=olive] (1) to (3);
	\draw [line width = 0.5mm, color=orange] (3) to node[above] {$\beta$} (4);
	\draw [line width = 0.5mm, color=olive] (4) to (5);
	\draw [line width = 0.5mm, color=orange] (5) to node[above] {$\beta$} (6);
	
	\draw [line width = 0.5mm, color=teal] (1) to (23);
	\draw [line width = 0.5mm, color=orange] (23) to node[above] {$\beta$} (24);
	\draw [line width = 0.5mm, color=teal] (24) to (25);
	\draw [line width = 0.5mm, color=orange] (25) to node[above] {$\beta$} (26);
	
	\draw [line width = 0.5mm, color=lime] (1) to (13);
	\draw [line width = 0.5mm, color=orange] (13) to node[above] {$\beta$} (14);
	\draw [line width = 0.5mm, color=lime] (14) to (15);
	\draw [line width = 0.5mm, color=orange] (15) to node[above] {$\beta$} (16);
	
	\draw (3, -2) node(33)[circle, draw, fill=black!50,
	inner sep=0pt, minimum width=6pt] {};
	\draw (1, -2) node(34)[circle, draw, fill=black!50,
	inner sep=0pt, minimum width=6pt] {};
	\draw (-1, -2) node(35)[circle, draw, fill=black!50,
	inner sep=0pt, minimum width=6pt] {};
	\draw (-3, -2) node(36)[circle, draw, fill=black!50, 
	inner sep=0pt, minimum width=6pt] {};
	
	\draw (3, -3) node(43)[circle, draw, fill=black!50,
	inner sep=0pt, minimum width=6pt] {};
	\draw (1, -3) node(44)[circle, draw, fill=black!50,
	inner sep=0pt, minimum width=6pt] {};
	\draw (-1, -3) node(45)[circle, draw, fill=black!50,
	inner sep=0pt, minimum width=6pt] {};
	\draw (-3, -3) node(46)[circle, draw, fill=black!50, 
	inner sep=0pt, minimum width=6pt] {};
	
	\draw (3, -4) node(53)[circle, draw, fill=black!50,
	inner sep=0pt, minimum width=6pt] {};
	\draw (1, -4) node(54)[circle, draw, fill=black!50,
	inner sep=0pt, minimum width=6pt] {};
	\draw (-1, -4) node(55)[circle, draw, fill=black!50,
	inner sep=0pt, minimum width=6pt] {};
	\draw (-3, -4) node(56)[circle, draw, fill=black!50, 
	inner sep=0pt, minimum width=6pt] {};
	
	\draw [line width = 0.5mm, color=Lavender] (2) to (33);
	\draw [line width = 0.5mm, color=orange] (33) to node[above] {$\beta$} (34);
	\draw [line width = 0.5mm, color=Lavender] (34) to (35);
	\draw [line width = 0.5mm, color=orange] (35) to node[above] {$\beta$} (36);
	
	\draw [line width = 0.5mm, color=Orchid] (2) to (43);
	\draw [line width = 0.5mm, color=orange] (43) to node[above] {$\beta$} (44);
	\draw [line width = 0.5mm, color=Orchid] (44) to (45);
	\draw [line width = 0.5mm, color=orange] (45) to node[above] {$\beta$} (46);
	
	\draw [line width = 0.5mm, color=Plum] (2) to (53);
	\draw [line width = 0.5mm, color=orange] (53) to node[above] {$\beta$} (54);
	\draw [line width = 0.5mm, color=Plum] (54) to (55);
	\draw [line width = 0.5mm, color=orange] (55) to node[above] {$\beta$} (56);
	
	\draw (9, -0.5) node(63)[circle, draw, fill=black!50,
	inner sep=0pt, minimum width=6pt] {};
	\draw (11, -0.5) node(64)[circle, draw, fill=black!50,
	inner sep=0pt, minimum width=6pt] {};
	\draw (13, -0.5) node(65)[circle, draw, fill=black!50,
	inner sep=0pt, minimum width=6pt] {};
	\draw (15, -0.5) node(66)[circle, draw, fill=black!50, 
	inner sep=0pt, minimum width=6pt] {};
	
	\draw (9, -1.5) node(73)[circle, draw, fill=black!50,
	inner sep=0pt, minimum width=6pt] {};
	\draw (11, -1.5) node(74)[circle, draw, fill=black!50,
	inner sep=0pt, minimum width=6pt] {};
	\draw (13, -1.5) node(75)[circle, draw, fill=black!50,
	inner sep=0pt, minimum width=6pt] {};
	\draw (15, -1.5) node(76)[circle, draw, fill=black!50, 
	inner sep=0pt, minimum width=6pt] {};
	
	\draw (9, -2.5) node(83)[circle, draw, fill=black!50,
	inner sep=0pt, minimum width=6pt] {};
	\draw (11, -2.5) node(84)[circle, draw, fill=black!50,
	inner sep=0pt, minimum width=6pt] {};
	\draw (13, -2.5) node(85)[circle, draw, fill=black!50,
	inner sep=0pt, minimum width=6pt] {};
	\draw (15, -2.5) node(86)[circle, draw, fill=black!50, 
	inner sep=0pt, minimum width=6pt] {};
	
	\draw [line width = 0.5mm, color=cyan] (0) to (63);
	\draw [line width = 0.5mm, color=red] (63) to node[above] {$\alpha$} (64);
	\draw [line width = 0.5mm, color=cyan] (64) to (65);
	\draw [line width = 0.5mm, color=red] (65) to node[above] {$\alpha$} (66);
	
	\draw [line width = 0.5mm, color=NavyBlue] (0) to (73);
	\draw [line width = 0.5mm, color=red] (73) to node[above] {$\alpha$} (74);
	\draw [line width = 0.5mm, color=NavyBlue] (74) to (75);
	\draw [line width = 0.5mm, color=red] (75) to node[above] {$\alpha$} (76);
	
	\draw [line width = 0.5mm, color=Blue] (0) to (83);
	\draw [line width = 0.5mm, color=red] (83) to node[above] {$\alpha$} (84);
	\draw [line width = 0.5mm, color=Blue] (84) to (85);
	\draw [line width = 0.5mm, color=red] (85) to node[above] {$\alpha$} (86);
	
	\draw [line width = 0.5mm, dashed] (0) to (1);
	\draw [line width = 0.5mm, dashed] (0) to (2);
	
	\draw [gray!50] plot [smooth cycle] coordinates {(8, -1.5) (5, 1) (4, -1.5) (5, -4)};
	\node at (7, -4) {u-fan $\f'\in \Phi$};
	
\end{tikzpicture}
    	\caption{In this picture, $\alpha$ is red, $\beta$ is orange, and $\f' = (u', v', w', \alpha, \beta)\in \Phi$ is a u-fan, and we have drawn $9$ different $\{\alpha, \cdot\}$- or $\{\beta, \cdot\}$-alternating paths starting from $\f'$.}\label{socialize-damage}
    \end{figure}
\end{proof}

\begin{lemma}\label{lem:bipartite:2}
    Each iteration of $\PrimeUFans$ takes time $O(m/ \lambda)$.
\end{lemma}

\begin{proof}  
    Using standard data structures, each iteration can be implemented in time proportional to the length of the alternating paths $P_u$, $P_v$ and $P_w$ considered by the algorithm during the iteration.
    We can check if $|P_u| + |P_v| + |P_w| \leq 128m/\lambda$ in $O(m/\lambda)$ time by traversing these paths and aborting if we notice that their total length exceeds $128m/\lambda$. If their total length is at most $128m/\lambda$, then we can flip these paths and update the collection $\U$ in $O(m /\lambda)$ time.
\end{proof}

\begin{lemma}\label{lemma:prime time}
    The subroutine $\PrimeUFans$ runs in time $O(m \log n/\Delta)$ with high probability. 
\end{lemma}

\begin{proof}
    Since each iteration of $\PrimeUFans$ succeeds with probability at least $1/4$ by \Cref{lem:bipartite:1} and the subroutine performs iterations until it succeeds $\lambda/(48\Delta)$ times, it follows that it performs at most $O(\lambda \log n / \Delta)$ iterations with high probability. Since each iteration takes $O(m /\lambda)$ time by \Cref{lem:bipartite:2}, it follows that the total running time is $O(m \log n/\Delta)$ with high probability.
\end{proof}

\subsection{The Subroutine $\ActUFans$}\label{sec:actUfans}

As input, this subroutine is given a graph $G$, a partial $(\Delta + 1)$-edge coloring $\chi$ of $G$, and a subset $\Phi \subseteq \U$ of $\mu$ vertex-disjoint $\{\alpha, \beta\}$-primed u-fans such that at most $O(m/\Delta)$ edges have color $\alpha$ or $\beta$.
The subroutine repeatedly picks any $\f \in \Phi$ and proceeds to activate the u-fan. It repeats this process until $\Phi = \varnothing$. The pseudocode in \Cref{alg:actufans} gives a formal description of the subroutine.

\begin{algorithm}[H]
    \SetAlgoLined
    \DontPrintSemicolon
    \SetKwRepeat{Do}{do}{while}
    \SetKwBlock{Loop}{repeat}{EndLoop}
    \While{$\Phi \neq \varnothing$}{
        Let $\f \in \Phi$\label{line:fan}\;
        Let $P$ and $P'$ be the $\{\alpha, \beta\}$-alternating paths starting at the leaves of $\f$\label{line:paths}\;
        Activate the u-fan $\f$ by flipping the path $P$ or $P'$\label{line:actpaths}\;
        Remove $\f$ and any damaged u-fans from $\U$ and $\Phi$\label{line:removedead}\;
    }
    \caption{$\ActUFans(\mathcal U, \alpha, \beta)$}
    \label{alg:actufans}
\end{algorithm}

\subsubsection*{Analysis of $\ActUFans$}

The following lemmas summarise the key properties of the subroutine $\ActUFans$.

\begin{lemma}\label{lem:activation number}
    The subroutine $\ActUFans$ extends the coloring $\chi$ to at least $\mu/2$ more edges.
\end{lemma}

\begin{proof}
    During each iteration of the \textup{\textbf{while}} loop, we activate a u-fan $\f$ and extend the coloring to an uncolored edge in $\f$. By \Cref{lem:low damage flips}, this process damages at most $2$ u-fans in $\U$ (including $\f$). Thus, in each iteration $|\Phi|$ decreases by at most $2$, and hence we extend the coloring to at least $\mu/2$ more edges across all the iterations performed by the subroutine.
\end{proof}

\begin{lemma}\label{lem:actUfanTime}
    The subroutine $\ActUFans$ has a running time of $O(m / \Delta)$.
\end{lemma}

\begin{proof}
    Let $\mathcal P$ denote the collection of maximal $\{\alpha, \beta\}$-alternating paths in the coloring $\chi$ when we first call $\ActUFans$. Since at most $O(m/\Delta)$ edges have color $\alpha$ and $\beta$ in this coloring, we have that the total length of the paths in $\mathcal P$ is $O(m/\Delta)$ since the paths in $\mathcal P$ are all vertex-disjoint.

    Throughout the run of $\ActUFans$, we only modify $\chi$ by flipping the colors of $\{\alpha,\beta\}$-alternating paths (see \Cref{line:actpaths}). Thus, the structure of the paths in $\mathcal P$ does not change (but their colors might be flipped).
    Furthermore, each time the subroutine flips an $\{\alpha,\beta\}$-alternating path, it removes any u-fan from $\Phi$ that contains an endpoint of this alternating path (see \Cref{line:removedead}). Thus, each path in $\mathcal P$ is only flipped at most once.

    Using standard data structures, each iteration of the \textup{\textbf{while}} loop can be implemented in time proportional to the length of the alternating path that is flipped during the iteration (see \Cref{line:actpaths}). It follows that the total running time of the subroutine across all iterations is at most $\sum_{P \in \mathcal P} O(|P|) \leq O(m /\Delta)$.
\end{proof}

\begin{claim}\label{cl:big U 2}
    $|\mathcal U|$ decreases by at most $2\mu$ throughout the run of $\ActUFans$.
\end{claim}

\begin{proof}
During each iteration, we flip an alternating path and remove the damaged u-fans from $\U$ and $\Phi$. 
It follows from \Cref{lem:low damage flips} that each iteration removes at most $2$ u-fans from these sets. Since we perform at most $\mu$ iterations, the claim follows.
\end{proof}

\subsection{Analysis of $\ColorUFans$: Proof of \Cref{lem:color u-fans}}\label{sec:proof of L1}

Given a separable collection of $\lambda$ u-fans $\mathcal U$, the algorithm $\ColorUFans$ repeatedly calls $\PrimeUFans$ and $\ActUFans$ with the set $\mathcal U$ as described in \Cref{alg:colorufansfinal}.
It repeats this process for $\Delta/2$ iterations.
It follows from Claims~\ref{cl:big U} and \ref{cl:big U 2} that $|\mathcal U|$ decreases by at most a $(1 - 1/\Delta)$ factor during each iteration. Thus, by Bernoulli's inequality, we get that
$$ |\mathcal U| \geq \lambda \cdot \left(1 - \frac{1}{\Delta} \right)^{\Delta/2} \geq \frac{\lambda}{2} $$
throughout the entire run of the algorithm. 
In each iteration of \Cref{alg:colorufansfinal}, we extend the coloring to $\Omega(\lambda / \Delta)$ edges in $O(m \log n/\Delta)$ time w.h.p. Thus, in total, we extend the coloring to $\Omega(\lambda)$ edges in $O(m \log n)$ time w.h.p.

\section{Implementation and Data Structures}\label{sec:data structs}

In this section, we describe the key data structures that we use to implement an edge coloring $\chi$ and a separable collection $\U$, allowing us to efficiently implement the operations performed by our algorithms.
We first describe the data structures and then show how they can be used to efficiently implement the queries described in \Cref{sec:data struc overview}.

\medskip
\noindent \textbf{Implementing an Edge Coloring:}
Let $G = (V, E)$ be a graph of maximum degree $\Delta$ and let $\mathcal C := [\Delta + 1] \cup \{\bot\}$.
We implement an edge coloring $\chi : E \longrightarrow \mathcal C$ of $G$ using the following:
\begin{itemize}
    \item The map $\phi : E \longrightarrow \mathcal C$ where $\phi(e) := \chi(e)$ for all $e \in E$.
    \item The map $\phi' : V \times \mathcal C \longrightarrow E$ where $\phi_u'(c) := \{e \ni u \mid \chi(e) = c\}$.
    \item The set $\chi^{-1}(c) := \{e \in E \mid \chi(e) = c\}$, for all $c \in \mathcal C$.
    \item The set $\miss_\chi(u) \cap [\deg_G(u) + 1]$, for all $u \in V$.\footnote{We take this intersection with $[\deg_G(u)+1]$ instead of maintaining $\miss_\chi(u)$ directly to ensure that the space complexity and initialization time of the data structures are $\tilde O(m)$ and not $\Omega(\Delta n)$.}
\end{itemize}
We implement all of the maps and sets using hashmaps, allowing us to perform membership queries, insertions and deletions in $O(1)$ time (see \Cref{prop:hash-map}).\footnote{By using balanced search trees instead of hashmaps, we can make the data structures deterministic while increasing the time taken to perform these operations to $O(\log n)$.}
The map $\phi'$ allows us to check if a color $c \in [\Delta + 1]$ is available at a vertex $u \in V$ in $O(1)$ time, and if it is not, to find the edge $e \ni u$ with $\chi(e) = c$.
The sets $\{\chi^{-1}(c)\}_{c \in \mathcal C}$ allow us to easily return all edges with a specific color (including $\bot$). Furthermore, we can determine which color classes are the least common in $O(1)$ time.\footnote{For example, we can then maintain a list of colors $c \in \mathcal C$ sorted by the values of $|\chi^{-1}(c)|$.} 
Each time an edge $e$ changes color under $\chi$, we can easily update all of these data structures in $O(1)$ time. Furthermore, given $O(1)$ time query access to an edge coloring $\chi$, we can initialize these data structures in $O(m)$ time.
We note that $\chi$ is a proper edge coloring if and only if $|\phi'_u(c)| \leq 1$ for all $u \in V, c \in [\Delta + 1]$.

Since the hashmap used to implement $\phi$ stores $m$ elements, it follows that it can be implemented with $O(m)$ space. Similarly, the map $\phi'$ stores $2m$ elements (if $\{e \ni u \mid \chi(e) = c\} = \varnothing$, then we do not store anything for $\phi'_u(c)$) and thus can be implemented with $O(m)$ space since each element has size $O(1)$ (recall that $|\phi'_u(c)| \leq 1$ since the coloring is proper). Since each set $\chi^{-1}(c)$ can be stored in space $O(|\chi^{-1}(c)|)$ and $\sum_{c}|\chi^{-1}(c)| = m$, these sets can be implemented in space $O(m)$. Similarly, since $\sum_u |\miss_\chi(u) \cap [\deg_G(u) + 1]| = O(m)$, the sets $\miss_\chi(u) \cap [\deg_G(u) + 1]$ can also be implemented in $O(m)$ space.

\medskip
\noindent \textbf{Implementing a Separable Collection:}
We implement a separable collection $\U$ in a similar manner using the following:
\begin{itemize}
    \item The map $\psi : V \times [\Delta + 1] \longrightarrow \U$ where $\psi_u(c) := \{\g \in \U \mid u \in \g, c_{\g}(u) = c\}$.
    \item The set $C_{\U}(u) := \{c_{\g}(u) \mid \g \in \U, u \in \g\}$, for all $u \in V$.
    \item The set $\overline{C}_{\U}(u) := \left(\miss_\chi(u) \cap [\deg_G(u) + 1] \right) \setminus C_{\U}(u)$, for all $u \in V$.
\end{itemize}
We again implement all of the maps and sets using hashmaps, allowing us to access and change entries in $O(1)$ time.
We note that, since $\U$ is separable, $|\psi_u(c)| \leq 1$ for all $u \in V, c \in [\Delta + 1]$.
Thus, we can determine the size of $\U$ and also sample from $\U$ uniformly at random in $O(1)$ time.\footnote{For example, we can sample a number $r \sim [|\mathcal U|]$ u.a.r.~and then return the $r^{th}$ element in the hashmap that implements $\psi$.} Each time we remove a color $c \in [\Delta + 1]$ from the palette $\miss_\chi(u)$ of a vertex $u \in V$, we can update $\overline{C}_{\U}(u)$ in $O(1)$ time and check $\psi_u(c)$ in $O(1)$ time to find any u-component that has been damaged. Each time we add or remove a u-component from $\U$, we can update all of these data structures in $O(1)$ time. 
Furthermore, we can initialize these data structures for an empty collection in $O(m)$ time by creating an empty map $\psi$, empty sets $C_{\U}(u)$ for each $u \in V$ and copying the sets $\overline{C}_{\U}(u) = \miss_\chi(u) \cap [\deg_G(u) + 1]$ for each $u \in V$ which are maintained by the data structures for the edge coloring $\chi$.
Since $\U$ is separable, we can see that $\overline{C}_{\U}(u) \neq \varnothing$. Thus, whenever we want a color from the set $\miss_\chi(u) \setminus C_{\U}(u)$, it suffices to take an arbitrary color from $\overline{C}_{\U}(u)$.

For each $\g \in \U$, we can see that $\g$ is contained at most $3$ times in $\psi$. Thus, the total space required to store the hashmap that implements $\psi$ is $O(|\U|)$. Since $\U$ is separable, the u-components in $\U$ are edge-disjoint, and thus $|\mathcal U| \leq m$. It follows that the map $\psi$ can be stored with $O(m)$ space. For each $u \in V$, we can observe that $|C_{\U}(u)| \leq \deg_G(u)$ since at most $\deg_G(u)$ many u-components in $\U$ contain the vertex $u$, and $|\overline{C}_{\U}(u)| \leq \deg_G(u) + 1$ since $\overline{C}_{\U}(u) \subseteq [\deg_G(u) + 1]$. Thus, the total space required to store the sets $\{C_{\U}(u)\}_{u \in V}$ and $\{\overline{C}_{\U}(u)\}_{u \in V}$ is $O(m)$.

\subsection{Implementing the Operations from \Cref{sec:data struc overview}}\label{sec:implementing queries}

We now describe how to implement each of the operations from \Cref{sec:data struc overview}.

\medskip
\noindent \textbf{Implementing} $\textsc{Initialize}(G, \chi)$: Suppose that we are given the graph $G$ and $O(1)$ time query access to an edge coloring $\chi$ of $G$. We can initialize the data structures used to maintain the maps $\phi$ and $\phi'$ in $O(m)$ time. We can then scan through the edges $e \in E$ and initialize the sets $\chi^{-1}(c)$ in $O(m)$ time. Finally, we can scan through the vertices $u \in V$ and initialize the sets $\miss_\chi(u) \cap [\deg_G(u) + 1]$ in $O(m)$ time. We can then initialize the data structures for an empty separable collection in $O(m)$ time by creating an empty map $\psi$ and, for each $u \in V$, initializing the sets $C_{\U}(u) \leftarrow \varnothing$ and $\overline{C}_{\U}(u) \leftarrow \miss_\chi(u) \cap [\deg_G(u) + 1]$.

\medskip
\noindent \textbf{Implementing} $\textsc{Insert}_{\U}(\g)$: By performing at most $3$ queries to the map $\psi$, we can check if $\U \cup \{\g\}$ is separable. If so we can update $\psi$ and the sets $C_{\U}(x)$ and $\overline{C}_{\U}(x)$ for $x \in \g$ in $O(1)$ time in order to insert $\g$ into $\U$. Otherwise, we return $\texttt{fail}$.

\medskip
\noindent \textbf{Implementing} $\textsc{Delete}_{\U}(\g)$: We can first make a query to $\psi$ to ensure that $\g \in \U$. If so, we can update $\psi$ and the sets $C_{\U}(x)$ and $\overline{C}_{\U}(x)$ for $x \in \g$ in $O(1)$ time to remove $\g$ from $\U$.

\medskip
\noindent \textbf{Implementing} $\textsc{Find-Component}_{\U}(x,c)$: We make a query to $\psi$ by checking if there is an element $\psi_x(c)$. If no such element is contained in $\psi$, then return $\texttt{fail}$. Otherwise, return the unique u-component in the set $\psi_x(c)$.

\medskip
\noindent \textbf{Implementing} $\textsc{Missing-Color}_{\U}(x)$: Return an arbitrary color from the set $\overline{C}_{\U}(x)$.

\newcommand{\set}[1]{\ensuremath{\{#1\}}}

\section{The Final Algorithm: Proof of~\Cref{thm:main}}\label{sec:log(n)}

Up until this point in the presentation of our algorithm and its analysis, we made no attempt in optimizing the logarithmic runtime factors, which led to an algorithm with $O(m\log^3{n})$ time for finding a $(\Delta+1)$-edge coloring with high probability in~\Cref{thm:main-tech}. We now show that we can further optimize the algorithm and obtain an $O(m\log{n})$ time algorithm and conclude the proof of~\Cref{thm:main}. 

Let us first start by listing where the $\log(n)$-terms come from in the proof of~\Cref{thm:main-tech}: 
\begin{enumerate}
	\item In~\Cref{lem:color u-fans}, when coloring u-fans, we are losing an $O(\log{n})$ factor to ensure the probabilistic guarantees of the algorithm hold with high probability in {each step} (see~\Cref{lemma:prime time}).\label{item:log:2} 
	\item Our application of~\Cref{lem:build u-fans,lem:color u-fans} can color a constant fraction of remaining uncolored edges, hence, we need to run them $O(\log{n})$ times to color all uncolored edges (see~\Cref{lem:main extend}).\label{item:log:3}  
	\item And finally, the entire framework of reducing $(\Delta+1)$ coloring to extending the coloring to $O(m/\Delta)$ uncolored edges using Eulerian partition technique (in the proof of~\Cref{thm:main-tech}) leads 
	to a recursion depth of $O(\log{n})$ leading to another $O(\log{n})$ overhead in the runtime.\label{item:log:4}
\end{enumerate}
All in all, these factors led to the $O(m  \log^3{n})$ bound of our algorithm in~\Cref{thm:main-tech}. 

We now show how these log factors can be reduced to a single one. The two main ideas are: (1) relaxing the requirement of~\Cref{lem:color u-fans} so that its runtime holds in expectation, plus a suitable tail bound. 
Then, instead of maintaining high probability bound on \emph{each} invocation of this lemma, we only bound the runtime of \emph{all} invocations of this lemma together with high probability; (2) relaxing the 
Eulerian partition technique to color most of the graph recursively, instead of the entire graph, and then taking care of the remaining uncolored edges at the end.

We start by presenting a more fine-grained version of two of our main technical lemmas in the proof of~\Cref{thm:main-tech} (\Cref{lem:color u-fans} and~\Cref{lem:main extend}; recall that~\Cref{lem:build u-fans} already runs in linear time deterministically), which correspond to part (1) above, and then use these 
to present part (2) of above and conclude the proof. 

\subsection{Fine-Grained Variants of ~\Cref{lem:color u-fans} and~\Cref{lem:main extend}}

We remove the $O(\log{n})$-term of~\Cref{lem:color u-fans} by focusing on the expected runtime of the algorithm (plus a crucial tail inequality). In the next part, we show how to recover the final result 
even from this weaker guarantee. 

\begin{lemma}[A slight modification of~\Cref{lem:color u-fans}]\label{lem:log(n)-color u-fans}
 	There is an algorithm that, given a graph $G$, a partial $(\Delta + 1)$-edge coloring $\chi$ and a separable collection of 
	$\lambda$ u-fans, extends $\chi$ to $\Omega(\lambda)$ in $T$ (randomized) time such that for some $T_0 = O(m)$, we have, 
	\[
	\expect{T} \leq T_0 \quad \text{and for all $\delta > 0$} \quad \Pr\sparen{T \geq (1+\delta) \cdot 2 T_0} \leq \exp\paren{-\frac{\delta \cdot \lambda}{100}}.
	\] 
\end{lemma}

\begin{proof}
	The amortized runtime for coloring each single edge in~\Cref{lem:color u-fans} is $O(m/\lambda)$ by~\Cref{lem:bipartite:2} assuming the coloring succeeds, 
	which happens with probability at least $1/4$ by~\Cref{lem:bipartite:1}. We emphasize that 
	this is in an average sense: the algorithm $\ColorUFans$ first uses $\PrimeUFans$ to prime $\Omega(\lambda/\Delta)$ u-fans in expected $O(m/\lambda)$ time per u-fan and thus $O((\lambda/\Delta) \cdot (m/\lambda)) = O(m/\Delta)$ expected
	total time, and then deterministically colors them in $O(m/\Delta)$ time using $\ActUFans$. It then repeats this process $\Delta/2$ times to color $\Omega(\lambda)$ edges, implying that $\expect{T} = O(m)$ time. 
		
        We now prove the desired tail inequality on $T$ as well. 
        Let us number all iterations of the while-loop in $\PrimeUFans$ from $1$ to $t$, across \emph{all} $\Delta/2$ times $\PrimeUFans$ is called in $\ColorUFans$: for the $i^{th}$ iteration, define an indicator random variable $X_i$ 
        which is $1$ iff this iteration of the while-loop is successful in priming a u-fan (i.e., increasing the size of $\Phi$). Additionally, for every such $i$, define $S_i := \sum_{j=1}^{i} X_j$ which 
        denotes the number of successful iterations after the algorithm has done $i$ iterations in total. Thus, $t$ is the smallest integer where $S_t = \gamma \cdot \lambda$ for some integer $\gamma \in (0,1)$ which is the fraction of u-fans the algorithm colors 
        (precisely, $\gamma = 1/48 \cdot 1/2 \geq 1/100$). 
        
        For any $i \geq 1$, $S_i$ stochastically dominates the binomial distribution with parameters $i$ and $p=1/4$ (by~\Cref{lem:bipartite:1} each iteration is 
        successful with probability at least $1/4$). Hence, for $\delta > 0$, by concentration results for the binomial distribution, 
        \[
        		\Pr\sparen{t \geq (1+\delta) \cdot 8\gamma \cdot \lambda} \leq \exp\paren{-\frac{(2(1+\delta))^2}{2+2(1+\delta)} \cdot 4\gamma \cdot \lambda} \leq \exp\paren{-\frac{\delta \cdot \lambda}{100}},
        \]
        using a loose upper bound in the last step. Given that the runtime of the algorithm is $t \cdot O(m/\lambda)$, the tail inequality follows. 
\end{proof}

We note~\Cref{lem:log(n)-color u-fans} guarantees that as long as $\lambda = \omega(\log{n})$, the $O(m)$ runtime of the algorithm also holds with high probability. For smaller values of $\lambda$ (which happens only as a corner case in the algorithm\footnote{The
only case in our algorithm where this tail bound cannot be replaced with a high probability bound is when $m/\Delta = o(\log{n})$ and at the same time $\Delta\log^2\Delta = \omega(m\log{\Delta})$, which means $m = o(n\log{n})$ and yet $\Delta = \omega(n/\log{n})$.}), we need the specific tail inequality proven in the lemma instead.

Furthermore, we provide a similarly fine-grained version of~\Cref{lem:main extend} that will be used in the last step of the argument to bypass the $O(\log{n})$ factor
loss of the original lemma. 

\begin{lemma}[A slight modification of \Cref{lem:main extend}]\label{lem:log(n)-main extend}	    
     There is an algorithm that given a graph $G$, a partial $(\Delta + 1)$-edge coloring $\chi$ of $G$ with $\lambda$ uncolored edges, 
     and an integer $\lambda_0 \leq \lambda$, extends $\chi$ to all but $\lambda_0$ uncolored edges in $T$ (randomized) time such that 
     for some $T_0 = O(m \cdot \log{(\lambda/\lambda_0)} + \Delta \lambda)$, we have, 
     \[
     	\expect{T} \leq T_0 \quad \text{and for all $\delta > 0$} \quad \Pr\sparen{T \geq (1+\delta) \cdot 2T_0} \leq \exp\paren{-\frac{\delta \cdot \lambda_0}{200}}. 
     \] 
\end{lemma}
\begin{proof}
	The algorithm is verbatim as in~\Cref{lem:main extend} by replacing~\Cref{lem:color u-fans} with~\Cref{lem:log(n)-color u-fans}: repeatedly color
	a constant fraction of remaining edges as long as $\lambda_0$ uncolored edges remain. 
	Since we start with $\lambda$ uncolored edges and finish with $\lambda_0$ many, and reduce uncolored edges by a constant factor each time, we apply~\Cref{lem:build u-fans} and~\Cref{lem:log(n)-color u-fans} for a total of $O(\log{(\lambda/\lambda_0)})$ times. Moreover, 
	 the $\Delta \lambda$ terms in the runtime of~\Cref{lem:build u-fans}  form a geometric series and thus their contribution to the runtime is $O(\Delta\lambda)$ in total.
	This proves the expected bound on the runtime, i.e., $\expect{T} \leq T_0 = O(m\cdot\log{(\lambda/\lambda_0)} + \Delta \lambda)$. 
	
	As for the tail bound, in each application of~\Cref{lem:log(n)-color u-fans} with an intermediate value $\lambda' \in [\lambda_0 , \lambda]$, 
	and $T'_0 = O(m)$ being the $T_0$-parameter of~\Cref{lem:log(n)-color u-fans}, we have, 
	\[
		\Pr\sparen{\text{runtime of~\Cref{lem:log(n)-color u-fans} with $\lambda'$ uncolored edges} \geq (1+\delta) \cdot 2T'_0} \leq  \exp\paren{-\frac{\delta \cdot \lambda'}{100}}.
	\]
	Thus, by union bound, 
	\begin{align*}
		\Pr\sparen{\text{total runtime} \geq k \cdot T_0} &\leq \sum_{\lambda'} \exp\paren{-\frac{\delta \cdot \lambda'}{100}} \\
		&\leq \sum_{i=0}^{\infty} \cdot \exp\paren{-\frac{\delta \cdot \lambda_0 \cdot \gamma^i}{100}} \tag{as number of uncolored edges drops by some constant factor, say, $\gamma = 1+\Theta(1)$ each time} \\
		&\leq \exp\paren{-\frac{\delta \cdot \lambda_0}{200}}, 
	\end{align*}
	by (loosely) upper bounding the sum of the geometric series using its first term. 
\end{proof}

Before moving on from this subsection, we mention the following specialized concentration inequality that we need in order to be able to exploit the tail bounds proven in~\Cref{lem:log(n)-color u-fans} and~\Cref{lem:log(n)-main extend}. 
The proof follows a standard moment generating function argument and is postponed to~\Cref{app:concentration}. 
\begin{proposition}\label{prop:conc}
	Let $\set{X_i}_{i=1}^{n}$ be $n$ independent non-negative random variables associated with parameters $\set{\alpha_i}_{i=1}^{n}$ and $\set{\beta_i}_{i=1}^{n}$ such that for each $i \in [n]$, $\alpha_i,\beta_i \geq 1$, and for every $\delta > 0$, 
	\[
		\Pr\sparen{X_i \geq (1+\delta) \cdot \alpha_i} \leq \exp\paren{- \delta \cdot \beta_i}; 
	\]
	then,  for every $t \geq 0$,
	\[
		\Pr\sparen{\sum_{i=1}^{n} X_i \geq \sum_{i=1}^{n} \alpha_i + t} \leq \exp\paren{-\left(\min_{i=1}^n\frac{\beta_i}{2\alpha_i} \right) \cdot \paren{t-2\sum_{i=1}^n\frac{\alpha_i}{\beta_i}}}.
	\]
\end{proposition}

\subsection{Proof of~\Cref{thm:main}}
We are now ready to use the more fine-grained versions of our main technical lemmas to $(\Delta+1)$ edge color the graph in $O(m\log{n})$ randomized time, and conclude the proof of~\Cref{thm:main}. 

The first part is based on a modification to the Eulerian partition approach used in the proof of~\Cref{thm:main-tech}. Notice that the runtime of this algorithm is even $O(m\log{\Delta})$ and not $O(m\log{n})$ (the distinction 
at this point is irrelevant for us in proving~\Cref{thm:main}, but we state the result this way so we can use it later in~\Cref{app:small-Delta} as well).  

\begin{lemma}\label{lem:log(n)-Euler}
     There is an algorithm that given a graph $G$, finds a $(\Delta+1)$ edge coloring of all but (exactly) $m/\Delta$ edges in $O(m\log{\Delta})$ time
     with high probability. 
\end{lemma}
\begin{proof}
	Consider the following recursive algorithm. 
	Find an Eulerian tour of $G$ and partition $G$ into two edge-disjoint subgraphs $G_1$ and $G_2$ on the same vertex set such that $\Delta(G_1),\Delta(G_2) \leq \lceil \Delta/2 \rceil$
	and $m(G_1),m(G_2) \leq \lceil m/2 \rceil$. For $i \in \{1,2\}$, recursively, find a $(\Delta(G_i)+1)$-edge coloring $\chi_i$ of $G_i$ that leaves $m(G_i)/\Delta(G_i) = O(m/\Delta)$ edges uncolored. 
	Combining $\chi_1$ and $\chi_2$ and uncoloring the two smallest color classes, gives a $(\Delta+1)$ coloring $\chi$ of all but $\lambda := O(m/\Delta)$ edges. 
	We then run~\Cref{lem:log(n)-main extend} to reduce the number of uncolored edges to $\lambda_0 := m/\Delta$ edges. 
	The correctness of the algorithm thus follows immediately. 
	
	For the runtime, we have $k:= O(\Delta)$ sub-problems in total, with $2^{d}$ subproblems on $\ceil{m/2^{d}}$ edges and maximum degree $\ceil{\Delta/2^{d}}$ for $d \leq \ceil{\log{\Delta}}$. 
	Let $X_1,\ldots,X_{k}$ denote the runtime of these subproblems and thus $X := \sum_{i=1}^{k} X_i$ is the total runtime. We have, 
	\[
		\expect{X} = \sum_{i} \expect{X_i} = \sum_{d=1}^{\ceil{\log{\Delta}}} 2^{d} \cdot O \left(\frac{m}{2^{d}} \cdot \log{(\lambda/\lambda_0)} + \frac{\Delta}{2^{d}} \cdot \lambda \right) = O(m\log{\Delta})
	\]
	by~\Cref{lem:log(n)-main extend}, since $\lambda = O(m/\Delta)$, $\lambda_0 = m/\Delta$, and thus $\log{(\lambda/\lambda_0)} = O(1)$ and $\Delta \cdot \lambda = O(m)$. Moreover, 
	by~\Cref{lem:log(n)-main extend}, for the sub-problem corresponding to $X_i$ at some level $d$ of the recursion, $T^d_0 = O(m/2^d)$, and every $\delta \geq 0$, 
	\[
		\Pr\sparen{X_i \geq (1+\delta) \cdot 2T^d_0} \leq \exp\paren{-\delta \cdot \frac{\lambda_0}{200}} = \exp\paren{-\delta \cdot \frac{m}{200\Delta}}. 
	\]
	For this $X_i$, define $\alpha_i := 2T^d_0$ and $\beta_i := m/200\Delta$. We can now apply~\Cref{prop:conc} and for 
	\[
	t := 1000 \sum_{i=1}^{k} \alpha_i = 1000 \cdot \sum_{d=1}^{\ceil{\log{\Delta}}} 2^d \cdot 2T^d_0 = O(m\log{\Delta}),
	\]
	 have, 
	\begin{align*}
		\Pr\sparen{X \geq t} &\leq \exp\paren{-\left(\min_{i=1}^n\frac{\beta_i}{2\alpha_i}\right) \cdot \paren{t-2\sum_{i=1}^n\frac{\alpha_i}{\beta_i}}} \\
		&= \exp\paren{-\frac{\Theta(1)}{\Delta} \cdot \paren{t- 2 \cdot \frac{t \cdot 200\Delta}{1000 \cdot m}}} \\
		&\leq \exp\paren{-\frac{\Theta(1)}{\Delta} \cdot (3t/5)} \tag{as $\Delta \leq m$ always trivially} \\
		&= \exp\paren{-\Theta(1) \cdot \frac{m\log{\Delta}}{\Delta}}, 
	\end{align*}
	where we can make the constant in $\Theta(1)$ arbitrarily large without changing the asymptotic runtime of the algorithm. 
	This probability is always at most $1/\poly(n)$ because either $\Delta < \sqrt{m}$, in which case, the probability is at most $1/2^{\Theta(\sqrt{m})}$, 
	or $\Delta > \sqrt{m}$, and thus $\log{\Delta} = \Theta(\log{n})$ and the probability is $1/\poly(n)$ at the very least. This concludes the proof. 
\end{proof}  

To conclude the proof, we have the following lemma that colors the remaining $m/\Delta$ edges left uncolored by the algorithm of~\Cref{lem:log(n)-Euler}. The proof is a straightforward application of~\Cref{lem:log(n)-main extend},
plus the original Vizing's Fan and Chain approach stated in~\Cref{lem:full vizing proof}.  

\begin{lemma}\label{lem:log(n)-wrap-up}
	There is an algorithm that, given a graph $G$, a partial $(\Delta+1)$-edge coloring $\chi$ 
	with $\lambda=m/\Delta$ uncolored edges, extends the coloring to all edges in $O(m\log{n})$ time with high probability. 
\end{lemma}
\begin{proof}
	We first run our~\Cref{lem:log(n)-main extend} with the given parameters $\lambda=m/\Delta$ and $\lambda_0=100\log{n}$. By~\Cref{lem:log(n)-main extend}, the expected runtime will be 
	some $T_0 = O(m\log{(\lambda/\lambda_0)} + \Delta \lambda) = O(m\log{n})$ since $\lambda \leq m$ and $\lambda_0 \geq 1$. 
	Moreover, by the same lemma, 
	\[
		\Pr\sparen{\text{runtime of the algorithm} \geq 2000 \cdot T_0} \leq \exp\paren{-10 \lambda_0} \leq \exp\paren{-1000\log{n}} = 1/\poly(n),  
	\]
	since $\lambda_0 \geq 100\log{n}$. Thus, after running this part, in $O(m\log{n})$ time, with high probability, we will be left with only $O(\log{n})$ uncolored edges. 
	We can then color each remaining uncolored edge in $O(n)$ time using the original Vizing's Fans and Chains approach of~\Cref{lem:full vizing proof}, 
	thus obtaining a $(\Delta+1)$ coloring of the entire graph in $O(m\log{n})$ time with high probability. 
\end{proof}

\Cref{thm:main} now follows immediately from~\Cref{lem:log(n)-Euler} and~\Cref{lem:log(n)-wrap-up}.

\section*{Acknowledgements} 
Part of this work was conducted while Sepehr Assadi and Soheil Behnezhad were visiting the Simons Institute for the Theory of Computing as part of the Sublinear Algorithms program. 

Sepehr Assadi is supported in part by a Sloan Research Fellowship, an NSERC Discovery Grant (RGPIN-2024-04290), and a Faculty of Math Research Chair grant from University of Waterloo.
Soheil Behnezhad is funded by an NSF CAREER award CCF-2442812 and a Google Faculty Research Award. Mart\'in Costa is supported by a Google PhD Fellowship. Shay Solomon is funded by the European Union (ERC, DynOpt, 101043159). Views and opinions expressed are however those of the author(s) only and do not necessarily reflect those of the European Union or the European Research Council. Neither the European Union nor the granting authority can be held responsible for them. Shay Solomon is also funded by a grant from the United States-Israel Binational Science Foundation (BSF), Jerusalem, Israel, and the United States National Science Foundation (NSF). Tianyi Zhang is supported by funding from the starting grant ``A New Paradigm for Flow and Cut Algorithms'' (no. TMSGI2\_218022) of the Swiss National Science Foundation.

\bibliography{bibliography.bib}
\bibliographystyle{alpha}

\newpage

\appendix

\section{Vizing's Theorem for Multigraphs in Near-Linear Time}

We now show how to extend our arguments to \emph{multigraphs}. A generalization of Vizing's theorem to multigraphs shows that any multigraph $G$ with maximum degree $\Delta$ and maximum multiplicity $\mu$ can be edge colored with $\Delta + \mu$ colors \cite{Vizing} (and this is worst-case optimal for $\mu \leq \Delta/2$ in the sense that not every multigraph admits an edge coloring with fewer colors). We prove the following theorem, which shows how to compute such a coloring in near-linear time.

\begin{theorem}\label{thm:multi Viz}
    There is a randomized algorithm that, given an undirected multigraph $G = (V, E)$ on $n$ vertices and $m$ edges with maximum degree $\Delta$ and maximum multiplicity $\mu$, finds a $(\Delta + \mu)$-edge coloring of $G$ in $O(m\log{n})$ time with high probability.
\end{theorem}

Similar to our results for $(\Delta+1)$ edge coloring, we first focus on only obtaining an $\tilde O(m)$ time algorithm, and then use almost exactly the same argument we used to extend~\Cref{thm:main-tech} 
to~\Cref{thm:main}, to optimize this runtime to $O(m\log{n})$ time. 

\medskip
\noindent\textbf{Notation for Multigraphs:}
We refer to edges in $G$ with the same endpoints as \emph{parallel}. Whenever we refer to an edge $e$ of the multigraph $G$, we are referring to a specific edge and say that this edge is distinct from its parallel edges.

\subsection{Vizing Fans and Separable Collection in Multigraphs}

Let $G = (V, E)$ be an undirected multigraph on $n$ vertices and $m$ edges with maximum degree $\Delta$ and maximum multiplicity $\mu$, and $\chi$ be a partial $(\Delta + \mu)$-edge coloring of $G$. 
We now describe how to generalize the objects used by our algorithm for simple graphs to deal with multigraphs.
We begin by defining a generalization of Vizing fans for multigraphs \cite{Vizing}. For convenience, we still refer to them as Vizing fans.

\begin{definition}[Vizing fan for multigraphs]\label{def:multifans}
A \emph{Vizing fan} for a multigraph is a sequence $\F = (u, \alpha),(v_1,e_1,C_1),\dots,(v_k,e_k,C_k)$ where $u,v_1,\dots,v_k$ are distinct vertices, $C_1,\dots,C_k \subseteq [\Delta + \mu]$ are subsets of colors and $e_1,\dots,e_k$ are edges such that
\begin{enumerate}
    \item $\alpha \in \miss_\chi(u)$ and $C_i \subseteq \miss_\chi(v_i)$ of size $\mu$ for all $i \in [k]$.\label{def:multif:1}
    \item $v_1,\dots,v_k$ are distinct neighbours of $u$ and $e_i$ is an edge with endpoints $u$ and $v_i$ for all $i \in [k]$.
    \item $\chi(e_1) = \bot$ and, for all $i > 1$, there exists $p(i) \in [i-1]$ such that $\chi(e_i) \in C_{p(i)}$.
    \item The sets $C_1,\dots,C_{k-1},\miss_\chi(u)$ are mutually disjoint.\label{def:multif:4}
    \item Either $C_k \cap \miss_\chi(u) \neq \varnothing$ or $C_k \cap \left(C_1 \cup \dots \cup C_{k-1}\right) \neq \varnothing$.
\end{enumerate}
\end{definition}

\noindent
We say that the Vizing fan $\F = (u, \alpha),(v_1,e_1,C_1),\dots,(v_k,e_k,C_k)$ is $\alpha$-\emph{primed}, has \emph{center} $u$ and \emph{leaves} $v_1,\dots,v_k$. We refer to $C_i$ as the colors of $v_i$ within $\F$.
A crucial property is that we can \emph{rotate} colors around the Vizing fan $\F$: given any $i \in [k]$, there is a sequence $1 = i_1 < \dots < i_\ell = i$ (where $i_{\ell-j} = p^{(j)}(i)$) such that we can set $\chi(e_{i_1}) \leftarrow \chi(e_{i_2}),\dots, \chi(e_{i_{\ell - 1}}) \leftarrow \chi(e_{i_\ell}), \chi(e_{i_\ell}) \leftarrow \bot$.
We say that $\F$ is a \emph{trivial} Vizing fan if $C_k \cap \miss_\chi(u) \neq \varnothing$.
Note that, if $\F$ is trivial, we can immediately extend the coloring $\chi$ to $(e_1)$ by rotating colors around $\F$ to leave $(e_1)$ uncolored and setting $\chi(e_k)$ to be any color in $C_k \cap \miss_\chi(u)$.

For completeness, we prove the following lemma, which shows how to efficiently construct such a Vizing fan.

\begin{lemma}[\cite{Vizing}] Given an uncolored edge $e = (u,v)$ and a color $\alpha \in \miss_\chi(u)$, we can construct an $\alpha$-primed Vizing fan $\F$ with center $u$ in $O(\Delta)$ time.
\end{lemma}

\begin{proof}
    Let $v_1 = v$, $e_1 = e$ and $C_1 \subseteq \miss_\chi(v)$ be a subset of size $\mu$. If $C_1 \cap \miss_\chi(u) \neq \varnothing$, then we can return $(u,\alpha), (v_1, e_1, C_1)$ as a Vizing fan. Otherwise, we assume inductively that we have constructed a sequence $(u,\alpha), (v_1,e_1, C_1),\dots, (v_i, e_i, C_i)$ satisfying Conditions~\ref{def:multif:1} to \ref{def:multif:4} of \Cref{def:multifans}. In this case, we have that $|C_1 \cup \dots \cup C_i| = i\mu$. Since there are at most $i \mu - 1$ many colored edges that have $u$ as one endpoint and one of $v_1,\dots,v_i$ as the other, it follows that there must be some edge $e_{i+1} = (u, v_{i+1})$ such that $\chi(e_{i+1}) \in C_1 \cup \dots \cup C_i$ and $v_{i+1} \notin \{v_1,\dots,v_i\}$. Let $C_{i+1} \subseteq \miss_\chi(v_{i+1})$ be a subset of size $\mu$. If either $C_{i+1} \cap \miss_\chi(u) \neq \varnothing$ or $C_{i+1} \cap \left(C_1 \cup \dots \cup C_{i}\right) \neq \varnothing$, then $(u,\alpha), (v_1, e_1, C_1),\dots, (v_{i+1}, e_{i+1}, C_{i+1})$ is a Vizing fan and we are done. Otherwise, we continue the induction on this longer sequence. Since this process must terminate within at most $\Delta/\mu$ steps, this process always returns a Vizing fan. Furthermore, using standard data structures, we can implement this process in $O(\Delta)$ time.
\end{proof}

\medskip
\noindent \textbf{Vizing Chains in Multigraphs:} Let $\F= (u, \alpha),(v_1,e_1,C_1),\dots,(v_k,e_k,C_k)$ be a non-trivial $\alpha$-primed Vizing fan, let $c$ be a color in $C_k \cap (C_1,\dots,C_{k-1})$ and let $P$ denote the maximal $\{\alpha, c\}$-alternating path starting at $u$.
Then, by a similar argument to the case of simple graphs, we can extend the coloring $\chi$ to the edge $e_1$ by flipping the path $P$ and rotating colors around $\F$. Furthermore, this can be done in time $O(\Delta + |P|)$. We denote a call to the algorithm that extends the coloring in the manner by $\Vizing(\F)$.

\medskip
\noindent \textbf{Separable Collections in Multigraphs:}
We define a separable collection of u-components $\U$ in a multigraph in the exact same way as simple graphs. We emphasise that, given a separable collection $\U$, any distinct u-components $\g_1, \g_2 \in \U$ are edge-disjoint but may contain parallel edges. In other words, it is possible for two distinct parallel uncolored edges $e_1$ and $e_2$ to be contained in u-components within the separable collection $\U$.

We define $\U$-avoiding Vizing fans for multigraphs in the following analogous way.

\begin{definition}\label{def:U-avoid:multi}
    Let $\U$ be a separable collection and $\F = (u,\alpha),(v_1,e_1,C_1),\dots, (v_k,e_k,C_k)$ be a Vizing fan. We say that the Vizing fan $\F$ is \emph{$\U$-avoiding} if $C_i \subseteq \miss_\chi(v_i) \setminus C_{\U}(v_i)$ for each $v_i \in \F$.
\end{definition}

Using these new definitions, we can directly extend \Cref{lem:fast U-avoid fan,lem:Vfans are safe} to multigraphs by slightly changing the proofs, giving us the following lemmas.

\begin{lemma}\label{lem:fast U-avoid fan:multi}
    Given a u-edge $\e \in \U$, there exists a $\U$-avoiding Vizing fan $\F$ of $\e$. Furthermore, we can compute such a Vizing fan in $O(\Delta)$ time.\footnote{Recall that we say that $\F$ is a Vizing fan of $\e = (u,v,\alpha)$ if $\F$ is $\alpha$-primed, has center $u$, and its first leaf is $v$.}
\end{lemma}

\begin{lemma}\label{lem:Vfans are safe:multi}
Let $\chi$ be a $(\Delta + \mu)$-edge coloring of a graph $G$ and $\U$ be a separable collection. For any u-edge $\e = (u,v,\alpha) \in \U$ with a $\U$-avoiding Vizing fan $\F$, we have the following:
\begin{enumerate}
    \item Rotating colors around $\F$ does not damage any u-component in $\U \setminus \{\e\}$.
    \item Calling $\Vizing(\F)$ damages at most one u-component in $\U \setminus \{\e\}$. Furthermore, we can identify the damaged u-component in $O(1)$ time.
\end{enumerate}
\end{lemma}

\subsection{Proof of \Cref{thm:multi Viz}}

Our main algorithm for multigraphs is completely analogous to our main algorithm for simple graphs, which we describe in \Cref{sec:main:algo}. In particular, the algorithm also consists of two main components that are combined in the same way, which we summarize in the following lemmas.

\begin{lemma}[A modification of~\Cref{lem:build u-fans} for multigraphs]\label{lem:build u-fans:multi}
    There is an algorithm that, given a multigraph $G$, a partial $(\Delta + \mu)$-edge coloring $\chi$ of $G$ and a set of $\lambda$ uncolored edges $U$, does one of the following in $O(m + \Delta \lambda)$ time:
    \begin{enumerate}
        \item  Extends the coloring to $\Omega(\lambda)$ uncolored edges.
        \item  Modifies $\chi$ to obtain a separable collection of $\Omega(\lambda)$ u-fans $\mathcal U$.
    \end{enumerate}
\end{lemma}

\noindent The proof of this lemma is similar to the proof of \Cref{lem:build u-fans} but with some minor changes.
In \Cref{sec:proof lem multi1}, we sketch how to modify the proof of \Cref{lem:build u-fans} to extend it to \Cref{lem:build u-fans:multi}.

\begin{lemma}[A modification of~\Cref{lem:color u-fans} for multigraphs]\label{lem:color u-fans:multi}
    There is an algorithm that, given a multigraph $G$, a partial $(\Delta + \mu)$-edge coloring $\chi$ of $G$ and a separable collection of $\lambda$ u-fans $\mathcal U$, extends $\chi$ to $\Omega(\lambda)$ edges in $O(m \log n)$ time w.h.p.
\end{lemma}

\noindent The proof of this lemma is almost identical to the proof of \Cref{lem:color u-fans}, thus we omit the details.

Using \Cref{lem:build u-fans:multi,lem:color u-fans:multi}, we get the following lemma which allows us to extend an edge coloring $\chi$ to a set of uncolored edges.

\begin{lemma}[A modification of~\Cref{lem:main extend} for multigraphs]\label{lem:main extend:multi}
    There is an algorithm that, given a multigraph $G$ and a partial $(\Delta + \mu)$-edge coloring $\chi$ of $G$ with $\lambda$ uncolored edges $U$, extends $\chi$ to the remaining uncolored edges in time $O((m + \Delta \lambda)\log^2 n)$ w.h.p.
\end{lemma}

\noindent The proof of this lemma is verbatim the same as the proof of \Cref{lem:main extend}, and hence is omitted. 

Finally, in the same way as in the proof of \Cref{thm:main}, we can apply standard Euler partitioning using \Cref{lem:main extend:multi} to merge the colorings. We only note that in multigraphs, 
splitting an $m$-edge multigraph $G$ via the Euler partition leads to two multigraphs, each with $\ceil{m/2}$ edges, maximum degree $\ceil{\Delta/2}$, and maximum multiplicity $\ceil{\mu/2}$. 
Thus, the number of colors used to color the graph recursively will be $2 \cdot (\ceil{\Delta/2} + \ceil{\mu/2}) = \Delta + \mu + O(1)$ colors. Hence, we again only need to uncolor $O(m/\Delta)$ edges
to obtain a $(\Delta+\mu)$ coloring and then extend this coloring to the remaining $O(m/\Delta)$ edges using~\Cref{lem:main extend:multi}. This leads to an $O(m\log^3{n})$ time algorithm for finding a $(\Delta+\mu)$ edge coloring with high probability. 

We can also optimize this algorithm to run in $O(m\log{n})$ time with high probability exactly as in our $(\Delta+1)$ edge coloring algorithm. Specifically, using exactly the same argument as in~\Cref{sec:log(n)}, we have the following analogues of~\Cref{lem:log(n)-Euler} and~\Cref{lem:log(n)-wrap-up} for multigraphs. 

\begin{lemma}[A modification of~\Cref{lem:log(n)-Euler} for multigraphs]\label{lem:log(n)-Euler-multi}
     There is an algorithm that given a multigraph $G$, finds a $(\Delta+\mu)$ edge coloring of all but (exactly) $m/\Delta$ edges in $O(m\log{\Delta})$ time
     with high probability. 
\end{lemma}

\begin{lemma}[A modification of~\Cref{lem:log(n)-wrap-up} for multigraphs]\label{lem:log(n)-wrap-up-multi}
	There is an algorithm that, given a multigraph $G$, a partial $(\Delta+\mu)$-edge coloring $\chi$ 
	with $\lambda=m/\Delta$ uncolored edges, extends the coloring to all edges in $O(m\log{n})$ time with high probability. 
\end{lemma}

\Cref{thm:multi Viz} now follows immediately from~\Cref{lem:log(n)-Euler-multi} and~\Cref{lem:log(n)-wrap-up-multi}.

\subsection{Proof Sketch of \Cref{lem:build u-fans:multi}}\label{sec:proof lem multi1}

The algorithm for \Cref{lem:build u-fans:multi} is almost identical to the algorithm for \Cref{lem:build u-fans}, except that we replace the relevant definitions with their generalizations for multigraphs, leading to a few extra cases. The overall structure of the algorithm is still the same as \Cref{alg:constructUfans}, but using modifications of $\PruneVFans$ and $\ReduceUEdges$ for multigraphs.

\medskip
\noindent \textbf{Initializing the Seperable Collection $\U$:} Given a set of $\lambda$ uncolored edges $U$, we can construct a separable collection of $\lambda$ u-edges $\U$ in the exact same way as simple graphs. However, there may exist distinct u-edges $\e, \e' \in \U$ that correspond to parallel (but distinct) uncolored edges.

\medskip
\noindent \textbf{Modifying $\PruneVFans$ for Multigraphs:} Suppose that we want to run $\PruneVFans$ on a separable collection $\U$ with color $\alpha$. In this call to the subroutine, we consider all of the $\alpha$-primed u-edges in $\E_\alpha(\U) \subseteq \U$. We implement this subroutine for multigraphs in the same way as for simple graphs, \emph{but with the following additional preprocessing step to ensure that none of the u-edges in $\E_\alpha(\U)$ are parallel}.

Suppose that we have distinct $\e, \e' \in \E_\alpha(\U)$ that correspond to parallel uncolored edges with endpoints $u$ and $v$. Then, by the definition of a separable collection, one of these u-edges must have $u$ as a center and the other must have $v$ as a center. Thus, $\alpha \in \miss_\chi(u) \cap \miss_\chi(v)$, so we can extend the coloring to one of these u-edges by coloring it with $\alpha$ and then remove them both from $\U$. We can scan through all of the $\alpha$-primed u-edges and perform this operation whenever we encounter parallel u-edges.
This maintains the invariant that $\U$ is a separable collection and ensures that none of the u-edges in $\E_\alpha(\U)$ are parallel.

After performing this preprocessing step, we can implement the rest of $\PruneVFans$ in the exact same way as described in \Cref{sec:prune}, except that we must use the generalization of Vizing fans for multigraphs. In particular, all of the lemmas in \Cref{sec:prune} describing the properties of $\PruneVFans$ still hold.

\medskip
\noindent \textbf{Modifying $\ReduceUEdges$ for Multigraphs:} Suppose that we want to run $\ReduceUEdges$ on a separable collection $\U$ with color $\alpha$ and have obtained a collection of vertex-disjoint $\U$-avoiding Vizing fans $\mathcal F$ corresponding to the u-edges in $\E_\alpha(\U)$ from calling $\PruneVFans$. We implement the subroutine $\ReduceUEdges$ in the same way as described in \Cref{sec:reduce} but with the following modification to \Cref{alg:update path} to account for an additional case.

Suppose that we call $\UpdatePath(\F)$ and take one more step along the path $\VizingP(\F)$ and observe the edge $e = (x,y)$. If the edge $e$ has already been seen (i.e.~is contained in $S$) then we proceed as normal.\footnote{That is, this exact copy of the edge is in $S$, not just an edge parallel to $e$.} If neither the edge $e$ nor any edge parallel to $e$ is contained in $S$, then we again proceed as normal. However, if we observe that $e$ is not already contained in $S$ but some edge $e'$ that is parallel to $e$ is contained in $S$, then we do the following: Let $\F' \in \mathcal F$ be the Vizing fan such that $e' \in P_{\F'}$, set $\chi(e) \leftarrow \bot$ and $\chi(e') \leftarrow \bot$, call $\Vizing(\F)$ and $\Vizing(\F')$, set $\chi(e) \leftarrow \alpha$, remove $\F$ and $\F'$ from $\mathcal F$ and $\e_{\F}$ and $\e_{\F'}$ from $\U$, and set $S \leftarrow S \setminus (P_{\F} \cup P_{\F'})$; see \Cref{multi} for an illustration. We can verify that, if this does happen, the edges $e$ and $e'$ must appear in different orientations in the paths $P_{\F}$ and $P_{\F'}$. Consequently, this operation removes 2 u-edges from $\U$ and extends the coloring to one more edge.

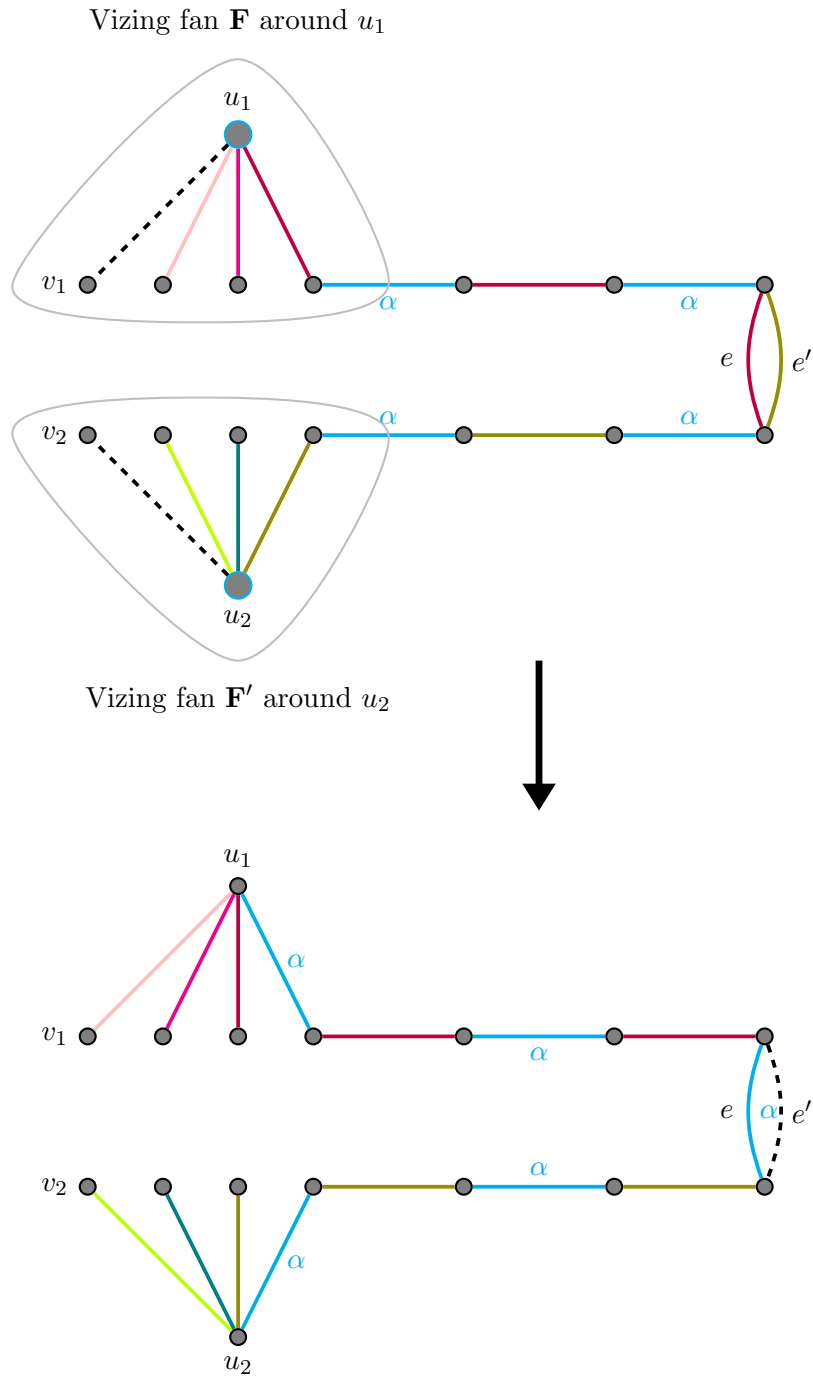
\begin{figure}
	\centering
	\begin{tikzpicture}[thick,scale=1]
	\draw (0, 3) node(1)[circle, draw, color=cyan, fill=black!50,
	inner sep=0pt, minimum width=10pt, label = $u_1$] {};
	
	\draw (-2, 1) node(2)[circle, draw, fill=black!50,
	inner sep=0pt, minimum width=6pt, label = 180:{$v_1$}] {};
	
	\draw (-1, 1) node(3)[circle, draw, fill=black!50,
	inner sep=0pt, minimum width=6pt] {};
	
	\draw (0, 1) node(4)[circle, draw, fill=black!50,
	inner sep=0pt, minimum width=6pt] {};
	
	\draw (1, 1) node(5)[circle, draw, fill=black!50,
	inner sep=0pt, minimum width=6pt] {};
	
	\draw (3, 1) node(6)[circle, draw, fill=black!50,
	inner sep=0pt, minimum width=6pt] {};
	
	\draw (5, 1) node(7)[circle, draw, fill=black!50,
	inner sep=0pt, minimum width=6pt] {};
	
	\draw (7, 1) node(8)[circle, draw, fill=black!50,
	inner sep=0pt, minimum width=6pt] {};
	
	
	\draw (0, -3) node(11)[circle, draw, color=cyan, fill=black!50,
	inner sep=0pt, minimum width=10pt, label = -90:{$u_2$}] {};
	
	\draw (-2, -1) node(12)[circle, draw, fill=black!50,
	inner sep=0pt, minimum width=6pt, label = 180:{$v_2$}] {};
	
	\draw (-1, -1) node(13)[circle, draw, fill=black!50,
	inner sep=0pt, minimum width=6pt] {};
	
	\draw (0, -1) node(14)[circle, draw, fill=black!50,
	inner sep=0pt, minimum width=6pt] {};
	
	\draw (1, -1) node(15)[circle, draw, fill=black!50,
	inner sep=0pt, minimum width=6pt] {};
	
	\draw (3, -1) node(16)[circle, draw, fill=black!50,
	inner sep=0pt, minimum width=6pt] {};
	
	\draw (5, -1) node(17)[circle, draw, fill=black!50,
	inner sep=0pt, minimum width=6pt] {};
	
	\draw (7, -1) node(18)[circle, draw, fill=black!50,
	inner sep=0pt, minimum width=6pt] {};
	
	\draw [line width = 0.5mm, dashed] (1) to (2);
	\draw [line width = 0.5mm, pink] (1) to (3);
	\draw [line width = 0.5mm, magenta] (1) to (4);
	\draw [line width = 0.5mm, purple] (1) to (5);
	\draw [line width = 0.5mm, cyan] (5) to node[below] {$\alpha$} (6);
	\draw [line width = 0.5mm, purple] (6) to (7);
	\draw [line width = 0.5mm, cyan] (7) to node[below] {$\alpha$} (8);
	\draw [line width = 0.5mm, purple] (8) to[out=-110, in=110] (18);
	\node at (6.5, 0) {$e$};
	\node at (7.5, 0) {$e'$};
	
	\draw [line width = 0.5mm, dashed] (11) to (12);
	\draw [line width = 0.5mm, lime] (11) to (13);
	\draw [line width = 0.5mm, teal] (11) to (14);
	\draw [line width = 0.5mm, olive] (11) to (15);
	\draw [line width = 0.5mm, cyan] (15) to node[above] {$\alpha$} (16);
	\draw [line width = 0.5mm, olive] (16) to (17);
	\draw [line width = 0.5mm, cyan] (17) to node[above] {$\alpha$} (18);
	\draw [line width = 0.5mm, olive] (18) to[out=70, in=-70] (8);
	
	\draw [gray!50] plot [smooth cycle] coordinates {(0, 4) (-3, 1)  (-0.5, 0.5) (2, 1)};
	\node at (0, 4.5) {Vizing fan $\F$ around $u_1$};
	
	\draw [gray!50] plot [smooth cycle] coordinates {(0, -4) (-3, -1)  (-0.5, -0.5) (2, -1)};
	\node at (0, -4.5) {Vizing fan $\F'$ around $u_2$};
	
        \draw[->, >={Triangle}, thick, line width = 0.9mm] (4, -4) to (4, -6);
	
	\draw (0, -7) node(21)[circle, draw, fill=black!50,
	inner sep=0pt, minimum width=6pt, label = $u_1$] {};
	
	\draw (-2, -9) node(22)[circle, draw, fill=black!50,
	inner sep=0pt, minimum width=6pt, label = 180:{$v_1$}] {};
	
	\draw (-1, -9) node(23)[circle, draw, fill=black!50,
	inner sep=0pt, minimum width=6pt] {};
	
	\draw (0, -9) node(24)[circle, draw, fill=black!50,
	inner sep=0pt, minimum width=6pt] {};
	
	\draw (1, -9) node(25)[circle, draw, fill=black!50,
	inner sep=0pt, minimum width=6pt] {};
	
	\draw (3, -9) node(26)[circle, draw, fill=black!50,
	inner sep=0pt, minimum width=6pt] {};
	
	\draw (5, -9) node(27)[circle, draw, fill=black!50,
	inner sep=0pt, minimum width=6pt] {};
	
	\draw (7, -9) node(28)[circle, draw, fill=black!50,
	inner sep=0pt, minimum width=6pt] {};
	
	
	\draw (0, -13) node(31)[circle, draw, fill=black!50,
	inner sep=0pt, minimum width=6pt, label = -90:{$u_2$}] {};
	
	\draw (-2, -11) node(32)[circle, draw, fill=black!50,
	inner sep=0pt, minimum width=6pt, label = 180:{$v_2$}] {};
	
	\draw (-1, -11) node(33)[circle, draw, fill=black!50,
	inner sep=0pt, minimum width=6pt] {};
	
	\draw (0, -11) node(34)[circle, draw, fill=black!50,
	inner sep=0pt, minimum width=6pt] {};
	
	\draw (1, -11) node(35)[circle, draw, fill=black!50,
	inner sep=0pt, minimum width=6pt] {};
	
	\draw (3, -11) node(36)[circle, draw, fill=black!50,
	inner sep=0pt, minimum width=6pt] {};
	
	\draw (5, -11) node(37)[circle, draw, fill=black!50,
	inner sep=0pt, minimum width=6pt] {};
	
	\draw (7, -11) node(38)[circle, draw, fill=black!50,
	inner sep=0pt, minimum width=6pt] {};
	
	\draw [line width = 0.5mm, pink] (21) to (22);
	\draw [line width = 0.5mm, magenta] (21) to (23);
	\draw [line width = 0.5mm, purple] (21) to (24);
	\draw [line width = 0.5mm, cyan] (21) to node[right] {$\alpha$} (25);
	\draw [line width = 0.5mm, purple] (25) to (26);
	\draw [line width = 0.5mm, cyan] (26) to node[below] {$\alpha$} (27);
	\draw [line width = 0.5mm, purple] (27) to (28);
	\draw [line width = 0.5mm, dashed] (28) to[out=-70, in=70] (38);
	
	\draw [line width = 0.5mm, lime] (31) to (32);
	\draw [line width = 0.5mm, teal] (31) to (33);
	\draw [line width = 0.5mm, olive] (31) to (34);
	\draw [line width = 0.5mm, cyan] (31) to node[right] {$\alpha$} (35);
	\draw [line width = 0.5mm, olive] (35) to (36);
	\draw [line width = 0.5mm, cyan] (36) to node[above] {$\alpha$} (37);
	\draw [line width = 0.5mm, olive] (37) to (38);
	\draw [line width = 0.5mm, cyan] (38) to[out=110, in=-110] node[right] {$\alpha$} (28);
	
	\node at (6.5, -10) {$e$};
	\node at (7.5, -10) {$e'$};	
\end{tikzpicture}
	\caption{If $P_{\F}$ and $P_{\F'}$ meet in the middle at a pair of parallel edges $e, e'$, then we can shift the uncolored edge from $\F'$ to $e'$ and apply $\Vizing(F)$ which assigns $\chi(e)\leftarrow \alpha$; here $\alpha$ is blue.}\label{multi}
\end{figure}

After modifying \Cref{alg:update path} with this additional case, we strengthen \Cref{invaraint:paths} so that no two edges contained in distinct prefix paths in $\{P_{\F}\}_{\F \in \mathcal F}$ are parallel. This is sufficient to ensure that the analysis of $\ReduceUEdges$ in \Cref{sec:reduce} extends to multigraphs. In particular, all of the lemmas in \Cref{sec:reduce} describing the properties of $\ReduceUEdges$ still hold.

\section{Shannon's Theorem for Multigraphs in Near-Linear Time}\label{sec:shannon}

A classical theorem by Shannon shows that any multigraph $G$ with maximum degree $\Delta$ can be edge colored with $\floor{3\Delta/2}$ colors \cite{shannon1949theorem}. Unlike Vizing's theorem for multigraphs, this bound does not depend on the maximum multiplicity of the graph. We show that we can also extend our approach to a near-linear time randomized algorithm for computing such a coloring. Previously, only 
the original $O(mn)$ time algorithm of~\cite{shannon1949theorem} and a recent $O(n \cdot \Delta^{18})$ time algorithm of~\cite{Dhawan24} were known for this problem. 

\begin{theorem}\label{thm:shannon}
    There is a randomized algorithm that, given an undirected multigraph graph $G = (V, E)$ on $n$ vertices and $m$ edges with maximum degree $\Delta$, finds a $\floor{3\Delta/2}$-edge coloring of $G$ in $O(m\log{n})$ time with high probability.
\end{theorem}

Yet again, we first focus on obtaining a near-linear time algorithm, and postpone optimizing the runtime to $O(m\log{n})$ to the end of the proof. 

\subsection{Shannon Fans}

Let $G = (V, E)$ be an undirected multigraph on $n$ vertices and $m$ edges with maximum degree $\Delta$ and $\chi$ be a partial $\floor{3\Delta/2}$-edge coloring of $G$. 
We begin by defining the notion of a \emph{Shannon fan}, which is conceptually very similar to a u-fan.

\begin{definition}[Shannon fan]
    A \emph{Shannon fan} is a tuple $\s = (u, v, w, \alpha, \beta, \gamma)$ where $u$, $v$ and $w$ are distinct vertices and $\alpha$, $\beta$ and $\gamma$ are distinct colors such that:
    \begin{enumerate}
        \item $(u,v)$ is an uncolored edge and $(u,w)$ is an edge with color $\gamma$.
        \item $\alpha \in \miss_\chi(u)$, $\beta \in \miss_\chi(v) \cap \miss_\chi(w)$ and $\gamma \in \miss_\chi(v)$.
    \end{enumerate}
\end{definition}
\noindent
We say that $u$ is the \emph{center} of $\s$ and that $v$ and $w$ are the \emph{leaves} of $\s$. We also say that the Shannon fan $\s$ is \emph{$\{\alpha, \beta\}$-primed}, and that $\gamma$ is the \emph{auxiliary} color. We say that $\s$ is the Shannon fan of the uncolored edge $(u,v)$.

\medskip
\noindent \textbf{Activating Shannon Fans:}
Let $\s$ be an $\{\alpha,\beta\}$-primed Shannon fan with center $u$ and leaves $v$ and $w$. The key property of Shannon fans is that at most one of the $\{\alpha, \beta\}$-alternating paths starting at $v$ or $w$ ends at $u$. Suppose that the $\{\alpha, \beta\}$-alternating path starting at $v$ does not end at $u$. Then, after flipping this $\{\alpha, \beta\}$-alternating path, both $u, v$ are missing color $\alpha$. Thus, we can extend the coloring $\chi$ by assigning $\chi(u, v)\leftarrow \alpha$. If the path does not end a $w$ instead, we can set $\chi(u,w) \leftarrow \bot$ and $\chi(u,v) \leftarrow \gamma$ and carry out an analogous argument.
We refer to this as \emph{activating} the Shannon fan $\s$.

\medskip
\noindent \textbf{Constructing Shannon Fans:} The following lemma shows how to construct a Shannon fan.

\begin{lemma}\label{lem:create S fan}
    Given an uncolored edge $e$, we can either construct a Shannon fan $\s$ of the edge $e$ or extend $\chi$ to $e$ in $O(\Delta)$ time.
\end{lemma}

\begin{proof}
    Let $u$ and $v$ be the endpoints of $e$ and let $\gamma$ be any color in $\miss_\chi(v)$. If $\gamma \in \miss_\chi(u)$, then we can set $\chi(u,v) \leftarrow \gamma$ and we are done. Otherwise, let $(u,w)$ be the edge incident on $u$ with color $\gamma$ (note that $w \neq v$ since $\gamma \in \miss_\chi(v)$). Now, suppose that $\miss_\chi(u)$ is disjoint from both $\miss_\chi(v)$ and $\miss_\chi(w)$. Otherwise, we can extend $\chi$ to $e$ by shifting the color $\gamma$ between $(u,v)$ and $(u,w)$ and directly coloring the uncolored edge. The following claim shows that we can find a color in $\miss_\chi(v) \cap \miss_\chi(w)$.
    \begin{claim}
        We have that $\miss_\chi(v) \cap \miss_\chi(w) \neq \varnothing$.
    \end{claim}
    \begin{proof}
        Suppose that $\miss_\chi(v)$ and $\miss_\chi(v)$ are disjoint. Then we have that 
        $$2 \cdot \paren{\floor{\frac{3\Delta}{2}} - (\Delta - 1)} + \paren{\floor{\frac{3\Delta}{2}} - \Delta} \geq |\miss_\chi(u)| + |\miss_\chi(v)| + |\miss_\chi(w)| \geq \floor{\frac{3\Delta}{2}}, $$
        which implies that $\floor{3\Delta/2} \geq (3\Delta/2) - 1$, giving a contradiction.
    \end{proof}
    \noindent
    Now, we can find colors $\alpha \in \miss_\chi(u)$ and $\beta \in \miss_\chi(v) \cap \miss_\chi(w)$ in $O(\Delta)$ time. It follows that $\s = (u,v,w,\alpha, \beta, \gamma)$ is a Shannon fan of the edge $e$.
\end{proof}

\medskip
\noindent \textbf{Pre-Shannon Fans:} For notational convenience, we introduce the notion of a \emph{pre-Shannon fan}.

\begin{definition}[pre-Shannon fan]
    A \emph{pre-Shannon fan} is a tuple $\s^\star = (u, v, w, \gamma)$ where $u$, $v$ and $w$ are distinct vertices and $\gamma$ is the auxiliary color such that $(u,v)$ is an uncolored edge, $(u,w)$ is an edge with color $\gamma$ and $\gamma \in \miss_\chi(v)$.
\end{definition}

\noindent
By slightly modifying the proof of \Cref{lem:create S fan}, we get the following lemma which shows that we can either convert a pre-Shannon fan $\s^\star$ into a Shannon fan $\s$ or extend the coloring to the uncolored edge in $\s^\star$.

\begin{lemma}\label{lem:pre S fan}
    Given a pre-Shannon fan $\s^\star = (u,v,w,\gamma)$, we can do one of the following in $O(\Delta)$ time:
    \begin{enumerate}
        \item Find colors $\alpha$ and $\beta$ such that $\s = (u,v,w,\alpha, \beta,\gamma)$ is a Shannon fan.
        \item Extend the coloring to $(u,v)$ while only changing the color of $(u,w)$.
    \end{enumerate}
\end{lemma}

\subsection{Proof of \Cref{thm:shannon}}

Our final algorithm for Shannon's theorem is completely analogous to our final algorithm for simple graphs, which we describe in \Cref{sec:main:algo}. In particular, the algorithm also consists of two main components that are combined in the same way, which we summarize in the following lemmas.

\begin{lemma}[A modification of~\Cref{lem:build u-fans} for Shannon fans]\label{lem:Shannon1}
    There is an algorithm that, given a multigraph $G$, a partial $\floor{3\Delta/2}$-edge coloring $\chi$ and $\lambda = O(m/\Delta)$ uncolored edges that form a matching, does one of the following in $O(m)$ time deterministically:
    \begin{enumerate}
        \item  Extends the coloring to $\Omega(\lambda)$ uncolored edges.
        \item  Returns a set of $\Omega(\lambda)$ vertex-disjoint Shannon fans of the uncolored edges that all share the same auxiliary color.
    \end{enumerate}
\end{lemma}

\noindent We defer the proof of this lemma to \Cref{sec:proof lem Sh1}.

\begin{lemma}[A modification of~\Cref{lem:color u-fans} for Shannon fans]\label{lem:Shannon2}
    There is an algorithm that, given a multigraph $G$, a partial $\floor{3\Delta/2}$-edge coloring $\chi$ and set of $\lambda$ vertex-disjoint Shannon fans $\mathcal U$ sharing the same auxiliary color, extends $\chi$ to $\Omega(\lambda)$ edges in $O(m \log n)$ time with high probability.
\end{lemma}

\noindent The proof of this lemma is almost identical to the proof of \Cref{lem:color u-fans}, thus we omit the details. In particular, the set of vertex-disjoint Shannon fans $\U$ in this lemma has the same key properties as the separable collection of u-fans in \Cref{lem:color u-fans}, so the proof extends directly by simply replacing the u-fans with Shannon fans.
We remark that the additional condition that the Shannon fans must share the same auxiliary color $\gamma$ is crucial. In particular, it ensures that we can prime these Shannon fans by flipping alternating paths \emph{which will not contain the color $\gamma$}, and thus will not interfere with the auxiliary colors of the other Shannon fans in $\U$.

Using \Cref{lem:Shannon1,lem:Shannon2}, we get the following lemma which allows us to extend an edge coloring $\chi$ to a matching of uncolored edges.

\begin{lemma}[A modification of~\Cref{lem:main extend} for Shannon fans]\label{lem:shannon extend}
    There is an algorithm that, given a multigraph $G$, a partial $\floor{3\Delta/2}$-edge coloring $\chi$ and a set of $\lambda = O(m/\Delta)$ uncolored edges that form a matching, extends $\chi$ to these uncolored edges in $O(m\log^2 n)$ time with high probability.
\end{lemma}

\noindent The proof of this lemma is almost identical to the proof of \Cref{lem:main extend}, with the observation that repeatedly applying \Cref{lem:Shannon1,lem:Shannon2} does not change the fact that the remaining uncolored edges form a matching.\footnote{Note that this is \emph{not} true when we repeatedly apply \Cref{lem:build u-fans,lem:color u-fans} in our algorithm for simple graphs, since \Cref{lem:build u-fans} changes the locations of the uncolored edges.}

Finally, in the same way as in the proof of \Cref{thm:multi Viz}, we can apply standard Euler partitioning using \Cref{lem:shannon extend} to merge the colorings, applying it separately for each color class that we uncolor. This allows for obtaining 
an $O(m\log^3{n})$ time algorithm. 

The extension to an $O(m\log{n})$ time algorithm is verbatim as in~\Cref{thm:multi Viz} and~\Cref{thm:main} and we only mention the relevant lemmas below without proving them. 

\begin{lemma}[A modification of~\Cref{lem:log(n)-Euler} for Shannon's Theorem]\label{lem:log(n)-Euler-shannon}
     There is an algorithm that, given a multigraph $G$, finds a $\floor{3\Delta/2}$-edge coloring of all but (exactly) $m/\Delta$ edges in $O(m\log{\Delta})$ time
     with high probability. 
\end{lemma}

\begin{lemma}[A modification of~\Cref{lem:log(n)-wrap-up} for Shannon's Theorem]\label{lem:log(n)-wrap-up-shannon}
	There is an algorithm that, given a multigraph $G$, a partial $\floor{3\Delta/2}$-edge coloring $\chi$ 
	with $\lambda=m/\Delta$ uncolored edges, extends the coloring to all edges in $O(m\log{n})$ time with high probability. 
\end{lemma}

This concludes the proof of \Cref{thm:shannon}.

\subsection{Proof of \Cref{lem:Shannon1}}\label{sec:proof lem Sh1}

Let $U$ be the set of $\lambda = O(m/\Delta)$ uncolored edges forming a matching. We first scan through the edges $(u,v) \in U$ and check if $\miss_\chi(u) \cap \miss_\chi(v) \neq \varnothing$ in $O(\Delta) \cdot \lambda = O(m)$ time. If this holds for at least $\lambda/2$ of these edges, we can extend the coloring to these edges in $O(m)$ time and we are done. Otherwise, let $U' \subseteq U$ be the subset of uncolored edges where this does not hold.

Now, for each color $c$, we compute 
$$\texttt{freq}(c) := \sum_{(u, v) \in U'} \1[c \in \miss_{\chi}(v)]$$ 
by scanning over each $(u, v) \in U'$, taking $O(\lambda \Delta)$ time in total. We now define an auxiliary color 
$$\gamma := \arg \max_{c} \texttt{freq}(c)$$
and proceed to construct a vertex-disjoint collection of pre-Shannon fans $\U'$ with auxiliary color $\gamma$ as described in \Cref{alg:shannon}.


\begin{algorithm}[H]
    \SetAlgoLined
    \DontPrintSemicolon
    \SetKwRepeat{Do}{do}{while}
    \SetKwBlock{Loop}{repeat}{EndLoop}
    $\U' \leftarrow \varnothing$ and $S \leftarrow \varnothing$\;
    \If{$(u,v) \in U'$\label{loop:for:shan}}{
        \If{ $\gamma \in \miss_\chi(v)$}{
            Let $(u,w)$ be the edge with $\chi(u,w) = \gamma$\;
            \If{$\{u,v,w\} \cap S = \varnothing$} {
                $\U' \leftarrow \U' \cup \{(u,v,w,\gamma)\}$\;
                $S \leftarrow S \cup \{u,v,w\}$\;
            }
        }
    }
    \Return{$\U'$}
    \caption{$\textsf{Construct-Pre-Shannon-Fans}(U')$}
    \label{alg:shannon}
\end{algorithm}

\noindent
We can observe that the set $\U'$ produced by \Cref{alg:shannon} consists of vertex-disjoint pre-Shannon fans with auxiliary color $\gamma$, and that this can be computed in $O(\lambda)$ time. The following lemma shows that $\U'$ contains $\Omega(\lambda)$ Shannon fans.

\begin{lemma}
    We have that $|{\U'}| \geq \lambda/24$.
\end{lemma}

\begin{proof}
    We begin by proving the following claim, which shows that $\gamma$ is available at the endpoints of $\Omega(\lambda)$ many edges in $U'$.
    \begin{claim}\label{P1:high freq v}
        We have that $\textnormal{\texttt{freq}}(\gamma) \geq \lambda / 12$.
    \end{claim}

    \begin{proof}
        We can first observe that
        $$\sum_c \texttt{freq}(c) = \sum_c \sum_{(u, v) \in U'} \1[c \in \miss_\chi(v)] = \sum_{(u, v) \in U'} \sum_c \1[c \in \miss_\chi(v)] $$
        $$ \geq \sum_{(u, v) \in U'} \left( \floor{\frac{3\Delta}{2}} - \Delta \right) \geq \left( \frac{\Delta - 1}{2} \right) \cdot |U'| \geq  \frac{\Delta}{4} \cdot |U'| \geq \frac{\lambda \Delta}{8}, $$
        where we are using the facts that each vertex has at least $\floor{3\Delta/2} - \Delta$ missing colors, that $\Delta > 1$ (otherwise, the problem is trivial), and that $|U'| \geq \lambda/2$. Thus, we have that
        $$ \texttt{freq}(\gamma) \geq \frac{1}{\floor{3\Delta/2}} \cdot \sum_c \texttt{freq}(c) \geq \frac{2}{3\Delta} \cdot \frac{\lambda \Delta}{8} = \frac{\lambda}{12}. \qedhere $$
    \end{proof}
    It follows from this claim that we can create pre-Shannon fans with $\gamma$ as an auxiliary color with at least $\lambda/12$ of the uncolored edges in $U'$. Since the uncolored edges in $U'$ are vertex-disjoint, two pre-Shannon fans with auxiliary color $\gamma$ on these edges, $(u_i, v_i, w_i,\gamma)$ and $(u_j, v_j, w_j,\gamma)$, intersect if and only if $u_i = w_j$ and $w_i = u_j$. Thus, any such fan can intersect at most one other such fan. It follows that the algorithm produces a collection $\U'$ of size at least $\lambda/24$.
\end{proof}

Finally, we can take the set $\U'$ of vertex-disjoint pre-Shannon fans and apply \Cref{lem:pre S fan} to each $\s^\star \in \U'$ to compute a set $\U$ of at least $|\mathcal U'|/2 = \Omega(\lambda)$ vertex-disjoint Shannon fans or extend $\chi$ to at least $|\mathcal U'|/2 = \Omega(\lambda)$ edges in $|\mathcal U'| \cdot O(\Delta) \leq \lambda \cdot O(\Delta) \leq O(m)$ time.

\section{Edge Coloring Algorithms for Small Values of $\Delta$}\label{app:small-Delta}

The main contribution and the primary message of our work is presenting the first near-linear time algorithm for $(\Delta+1)$-edge coloring and polynomial time factor improvements over prior work.
However, it turns out that a simple corollary of our main results, plus the prior work of~\cite{BernshteynD23}, is also enough to obtain quite efficient algorithms 
in the regime when $\Delta = n^{o(1)}$, namely, an algorithm with only $O(m\log{\Delta})$ runtime with high probability (basically, replacing $\log{n}$-term in our main results with a $\log{\Delta}$-term instead). 
In this regime, previously, the best known bounds where $O(m \Delta^4\log \Delta)$ time randomized algorithm of~\cite{bernshteyn2024linear} for $(\Delta+1)$-edge coloring and $O(m\log{\Delta})$ 
time randomized algorithm of~\cite{Assadi24} for $(\Delta+O(\log{n}))$-edge coloring. 

\begin{corollary}\label{cor:small-Delta}
    There are randomized algorithms that, given any undirected multigraph $G$ on $n$ vertices and $m$ edges with maximum degree $\Delta$ and maximum edge multiplicity $\mu$, in $O(m\log{\Delta})$ time, with high probability, output: 
    \begin{enumerate}
    	\item a $(\Delta+1)$-edge coloring of any simple graph $G$ (Vizing's theorem for simple graphs);
	\item a $(\Delta+\mu)$-edge coloring of any multigraph $G$ (Vizing's theorem for multigraphs); 
	\item a $\floor{3\Delta/2}$-edge coloring of any multigraph $G$ (Shannon's theorem for multigraphs).  
    \end{enumerate}
\end{corollary}
\Cref{cor:small-Delta} improves the runtime dependence of~\Cref{thm:main},~\Cref{thm:multi Viz}, and~\Cref{thm:shannon} from $O(m\log{n})$ time to $O(m\log{\Delta})$ time instead. 
In the following, we prove the first part of this result, namely, for $(\Delta+1)$-edge coloring of simple graphs, and then show how this argument extends to the other two parts as well. 

\subsection{Proof of~\Cref{cor:small-Delta} for Simple Graphs}

Recall that by~\Cref{lem:log(n)-Euler}, we already have an algorithm that, with high probability, in $O(m\log{\Delta})$ time, $(\Delta+1)$ colors all but $m/\Delta$ edges of the graph. 
Thus, our task is now to extend the coloring to these edges as well. Previously, in~\Cref{lem:log(n)-wrap-up}, we did this by coloring a constant fraction of the edges repeatedly (using~\Cref{lem:log(n)-main extend}); since, we need $\Theta(\log{(m/\Delta)})$ attempts to color all edges and each one requires $O(m)$ time, this lead to an $O(m\log{n})$ time algorithm. 

To prove~\Cref{cor:small-Delta}, we split this final task: we will 
again run the algorithm of~\Cref{lem:log(n)-main extend}, but this time to reduce the number of uncolored edges from $m/\Delta$ to $m/\poly(\Delta)$ for some arbitrarily large polynomial in $\Delta$; thus, 
we need $\Theta(\log{((m/\Delta)/(m/\poly(\Delta)))}) = \Theta(\log{\Delta})$ steps of coloring, and hence $O(m\log{\Delta})$ time in total. Then, to handle the very last $m/\poly(\Delta)$ remaining edges, 
we will run the algorithm of~\cite{bernshteyn2024linear}, which extends the coloring to each 
uncolored edge in essentially (but not formally) $\poly(\Delta)$ time per uncolored edge for some polynomial in $\Delta$. Thus, the second part of the algorithm only takes $O(m)$ time now, leaving us with an $O(m\log{\Delta})$ 
time algorithm all in all. We formalize this approach in the following (which specifically requires extra care for obtaining a high probability guarantee). 

\begin{lemma}\label{lem:Delta-first}
	Let $c \geq 1$ be any arbitrary fixed constant and $\lambda_0 := \max\{{m}/{\Delta^{2c}},m\log{n}/n\}$. There is an algorithm that, given a graph $G$, a partial $(\Delta+1)$-edge coloring $\chi$ 
	with $\lambda=m/\Delta$ uncolored edges, extends the coloring to all but $\Theta(\lambda_0)$ uncolored edges in $O(m\log{\Delta})$ time with high probability. 
\end{lemma}
\begin{proof}
	We simply run our~\Cref{lem:log(n)-main extend} with the given parameters $\lambda$ and $\lambda_0$. By~\Cref{lem:log(n)-main extend}, the expected runtime will be 
	some $T_0 = O(m\log{(\lambda/\lambda_0)} + \Delta \lambda) = O(m\log{\Delta})$ since $\lambda_0 \geq m/\Delta^{2c}$ for some absolute constant $c \geq 1$. 
	Moreover, by the same lemma, 
	\[
		\Pr\sparen{\text{runtime of the algorithm} \geq 2000 \cdot T_0} \leq \exp\paren{-10 \lambda_0} \leq \exp\paren{-\Theta(m\log{n})/n} = 1/\poly(n),  
	\]
	since $\lambda_0 \geq m\log{n}/n$. Thus, the algorithm finishes in $O(m\log{\Delta})$ time with high probability. 
\end{proof}

We now present the final part for coloring the last remaining edges. For this step, we use the following result from~\cite{BernshteynD23} (itself based on~\cite{Bernshteyn22,Christiansen23}) that allows for coloring a random uncolored edge in {essentially} (but not exactly) $\poly(\Delta)$ time. We note that a recent work of~\cite{bernshteyn2024linear} presents a further improvement to these bounds (in terms of the exponent of $\Delta$) but as our dependence on $\Delta$ is only logarithmic, this improvement is inconsequential and thus we stick with the simpler work of~\cite{BernshteynD23}.  

\begin{proposition}[\cite{BernshteynD23}]\label{prop:BD23}
	There is an algorithm that, given a graph $G$, a partial $(\Delta+1)$-edge coloring $\chi$, and $\lambda$ uncolored edges, extends the coloring to a randomly chosen uncolored
	edge in expected $\poly(\Delta) \cdot \log{(n/\lambda)}$ time. 
\end{proposition}

\noindent
We use this result to obtain the following lemma for our purpose. 

\begin{lemma}\label{lem:Delta-second}
	Let $c$ be the $\poly(\Delta)$-exponent in~\Cref{prop:BD23}, and $\lambda_0 := \max\{{m}/{\Delta^{2c}},m\log{n}/n\}$. There is an algorithm that, given a graph $G$, a partial $(\Delta+1)$-edge coloring $\chi$ and $\lambda = \Theta(\lambda_0)$ uncolored edges, extends the coloring to all uncolored edges in $O(m\log \Delta)$ time with high probability.
\end{lemma}
\begin{proof}
	Suppose first that the max-term of $\lambda_0$ is the second one, i.e., $\lambda_0 = m\log{n}/n$. This implies that $\Delta \geq (n/\log{n})^{1/2c}$ and since $c$ is a constant, we have $\log(\Delta) = \Theta(\log{n})$ in this case. 
	We can use the standard Vizing's fan and Vizing's chain approach in~\Cref{lem:full vizing proof} to color the remaining edges deterministically in $O(\lambda_0 \cdot n) = O(m\log{n}/n \cdot n) = O(m\log{\Delta})$ time
	in this case.  
	
	Now consider the case when $\lambda_0 = m/\Delta^{2c}$. We run the algorithm of~\Cref{prop:BD23} with the following modification: when coloring a random edge, if the runtime becomes two times more than the expected
	time guarantee of~\Cref{prop:BD23}, we terminate the coloring and repeat from the beginning. Thus, each step takes $O(\Delta^{c} \cdot \log{(n/i)})$ time deterministically and succeeds in extending the coloring by 
	one edge with probability at least half. 
	
	For $i \in [\lambda]$, define $\alpha_i := 4\log{(n/i)}$ and $X_i$ as the random variable denoting $\ell_i \cdot \alpha_i$ where $\ell_i$ is the number of times the algorithm attempts to color the next random edge until it succeeds 
	when there are only $i$ edges left uncolored. We first prove that 
	\begin{align}
		 \sum_{i=1}^{\lambda} \alpha_i = o(m/\Delta^{c}), \label{eq:Delta-expect}
	\end{align}
	and then show that 
	\begin{align}
	\Pr\!\sparen{\sum_{i=1}^{\lambda}X_i \geq \frac{2m}{\Delta^{c}}} \ll 1/\poly(m), \label{eq:Delta-concentration}
	\end{align}
	which concludes the proof in this case as well since the runtime of coloring the remaining edges is $O(\Delta^{c_0} \cdot \sum_{i=1}^{\lambda}X_i)$ as stated earlier. 
	
\medskip
\noindent \emph{Proof of~\Cref{eq:Delta-expect}.} We have, 
\begin{align*}
	\sum_{i=1}^{\lambda}\alpha_i &= \sum_{i=1}^{\lambda} 4\log{ \!\left(\frac{n}{i}\right)} = 4\log\!\paren{\prod_{i=1}^{\lambda} \frac{n}{i}} = 4\log\!\paren{\frac{n^\lambda}{\lambda!}} \leq \log\!\paren{\left(\frac{e \cdot n}{\lambda} \right)^{\lambda}} \tag{using the inequality $x! \geq (x/e)^x$ for all $x \geq 1$} \\
	&= 4 \cdot \frac{m}{\Delta^{2c}} \cdot \log\paren{\frac{e \cdot n \cdot \Delta^{2c}}{m}} =o(m/\Delta^{c}) \tag{as $\lambda = m/\Delta^{2c}$ and $m \geq n$}.
\end{align*} 

\medskip
\noindent \emph{Proof of~\Cref{eq:Delta-concentration}.} Note that for any $i \in [\lambda]$ and $\delta > 0$, 
\[
	\Pr\sparen{X_i \geq (1+\delta) \cdot \alpha_i} = \Pr\sparen{\ell_i \geq (1+\delta)4} \leq 2^{-(3+4\delta)} \leq \exp\paren{-\delta}, 
\]
where the first inequality is because the probability that $\ell_i$ increases by one each time is at most $1/2$ by Markov's inequality. 

While $X_i$'s are technically not independent (the choice of some edge being colored may change the structure of the graph for coloring next uncolored edge), 
they are stochastically dominated by independent random variables with above tail bounds (given the guarantee of~\Cref{prop:BD23} for \emph{each} uncolored edge). 
Thus, we can apply~\Cref{prop:conc} with parameters $\set{\alpha_i}_{i=1}^{\lambda}$ as above and $\beta_i = 1$ for all $i \in [\lambda]$ to obtain that for $t := m/\Delta^{c}$, 
\begin{align*}
	\Pr\!\sparen{\sum_{i=1}^{\lambda}X_i \geq \frac{2m}{\Delta^{c}}} \leq \Pr\sparen{\sum_{i=1}^{\lambda}X_i \geq \sum_{i=1}^{\lambda}\alpha_i + t} &\leq \exp\paren{-\frac{1}{4\log{n}} \cdot \frac{m}{2\Delta^{c}}} \ll 1/\poly(m),
\end{align*}
where the first two inequalities use~Equation~\eqref{eq:Delta-expect} for $\sum_{i}\alpha_i$ and the third uses $\Delta^{c} \leq (n/\log{n})^{c/2c} < \sqrt{n}$ in this case (when $\lambda_0$ is the first-term of its maximum). 
This concludes the proof. 
\end{proof}

The proof of~\Cref{cor:small-Delta} for $(\Delta+1)$ edge coloring is now as follows: run~\Cref{lem:log(n)-Euler} to color all but $m/\Delta$ edges. Then, let $c$ be the constant in the exponent of $\poly(\Delta)$ in~\Cref{prop:BD23} and run~\Cref{lem:Delta-first} to 
extend the coloring to all but $\lambda_0 := \max\{{m}/{\Delta^{2c}},m\log{n}/n\}$ edges. Finally, run~\Cref{lem:Delta-second} to extend the coloring 
to all these remaining $\lambda_0$ edges. Since each of these algorithms take $O(m\lo{\Delta})$ time and succeed with high probability, we obtain the final algorithm. 

\subsection{Proof of~\Cref{cor:small-Delta} for Multigraphs} 

To extend the proof to multigraphs, recall that by~\Cref{lem:log(n)-Euler-multi} and~\Cref{lem:log(n)-Euler-shannon}, we can, respectively, $(\Delta+\mu)$ color or $\floor{3\Delta/2}$ color all but $m/\Delta$ edges of a multigraph in $O(m\log{\Delta})$ time. 
Thus, we can follow the same exact strategy as in $(\Delta+1)$ coloring approach outlined in the previous subsection, and we only need an analogue of the result of~\cite{BernshteynD23} to conclude the proof. 
Such results are also already established by~\cite{Dhawan24} and we do not need to reinvent the wheel here. 

\begin{proposition}[A subroutine in~{\cite[Theorem 1.6; part 1]{Dhawan24}}]\label{prop:D24-Vizing}
	There is an algorithm that, given a multigraph $G$, a partial $(\Delta+\mu)$-edge coloring $\chi$, and $\lambda$ uncolored edges, extends the coloring to a randomly chosen uncolored
	edge in expected $\poly(\Delta) \cdot \log{(n/\lambda)}$ time. 
\end{proposition}

\begin{proposition}[A subroutine in~{\cite[Theorem 1.5; part 1]{Dhawan24}}]\label{prop:D24-shannon}
	There is an algorithm that, given a multigraph $G$, a partial $\floor{3\Delta/2}$-edge coloring $\chi$, and $\lambda$ uncolored edges, extends the coloring to a randomly chosen uncolored
	edge in expected $\poly(\Delta) \cdot \log{(n/\lambda)}$ time. 
\end{proposition}
\noindent
The rest of the proof is verbatim as in simple graphs in~\Cref{cor:small-Delta} and we omit it here.

\section{A Specialized Concentration Inequality}\label{app:concentration}

Throughout the paper, for optimizing log factors in the runtime of our algorithms, we needed the following specialized concentration inequality. For completeness, we now 
prove this inequality using a standard argument based on moment generating functions. 
\begin{proposition*}[Restatement of~\Cref{prop:conc}]
	Let $\set{X_i}_{i=1}^{n}$ be $n$ independent non-negative random variables associated with parameters $\set{\alpha_i}_{i=1}^{n}$ and $\set{\beta_i}_{i=1}^{n}$ such that for each $i \in [n]$, $\alpha_i,\beta_i \geq 1$, and for every $\delta > 0$, 
	\[
		\Pr\sparen{X_i \geq (1+\delta) \cdot \alpha_i} \leq \exp\paren{- \delta \cdot \beta_i}; 
	\]
	then,  for every $t \geq 0$,
	\[
		\Pr\sparen{\sum_{i=1}^{n} X_i \geq \sum_{i=1}^{n} \alpha_i + t} \leq \exp\paren{-\left(\min_{i=1}^n\frac{\beta_i}{2\alpha_i} \right) \cdot \paren{t-2\sum_{i=1}^n\frac{\alpha_i}{\beta_i}}}.
	\]
\end{proposition*}
\begin{proof}
	 Let $s := 1/2 \cdot \min_i \beta_i/\alpha_i$ and for  $i \in [n]$, $Y_i := \max(X_i - \alpha_i , 0)$. Letting $Y := \sum_{i=1}^{n} Y_i$, leads 
	 \[
	 	\Pr\sparen{\sum_{i=1}^{n} X_i \geq \sum_{i=1}^{n} \alpha_i + t} \leq \Pr\sparen{Y \geq t}, 
	 \]
	 for all $t \geq 0$. We further have, 
	\begin{align}
		\Pr\sparen{Y \geq t} = \Pr\sparen{e^{s \cdot Y} \geq e^{s \cdot t}} \leq e^{-s \cdot t} \cdot \expect{e^{s \cdot Y}} = e^{-s \cdot t} \cdot \expect{\prod_{i=1}^{n} e^{s \cdot Y_i}} = e^{-s \cdot t} \cdot \prod_{i=1}^{n} \expect{e^{s \cdot Y_i}} \label{eq:conc1}
	\end{align}
	where the inequality is by Markov inequality and the final equality is by the independence of the variables. Now, for every $i \in [n]$, 
	\begin{align*}
		\expect{e^{s \cdot Y_i}} &= \int_{y=0}^{\infty} \Pr\sparen{Y_i = y \cdot \alpha_i} \cdot \exp\paren{s \cdot y \cdot \alpha_i}~dy \\
		&= \int_{y=0}^{\infty} \Pr\sparen{X_i = (1+y) \cdot \alpha_i} \cdot \exp\paren{s \cdot y \cdot \alpha_i}~dy \tag{by the connection between $X_i$ and $Y_i$ when $y \geq 0$} \\
		&\leq \int_{y=0}^{\infty} \exp\paren{-y \cdot \beta_i} \cdot \exp\paren{s \cdot y \cdot \alpha_i}~dy \tag{by the tail inequality of $X_i$ in the proposition statement} \\
		&= \frac{1}{\beta_i \cdot (1-s \cdot ({\alpha_i}/{\beta_i}))} \tag{since the integral of $e^{-Kx} = -e^{-Kx}/K$ and $s < \beta_i/\alpha_i$ thus the value of this function in $\infty$ is zero} \\
		&\leq \frac{1}{1-s \cdot (\alpha_i/\beta_i)} \tag{as $\beta_i \geq 1$} \\
		&\leq \exp\paren{2s \cdot \alpha_i/\beta_i} \tag{as $s\cdot\alpha_i/\beta_i < 1/2$ and $1-x \geq e^{-2x}$ for $x < 1/2$}. 
	\end{align*}
	Plugging this bounds in~\Cref{eq:conc1}, implies that 
	\begin{align*}
		\Pr\sparen{Y \geq t} \leq \exp\paren{-s \cdot t + \sum_{i=1}^n 2s \cdot \frac{\alpha_i}{\beta_i}}  = \exp\paren{-\left(\min_{i=1}^n\frac{\beta_i}{2\alpha_i}\right) \cdot \paren{t-2\sum_{i=1}^{n}\frac{\alpha_i}{\beta_i}}},
	\end{align*}
	by the choice of $s$, concluding the proof. 
\end{proof}

\end{document}